\documentclass[aps,preprint,showkeys,nofootinbib,prb,onecolumn]{revtex4-1}%
\usepackage{amssymb}
\usepackage{amsfonts}
\usepackage{amsmath}
\usepackage{amsthm}
\usepackage{float}
\usepackage{graphicx}
\usepackage{hyperref}
\usepackage[compact]{titlesec}
\usepackage{caption}%
\providecommand{\U}[1]{\protect\rule{.1in}{.1in}}
\providecommand{\U}[1]{\protect\rule{.1in}{.1in}}
\newtheorem{theorem}{Theorem}

\newtheorem{lemma}[theorem]{Lemma}

\newtheorem{remark}[theorem]{Remark}

\begin{document}
\title{Generalized Nonequilibrium Quantum Transport of Spin and Pseudospins:
Entanglements and Topological Phases}
\author{Felix A. Buot${}^{1, 2, 3}$, Karla B. Rivero${}^{2, 4}$, and Roland E. S.
Otadoy${}^{2}$}
\affiliation{${}^{1}$Center for Simulation and Modeling, College of Science, George Mason
University, Fairfax, VA 22030, U.S.A,}
\affiliation{${}^{2}$Laboratory of Computational Functional Materials, Nanoscience and
Nanotechnology, Department of Physics, University of San Carlos - Talamban
Campus, Talamban, Cebu 6000, Philippines}
\affiliation{${}^{3}$C\&LB Research Institute, Carmen, Cebu 6005, Philippines}
\affiliation{${}^{4}$Physics Department, College of Science and Mathematics, }
\affiliation{Western Mindanao State University, Zamboanga City 7000, Philippines}

\begin{abstract}
General nonequilibrium quantum transport equations are derived for a coupled
system of charge carriers, Dirac spin, isospin (or valley spin), and
pseudospin, such as either one of the band, layer, impurity, and boundary
pseudospins. Limiting cases are obtained for one, two or three different kinds
of spin ocurring in a system. We show that a characteristic integer number
$N_{s}$ determines the formal form of spin quantum transport equations,
irrespective of the type of spins or pseudospins, as well as the maximal
entanglement entropy. The results may shed a new perspective on the mechanism
leading to zero modes and chiral/helical edge states in topological
insulators, integer quantum Hall effect topological insulator (QHE-TI),
quantum spin Hall effect topological insulator (QSHE-TI) and Kondo topological
insulator (Kondo-TI). It also shed new light in the observed competing weak
localization and antilocalization in spin-dependent quantum transport
measurements. In particular, a novel mechanism of localization and
delocalization, as well as the new mechanism leading to the onset of
superconductivity in bilayer systems seems to emerge naturally from torque
entanglements in nonequilibrium quantum transport equations of spin and
pseudospins. Moreover, the general results may serve as a foundation for
engineering approximations of the quantum transport simulations of spintronic
devices based on graphene and other 2-D materials such as the transition metal
dichalcogenides (TMDs), silicene, as well as based on topological materials
exhibiting quantum spin Hall effects. The extension of the formalism to
spincaloritronics and pseudo-spincaloritronics is straightforward.

\end{abstract}

\email{felixa.buot@gmail.com}
\maketitle
\tableofcontents
\endpage{ }

\makeatletter
\let\toc@pre\relax
\let\toc@post\relax
\makeatother

%\preprint{Research Paper}

\section{Introduction}

The recent history of condensed matter physics has shown that the study of
vortices, cyclotron orbits, spinors, Berry connections (in older form as
Peierls phase factor), Berry curvatures, and Chern numbers have ushered the
incarnation of topology in quantum physics \cite{kost, refa}. Indeed, the
importance of several discrete-two degrees of freedom i.e., spinors and their
entanglements has emerge as ubiquitous in the physics of new and low
dimensional materials, such as graphene, TMDs, silicene, and topological
systems \cite{refa, ref1}. For example, Pauli-Dirac spin, valley spin
(isospin), pseudospin due to low-energy electron-hole symmtery at Dirac points
(electron and hole have the same pseudospin), and bilayer pseudospin have
gained importance in consideration of the spin quantum transport of
two-dimenional materials.

The effect of atomic-layer pseudospin degrees of freedom in bilayer graphene
and TMD materials has been vigorously pursued both theoretically and
experimentally due to exotic properties, namely, exciton condensation and a
new mechanism for the onset of superconductivity \cite{ref1,ref2}. Entangled
electron-hole pairs have also been proposed for monolayer graphene
\cite{ref3}. Indeed, in bilayer materials a new mechansim for the onset of
superconductivity has been an intriguing discovery. This is related to the
onset of exciton condensate due to the entanglement of layer pseudospin with
$\nu=1$ filling of the lowest Landau levels (LLL), i.e., lowest Landau orbit
(LLO) in each of the layers,\cite{ref4,ref5} i.e.with intralayer quantum Hall
effect at $\nu=1$ of LL0. In topological insulators, the Anderson type of
localization of metallic edge states due to interaction with magnetic
impurities causing successive spin flips has also been proposed \cite{ref6}
(in view of our recent findings, this kind of edge-states localization maybe
attributed to the entanglement with impurity spin). Spin Kondo effect at the
metallic edge states has also been treated \cite{ref6a}. In other words, the
edge states have been studied and subjected to the usual treatment of $1$-D
conductors, mostly based on the usual Hamiltonian and exchange interaction context.

In this paper, we propose a rounded picture and a new perspective on the
physics of topological insulators, QHE-TI, QSHE-TI, and Kondo-TI based on
quantum nonlocality as a result of the entanglement of torques induced by the
various spin degrees of freedom. This is inferred from our results of the
generalized nonequilibrium quantum transport equations of spin and
pseudospins, and their entanglements. The present proposal implies a new
mechanism for localization and delocalization based on the series of spin and
pseudospin torque entanglements. The sort of localization being referred to
here is typefied by cyclotron-orbit current localized around orbit center due
to either external magnetic field or intrinsic Berry curvature in strongly
spin-orbit coupled materials. On the other hand, the sort of delocalization
referred to is represented by the metallic edge states in topological
insulators, \cite{ref6} as a result of the \textit{quantum nonlocality}
brought about by spin and pseudospin entanglements, the type of delocalization
responsible for the \textit{resonant quantum tunneling transport} phenomena in
resonant tunneling diodes \cite{bj,jb}.

Here, the point of view in all of these is that of torque entanglement in
materials with multi-spin degrees of freedom. Localization and delocalization
of vortices have also found some treatments in the literature that could be
interpreted as due to torque interactions \cite{ref7,ref8,ref9}. In fact,
these sort of localization and \textit{quantum-nonlocality} delocalization
effects also shed light on the experimentally observed and ubiquitous
phenomena in spin-dependent quantum transport physics. These are the so-called
\textit{weak localization} (WL)\cite{WL} and \textit{weak antilocalization,}
(WAL) \cite{WAL, WALa} often referred to in the literature. By virtue of their
competing effects in nonequilibrium quantum transport, the crossover from WL
to WAL has also been experimentally observed \cite{CWLWAL}.

Our analysis is based on the use of the phase-space formulation of quantum
mechanics.\cite{kafatos, trH} The use of position and momentum variables
should not be construed as employing a semi-classical analysis. The present
analysis is founded on a basic principle, i.e. on the expansion of any quantum
mechanical operator in terms of a complete set of orthogonal basis operators
or\textit{ projectors}.\cite{ref10} In the phase-space formalism, either
discrete \ or continuum, the orthogonal and complete basis operators are the
\textit{phase-space point projectors}, $\hat{\Delta}\left(  p,q\right)  $.
Thus, in \textit{operator Hilbert-space} method\cite{ref10} of quantum theory,
any operator $\hat{A}$ can be expanded as the sum or integral of the basis
operator, $\hat{\Delta}\left(  p,q\right)  $. As an integral expansion we
have,
\[
\hat{A}=\left(  \frac{1}{2\pi\hbar}\right)  ^{d}%
%TCIMACRO{\dint }%
%BeginExpansion
{\displaystyle\int}
%EndExpansion
d^{d}p^{\prime}d^{d}q^{\prime}A\left(  p^{\prime},q^{\prime}\right)
\hat{\Delta}\left(  p^{\prime},q^{\prime}\right)
\]
where $d$ is the dimensionality and the expansion coefficients are the
$A\left(  p,q\right)  $'s, where each coefficient is given by the trace of the
operator with the \textit{phase-space point projectors}, $\hat{\Delta}\left(
p,q\right)  $,%
\[
A\left(  p,q\right)  =Tr\left[  \hat{\Delta}\left(  p,q\right)  \hat
{A}\right]
\]
The $c$-function $A\left(  p,q\right)  $ is known as the \textit{Lattice Weyl
transform} of the operator $\hat{A}$.\cite{trH} Notice that the
\textit{Lattice Weyl transform }operates on matrices. If $\hat{A}=\hat{\rho}$
(the density matrix operator), then $\rho\left(  p,q\right)  \equiv
f_{w}\left(  p,q\right)  $, which is the Wigner distribution function.

It follows that any commutator or anticommutator of $\hat{A}$ and $\hat{B}$
can be treated similarly in the operator Hilbert-space method\cite{ref10},
\[
\left[  \hat{A},\hat{B}\right]  =\left(  \frac{1}{2\pi\hbar}\right)  ^{d}%
%TCIMACRO{\dint }%
%BeginExpansion
{\displaystyle\int}
%EndExpansion
d^{d}p^{\prime}d^{d}q^{\prime}\left[  \hat{A},\hat{B}\right]  \left(
p^{\prime},q^{\prime}\right)  \hat{\Delta}\left(  p^{\prime},q^{\prime
}\right)
\]
where the coefficient of expansion is given by the trace of the operator with
$\hat{\Delta}\left(  p,q\right)  $,%
\[
\left[  \hat{A},\hat{B}\right]  \left(  p,q\right)  =Tr\left(  \hat{\Delta
}\left(  p,q\right)  \left[  \hat{A},\hat{B}\right]  \right)
\]
and
\[
\left\{  \hat{A},\hat{B}\right\}  =\left(  \frac{1}{2\pi\hbar}\right)  ^{d}%
%TCIMACRO{\dint }%
%BeginExpansion
{\displaystyle\int}
%EndExpansion
d^{d}p^{\prime}d^{d}q^{\prime}\left\{  \hat{A},\hat{B}\right\}  \left(
p^{\prime},q^{\prime}\right)  \hat{\Delta}\left(  p^{\prime},q^{\prime
}\right)
\]
where,%
\[
\left\{  \hat{A},\hat{B}\right\}  \left(  p,q\right)  =Tr\left(  \hat{\Delta
}\left(  p,q\right)  \left\{  \hat{A},\hat{B}\right\}  \right)
\]
respectively. Note that no semi-classical approximation is involved or needed
whatsoever in the above phase-space formalism of operator Hilbert-space
method. The exact expression for $\left[  \hat{A},\hat{B}\right]  \left(
p,q\right)  $ and $\left\{  \hat{A},\hat{B}\right\}  \left(  p,q\right)  $ are
given in terms of powers of Poisson braket (equivalent to expansion to all
powers of $\hbar$) or in terms of integral kernels as shown in Sec.\ref{lww}
of the Appendix. It is only when the first power of the Poisson bracket is
used that we have a semi-classical approximation. Indeed, the exact quantum
nonlocality property of the commutator has been rigorously demonstrated in
resonant tunneling diodes, where the theoretical numerical simulation resolves
the controversial experimental quantum transport measurements on these
nanoodevices.\cite{jb} Moreover,the simulation was able to resolve the
discrete energy levels of the quantum wells,\cite{bj} and yields the THz
current oscillations responsible for the current plateau in the dc-current
measurements.\cite{jb} Even in the semi-classical approximation, the quantum
nonlocality property of the commutator is already apparent. However, no
semiclassical approximation is capable of simulating the current-voltage
characteristics of reonant tunneling diodes.\cite{jb}

Thus, in the context of operator Hilbert-space method, the von Neumann
density-matrix evolution equation of quantum statistical dynamics in H-space
becomes a super-statevector quantum dynamical equation in L-space or operator
Hilbert-space. Thus, the familiar von Neumann density-matrix operator equation
in H-space given by
\begin{equation}
i\hbar\frac{\partial}{\partial t}\rho\left(  t\right)  =\left[  \mathcal{H}%
,\rho\right]  \label{rnceq2.1}%
\end{equation}
becomes a super-Schr\"{o}dinger equation for the super-statevector in L-space
expressed as
\begin{equation}
i\hbar\frac{\partial}{\partial t}\left.  \left\vert \rho\left(  t\right)
\right\rangle \right\rangle =\mathcal{L}\left.  \left\vert \rho\left(
t\right)  \right\rangle \right\rangle \text{.} \label{rnceq2.2}%
\end{equation}
In Eq. (\ref{rnceq2.1}), $\rho\left(  t\right)  $ is the density-matrix
operator for the whole many-body system in H-space, whereas in Eq.
(\ref{rnceq2.2}) $\left.  \left\vert \rho\left(  t\right)  \right\rangle
\right\rangle $ is the corresponding super-statevector in L-space. The
superoperator $\mathcal{L}$ corresponds to the commutator $\left[
\mathcal{H},\rho\right]  $ of Eq. (\ref{rnceq2.1}), and is referred to as the
Liouvillian. The corresponding many-body quantum-field operators becomes
quantum superfield operators in L-space. Similar to the zero temperature
Green's function technique, the super-nonequilibrium Green's functions are
calculated in the L-space theory. The details in deriving the nonequilibrium
quantum transport equations including Cooper pairings are given e.g. in one of
the authors book and references therein.\cite{ref10}

On the other hand, without treating the Cooper pairing it is not clear what
information, if any, concerning gapless majorana edge states in $p_{x}+ip_{y}%
$-superconductors, the so-called spin-triplet superconductoror. The reason why
superconductors are the mining field for majoranas is that in general Cooper
pairings, i.e. the presence of anomalous Green's functions, correspond
trivially to paired majorana fermions in non-topological superconductors, but
a nontrivial pairing of $m_{i}=(\Psi_{i}^{\dagger}+\Psi_{i})/2$ and
$m_{j}=(\Psi_{j}^{\dagger}-\Psi_{j})/2i$, with $\left\langle m_{i}%
m_{j}\right\rangle \neq0$ as condensed pair of majorana fermions at
neighboring sites, i.e., $i\neq j$, describes a spinless or ferromagnetic
one-dimensional topological superconductor. This combination is nontrivial
precisely because when $\left\langle m_{i}\right\rangle =0,$ then
$\left\langle m_{j}\right\rangle \neq0$ and \textit{vice versa}, which
indicates the presence of unpaired majorana at site $j$ or $i$ as the case
maybe. This can be induced in a Kitaev wire \cite{kitaev} by proximity effect
with $s$-type superconductor. This calls for use of anomalous Green's function
and ordinary Green's function, say $\frac{1}{4i}\left[  \left\langle \Psi
_{i}^{\dagger}\Psi_{j}^{\dagger}\right\rangle -\left\langle \Psi_{i}\Psi
_{j}\right\rangle -2\left\langle \Psi_{i}^{\dagger}\Psi_{j}\right\rangle
\right]  $, to simulate $\left\langle m_{i}m_{j}\right\rangle $ in our
nonequilibrium quantum transport equations, However, Cooper pairings are all
beyond the scope of the present treatment.

The general outline of this paper is that first we give all the nonequilibrium
quantum transport equations for various number of spin degrees of freedom in a
system. This is followed by a discussion and concluding remarks in Sec.
\ref{discon}. Some of the detailed calculations are relegated to the Appendix.

\section{Quantum Transport Equations}

Our starting point is the general quantum transport expressions for fermions,
where nonequilibrium quantum superfields and their correlations are the basic
variables. These are obtained from the real-time nonequilibrium quantum
superfield theoretical transport formulation of Buot\cite{ref10,
ref11,ref12}:
\begin{align}
i\hbar\left(  \frac{\partial}{\partial t_{1}}+\frac{\partial}{\partial t_{2}%
}\right)  G^{\gtrless}  &  =\left[  \mathcal{H}G^{\gtrless}-G^{\gtrless
}\mathcal{H}\right] \nonumber\\
&  +\left[  \Sigma^{r}G^{\gtrless}-G^{\gtrless}\Sigma^{a}\right]  +\left[
\Sigma^{\gtrless}G^{a}-G^{r}\Sigma^{\gtrless}\right] \nonumber\\
&  +\left[  \Delta_{hh}^{r}g_{ee}^{\gtrless}-g_{hh}^{\gtrless}\Delta_{ee}%
^{a}\right] \nonumber\\
&  +\left[  \Delta_{hh}^{\gtrless}g_{ee}^{a}-g_{hh}^{r}\Delta_{ee}^{\gtrless
}\right]  . \label{grnlesseq}%
\end{align}
The last two brackets account for the Cooper pairings between fermions of the
same specie. These do not concern us in the present paper (their corresponding
transport equations\cite{ref10,ref11} are important in nonequilibrium
superconductivity, where we also need to solve for the nonequilibrium
anomalous Green functions). In what follows, we will drop these last two
square brackets of the RHS of Eq. (\ref{grnlesseq}).

\section{'Cube' Matrix Quantum Transport Equations}

In the absence of Cooper pairing between fermions of the same-specie Eq.
(\ref{grnlesseq}) becomes, by separately writing the discrete spin
quantum-label arguments, essentially as $8\times8$ matrix equations on the
spin and pseudospin indices,%
\begin{align}
&  i\hbar\left(  \frac{\partial}{\partial t_{1}}+\frac{\partial}{\partial
t_{2}}\right)  G_{kk^{\prime}ll^{\prime}mm^{\prime}}^{\lessgtr}\left(
12\right)  =\nonumber\\
&  \left[  \mathcal{H}_{kjl\gamma m\alpha}G_{jk^{\prime}\gamma l^{\prime
}\alpha m^{\prime}}^{\lessgtr}\left(  \bar{2}2\right)  -G_{kjl\gamma m\alpha
}^{\lessgtr}\left(  1\bar{2}\right)  \mathcal{H}_{jk^{\prime}\gamma l^{\prime
}\alpha m^{\prime}}\delta_{\bar{2}2}\right] \nonumber\\
&  +\left[  \Sigma_{kjl\gamma m\alpha}^{r}\left(  1\bar{2}\right)
G_{jk^{\prime}\gamma l^{\prime}\alpha m^{\prime}}^{\lessgtr}\left(  \bar
{2}2\right)  -G_{kjl\gamma m\alpha}^{\lessgtr}\left(  1\bar{2}\right)
\Sigma_{jk^{\prime}\gamma l^{\prime}\alpha m^{\prime}}^{a}\left(  \bar
{2}2\right)  \right] \nonumber\\
&  +\left[  \Sigma_{kjl\gamma m\alpha}^{\lessgtr}\left(  1\bar{2}\right)
G_{jk^{\prime}\gamma l^{\prime}\alpha m^{\prime}}^{a}\left(  \bar{2}2\right)
-G_{kjl\gamma m\alpha}^{r}\left(  1\bar{2}\right)  \Sigma_{jk^{\prime}\gamma
l^{\prime}\alpha m^{\prime}}^{\lessgtr}\left(  \bar{2}2\right)  \right]  ,
\label{multiband1}%
\end{align}
where the Greek subscript indices correspond to discrete degrees of freedom,
namely, Pauli-Dirac spin $\left(  k,k^{\prime}=\left\{  \downarrow
,\uparrow\right\}  \right)  $, valley indices $\left(  l,l^{\prime}=\left\{
K,K^{\prime}\right\}  \right)  $, and layer or band indices as the case maybe
$\left(  e.g.,m,m^{\prime}=\left\{  t,b\right\}  \right)  $. The numeral
indices correspond to the two-point space-time arguments. In what follows, we
will treat either the two-layer model of chiral Dirac fermions $\left(
m,m^{\prime}\right)  $, e.g. the bottom and top layer of graphene and TMDs in
strong magnetic field or other $2$-D materials with strong spin-orbit coupling
but without magnetic fields\cite{ref14b}.

In the absence of Cooper pairing of superconductivity, we may also write the
general transport equation for '$\gtrless$'-quantities as [here $\left\{
A,B\right\}  $ and $\left[  A,B\right]  $ means anticommutator and commutator,
respectively, of operators $A$ and $B$], by dropping all arguments and
discrete indices, as%
\begin{align}
&  i\hbar\left(  \frac{\partial}{\partial t_{1}}+\frac{\partial}{\partial
t_{2}}\right)  G^{\lessgtr}=\nonumber\\
&  \left[  \mathcal{\tilde{H}},G^{\lessgtr}\right]  +\left[  \Sigma^{\lessgtr
},\mathrm{Re}G^{r}\right]  -\frac{i}{2}\left\{  \Gamma,G^{\lessgtr}\right\}
+\frac{i}{2}\left\{  \Sigma^{\lessgtr},A\right\}  , \label{QTE}%
\end{align}
where $\Gamma=-\frac{1}{2}\mathrm{Im}\Sigma^{r}$, is the scattering-out
matrix, $A=-\frac{1}{2}\mathrm{Im}G^{r}$ is the spectral function, and the
single-particle Hamiltonian is given by,
\begin{equation}
\mathcal{\tilde{H}=H}_{o}+\Sigma^{\delta}+\mathrm{Re}\Sigma^{r},
\label{absorbed_in_H}%
\end{equation}
where $\Sigma^{\delta}$ and $\mathrm{Re}\Sigma^{r}$ correspond to the
renormalization of the bands or mass terms induced by the
self-energy\cite{ref10,ref11}. The appearance of the imaginary $i$ in Eq.
(\ref{QTE}) in the anticommutator is commensurate when viewed in their
gradient expansions in terms of Poisson bracket differential operator wherein
a commutator of Eq. (\ref{QTE}) has a factor $i$ whereas the anticommutator
does not have. As is well-known, in the classical limit, the gradient
expansion leads to the Boltzman kinetic transport equation\cite{ref12}.

\section{\label{canonical}Canonical Spinor Form of a 'Cube' Matrix}

Consider $G_{k,k^{\prime}l,l^{\prime}m,m^{\prime}}^{<}$, where the indices
denote the elements of the $8\times8$ matrix. Here $k,k^{\prime}$ denote the
spin indices, $l,l^{\prime}$ denote the isospin or valley indices, and
$m,m^{\prime}$ denote the pseudospin of either band or bilayer indices. We
will now demonstrate how to reduce the $64$ variables of the $8\times8$ matrix
to just $8$ tensor variables, which include pure scalars or total charge. This
is reminiscent to a decomposition, in the absence of dissipation, of the
$SU\left(  8\right)  $ into a direct product of $SU\left(  2\right)  \otimes
SU\left(  2\right)  \otimes SU\left(  2\right)  $, i.e., in order to keep
track of the spin-like degrees of freedom. In what follows, we make constant
use of the theorem of Sec. \ref{product} in the Appendix.

\begin{description}
\item[A] \textit{First Stage:} reduction of Pauli-Dirac spin indices to scalar
and vector or Pauli-Dirac spin components, we have, where $x,y,z$ denote
components of vector quantities,%
\[
G_{k,k^{\prime}l,l^{\prime}m,m^{\prime}}^{<}=\frac{1}{2}\left(
\begin{array}
[c]{cc}%
G_{o,l,l^{\prime}m,m^{\prime}}^{<}+G_{z,l,l^{\prime}m,m^{\prime}}^{<} &
G_{x,l,l^{\prime}m,m^{\prime}}^{<}-iG_{y,l,l^{\prime}m,m^{\prime}}^{<}\\
G_{x,l,l^{\prime}m,m^{\prime}}^{<}+iG_{y,l,l^{\prime}m,m^{\prime}}^{<} &
G_{o,l,l^{\prime}m,m^{\prime}}^{<}-G_{z,l,l^{\prime}m,m^{\prime}}^{<}%
\end{array}
\right)
\]

\item[B] \textit{Second stage}: reduction of the remaining valley indices to
scalar and vector or isospin components.%
\[
\frac{1}{2}G_{o,l,l^{\prime}m,m^{\prime}}^{<}=\frac{1}{4}\left(
\begin{array}
[c]{cc}%
G_{o,o,m,m^{\prime}}^{<}+G_{o,z,m,m^{\prime}}^{<} & G_{o,x,m,m^{\prime}}%
^{<}-iG_{o,y,m,m^{\prime}}^{<}\\
G_{o,x,m,m^{\prime}}^{<}+iG_{o,y,m,m^{\prime}}^{<} & G_{o,o,m,m^{\prime}}%
^{<}-G_{o,z,m,m^{\prime}}^{<}%
\end{array}
\right)
\]%
\[
\frac{1}{2}G_{k,l.l^{\prime},m,m^{\prime}}^{<}=\frac{1}{4}\left(
\begin{array}
[c]{cc}%
G_{k,o,m,m^{\prime}}^{<}+G_{k,z,m,m^{\prime}}^{<} & G_{k,x,m,m^{\prime}}%
^{<}-iG_{k,y,m,m^{\prime}}^{<}\\
G_{k,x,m,m^{\prime}}^{<}+iG_{o,y,m,m^{\prime}}^{<} & G_{k,o,m,m^{\prime}}%
^{<}-G_{o,z,m,m^{\prime}}^{<}%
\end{array}
\right)
\]

\item[C] \textit{Third and final stage:} reduction of the last band or layer
indices to scalar and vector or pseudospin components.
\end{description}

\[
\frac{1}{4}G_{o,o,m,m^{\prime}}^{<}=\left(  \frac{1}{2}\right)  ^{3}\left(
\begin{array}
[c]{cc}%
G_{o,o,o}^{<}+G_{o,o,z}^{<} & G_{o,o,x}^{<}-iG_{o,o,y}^{<}\\
G_{o,o,x}^{<}+iG_{o,o,y}^{<} & G_{o,o,o}^{<}-G_{o,o,z}^{<}%
\end{array}
\right)
\]%
\[
\frac{1}{4}G_{o,l,m,m^{\prime}}^{<}=\left(  \frac{1}{2}\right)  ^{3}\left(
\begin{array}
[c]{cc}%
G_{o,l,o}^{<}+G_{o,l,z}^{<} & G_{o,l,x}^{<}-iG_{o,l,y}^{<}\\
G_{o,l,x}^{<}+iG_{o,l,y}^{<} & G_{o,l,o}^{<}-G_{o,l,z}^{<}%
\end{array}
\right)
\]%
\[
\frac{1}{4}G_{k,o,m,m^{\prime}}^{<}=\left(  \frac{1}{2}\right)  ^{3}\left(
\begin{array}
[c]{cc}%
G_{k,o,o}^{<}+G_{k,o,z}^{<} & G_{k,o,x}^{<}-iG_{k,o,y}^{<}\\
G_{k,o,x}^{<}+iG_{k,o,y}^{<} & G_{k,o,o}^{<}-G_{k,o,z}^{<}%
\end{array}
\right)
\]%
\[
\frac{1}{4}G_{k,l^{\prime},m,m^{\prime}}^{<}=\left(  \frac{1}{2}\right)
^{3}\left(
\begin{array}
[c]{cc}%
G_{k,l^{\prime},o}^{<}+G_{k,l^{\prime},z}^{<} & G_{k,l^{\prime},x}%
^{<}-iG_{k,l^{\prime},y}^{<}\\
G_{k,l^{\prime},x}^{<}+iG_{k,l^{\prime},y}^{<} & G_{k,l^{\prime},o}%
^{<}-G_{k,l^{\prime},z}^{<}%
\end{array}
\right)
\]
Collecting results of the third and final reduction stage, we finally
transformed $G_{k,k^{\prime}l,l^{\prime}m,m^{\prime}}^{<}$ in to its spinor
form as%
\begin{align}
&  \left(  2\right)  ^{3}G_{k,k^{\prime}l,l^{\prime}m,m^{\prime}}%
^{<}\nonumber\\
^{spinor}  &  \Longrightarrow\left\{
\begin{array}
[c]{c}%
\left(
\begin{array}
[c]{cc}%
G_{o,o,o}^{<}+G_{o,o,z}^{<} & G_{o,o,x}^{<}-iG_{o,o,y}^{<}\\
G_{o,o,x}^{<}+iG_{o,o,y}^{<} & G_{o,o,o}^{<}-G_{o,o,z}^{<}%
\end{array}
\right) \\
+\left(
\begin{array}
[c]{cc}%
G_{o,l,o}^{<}+G_{o,l,z}^{<} & G_{o,l,x}^{<}-iG_{o,l,y}^{<}\\
G_{o,l,x}^{<}+iG_{o,l,y}^{<} & G_{o,l,o}^{<}-G_{o,l,z}^{<}%
\end{array}
\right) \\
+\left(
\begin{array}
[c]{cc}%
G_{k,o,o}^{<}+G_{k,o,z}^{<} & G_{k,o,x}^{<}-iG_{k,o,y}^{<}\\
G_{k,o,x}^{<}+iG_{k,o,y}^{<} & G_{k,o,o}^{<}-G_{k,o,z}^{<}%
\end{array}
\right) \\
+\left(
\begin{array}
[c]{cc}%
G_{k,l^{\prime},o}^{<}+G_{k,l^{\prime},z}^{<} & G_{k,l^{\prime},x}%
^{<}-iG_{k,l^{\prime},y}^{<}\\
G_{k,l^{\prime},x}^{<}+iG_{k,l^{\prime},y}^{<} & G_{k,l^{\prime},o}%
^{<}-G_{k,l^{\prime},z}^{<}%
\end{array}
\right)
\end{array}
\right\}  \label{spinorTansformedG}%
\end{align}

Thus, from Eq. (\ref{spinorTansformedG}), we end up with the eight tensor
transport variables, similar to the number of distinct configurations of three
spin-qubits, which include joint distributions or spin-torque
entanglements,\cite{ref13} defined in the Table 1,

\begin{table}
[h!]

\caption{The Eight Transport Variables} \label{Table}
\begin{tabular}
[c]{|l|l|}\hline
$G_{0, 0, 0}^{<}\text{ }$ & \ \ complete scalar, i.e., particle number density
or total charge density\\\hline
$G_{k, 0, 0}^{<}\text{ }$ & \ \ Pauli-Dirac spin magnetization density\\\hline
$G_{0, l, 0}^{<}\text{ }$ & \ \ isospin magnetization density\\\hline
$G_{0, 0, m}^{<}\text{ }$ & \ \ pseudospin magnetization density\\\hline
$G_{k, l, 0}^{<}\text{ }$ & \ \ entangled or joint Pauli-Dirac spin and
isospin magnetization density\\\hline
$G_{k, 0, m}^{<}\text{ }$ & \ \ entangled or joint Pauli-Dirac spin and
pseudospin magnetization density\\\hline
$G_{0, l, m}^{<}\text{ }$ & \ \ entangled or joint isospin-pseudospin
magnetization density\\\hline
$G_{k, l, m}^{<}\text{ }$ & \ \ entangled or joint Pauli-Dirac spin, isospin,
pseudospin magnetization density\\\hline
\end{tabular}

\end{table}

Note that in Table $1$ the $8$ tensor variables, including the total charge,
are mapped to the configurations of three qubits as exhibited by the subcripts
of the nonequilibrium tensorial Green's function, $G^{<}$.

\section{Quantum Transport Equations for a 'Cube' Matrix $G^{\gtrless}$,
$N_{s}=3$}

By treating all $SU\left(  2\right)  $ indices on equal footing, we evaluate
the canonical form of every matrix involved in all matrix products as a series
of product of $2\times2$ matrices: $\left(  kj\right)  \left(  jk^{\prime
}\right)  ,\left(  l\gamma\right)  \left(  \gamma l^{\prime}\right)  ,\left(
m\alpha\right)  \left(  \alpha m^{\prime}\right)  $, as was done in
Sec.\ref{canonical}. We finally end up with the eight coupled quantum
transport equations for $G_{o,o,o}^{<}$, $G_{k,o,o}^{<}$, $G_{o,l,o}^{<}$,
$G_{o,o,m}^{<}$, $G_{k,l,o}^{<}$, $G_{k,o,m}^{<}$, $G_{o,l,m}^{<}$, and
$G_{k,l,m}^{<}$ defined in Table 1. The technique followed in the derivation
is explained in the Sec. \ref{product} - \ref{tensor} of the Appendix. Note
that in line with the four terms in the right-hand-side (RHS) of Eq.
(\ref{QTE}), the RHS of each of the eight coupled transport equations
correspondingly consist of four groups. We have generalized the Hamiltonian
spinor to account for either of the magnetic field and/or Dresselhaus and/or
Rashba spin-orbit coupling \cite{dres-rash}.

Guided by the results of Sec.\ref{prodcubematrx}, we obtain the nonequilibrium
quantum transport equations of a system with three spin degrees of freedom.
For example, we have for the total scalar,%

\begin{align}
&  2^{N_{s}}i\hbar\left(  \frac{\partial}{\partial t_{1}}+\frac{\partial
}{\partial t_{2}}\right)  G_{o,o,o}^{\gtrless}=\nonumber\\
&  =\left[  \mathcal{\tilde{H}},G^{\lessgtr}\right]  _{o,o,o}+\left[
\Sigma^{\lessgtr},\mathrm{Re}G^{r}\right]  _{o,o,o}-\frac{i}{2}\left\{
\Gamma,G^{\lessgtr}\right\}  _{o,o,o}+\frac{i}{2}\left\{  \Sigma^{\lessgtr
},A\right\}  _{o,o,o} \label{QT1}%
\end{align}
where%
\begin{align}
&  \left[  \mathcal{\tilde{H}},G^{\lessgtr}\right]  _{o,o,o}=\nonumber\\
&  =%
\begin{array}
[c]{c}%
\left[  H_{o,o,o},G_{o,o,o}^{\gtrless}\right]  +\left[  H_{o,o,m}%
,G_{o,o,m}^{\gtrless}\right]  +\left[  H_{o,l,o},G_{o,l,o}^{\gtrless}\right]
+\left[  H_{o,l,m},G_{o,l,m}^{\gtrless}\right] \\
+\left[  H_{k,o,o\ },G_{k,o,o}^{\gtrless}\right]  +\left[  H_{k,o,m}%
,G_{k,o,m}^{\lessgtr}\right]  +\left[  H_{k,l,o},G_{k,l,o}^{\gtrless}\right]
+\left[  H_{k,l,m},G_{k,l,m}^{\gtrless}\right]
\end{array}
\nonumber\\
&  \label{QT1-1}%
\end{align}%
\begin{align}
&  \left[  \Sigma^{\lessgtr},\mathrm{Re}G^{r}\right]  _{o,o,o}=\nonumber\\
&  =%
\begin{array}
[c]{c}%
\left[  \Sigma_{o,o,o}^{\gtrless},\mathrm{Re}G_{o,o,o}^{r}\right]  +\left[
\Sigma_{o,o,m}^{\lessgtr},\mathrm{Re}G_{o,o,m}^{r}\right]  +\left[
\Sigma_{o,l,o}^{\gtrless},\mathrm{Re}G_{o,l,o}^{r}\right]  +\left[
\Sigma_{o,l,m}^{\gtrless},\mathrm{Re}G_{o,l,m}^{r}\right] \\
+\left[  \Sigma_{k,o,o\ }^{\lessgtr},\mathrm{Re}G_{k,o,o}^{r}\right]  +\left[
\Sigma_{k,o,m}^{\gtrless},\mathrm{Re}G_{k,o,m}^{r}\right]  +\left[
\Sigma_{k,l,o}^{\lessgtr},\mathrm{Re}G_{k,l,o}^{r}\right]  +\left[
\Sigma_{k,l,m}^{\lessgtr},\mathrm{Re}G_{k,l,m}^{r}\right]
\end{array}
\nonumber\\
&  \label{QT1-2}%
\end{align}%
\begin{align}
&  \left\{  \Gamma,G^{\lessgtr}\right\}  _{o,o,o}=\nonumber\\
&  =%
\begin{array}
[c]{c}%
\left\{  \Gamma_{o,o,o},G_{o,o,o}^{\gtrless}\right\}  +\left\{  \Gamma
_{o,o,m},G_{o,o,m}^{\gtrless}\right\}  +\left\{  \Gamma_{o,l,o},G_{o,l,o}%
^{\gtrless}\right\}  +\left\{  \Gamma_{o,l,m},G_{o,l,m}^{\gtrless}\right\} \\
+\left\{  \Gamma_{k,o,o\ },G_{k,o,o}^{\gtrless}\right\}  +\left\{
\Gamma_{k,o,m},G_{k,o,m}^{\lessgtr}\right\}  +\left\{  \Gamma_{k,l,o}%
,G_{k,l,o}^{\gtrless}\right\}  +\left\{  \Gamma_{k,l,m},G_{k,l,m}^{\gtrless
}\right\}
\end{array}
\nonumber\\
&  \label{QT1-3}%
\end{align}%
\begin{align}
&  \left\{  \Sigma^{\lessgtr},A\right\}  _{o,o,o}=\nonumber\\
&  =\left\{  \Sigma_{o,o,o}^{\gtrless},A_{o,o,o}\right\}  +\left\{
\Sigma_{o,o,m}^{\lessgtr},A_{o,o,m}\right\}  +\left\{  \Sigma_{o,l,o}%
^{\gtrless},A_{o,l,o}\right\}  +\left\{  \Sigma_{o,l,m}^{\gtrless}%
,A_{o,l,m}\right\} \nonumber\\
&  +\left\{  \Sigma_{k,o,o\ }^{\lessgtr},A_{k,o,o}\right\}  +\left\{
\Sigma_{k,o,m}^{\gtrless},A_{k,o,m}\right\}  +\left\{  \Sigma_{k,l,o}%
^{\lessgtr},A_{k,l,o}\right\}  +\left\{  \Sigma_{k,l,m}^{\lessgtr}%
,A_{k,l,m}\right\}  \label{QT1-4}%
\end{align}
where repeated vector-component indices, corresponding to dot products, are
summed over following the Einstein summation convention, and $N_{s}$ is the
integer number of spin-like degrees of freedom. We have Pauli-Dirac spin,
valley spin (or isospin) and pseudospin degrees of freedom acting on the
systems. Here in Eqs. (\ref{QT1}) - (\ref{QT1-4}), we have $N_{s}=3$. The rest
of the equations for $G_{k,o,o}^{<}$, $G_{o,l,o}^{\lessgtr}$, $G_{o,o,m}%
^{\lessgtr}$, $G_{k,l,o}^{\lessgtr}$, $G_{k,o,m}^{\lessgtr}$, $G_{o,l,m}%
^{\lessgtr}$, and $G_{k,l,m}^{\lessgtr}$ are given in Sec. \ref{qtens3} of the Appendix.

\section{Limiting Case of Two Spin Degrees of Freedom, $N_{s}=2$}

To obtain the nonequilibrium quantum transport equations for $N_{s}=2$ from
the above quantum transport equations for $N_{s}=3$, one simply deletes one of
the three indices. For example, for the case of valley spin and pseudospin
degrees of freedom, we obtain%

\begin{align}
&  2^{N_{s}}i\hbar\left(  \frac{\partial}{\partial t_{1}}+\frac{\partial
}{\partial t_{2}}\right)  G_{o,o}^{\gtrless}=\nonumber\\
=  &  \left[  \mathcal{\tilde{H}},G^{\lessgtr}\right]  _{o,o}+\left[
\Sigma^{\lessgtr},\mathrm{Re}G^{r}\right]  _{o,o}-\frac{i}{2}\left\{
\Gamma,G^{\lessgtr}\right\}  _{o,o}+\frac{i}{2}\left\{  \Sigma^{\lessgtr
},A\right\}  _{o,o} \label{2spin1}%
\end{align}
where,%
\begin{align}
&  \left[  \mathcal{\tilde{H}},G^{\lessgtr}\right]  _{o,o}=\nonumber\\
=  &  \left[  H_{o,o},G_{o,o}^{\gtrless}\right]  +\left[  H_{o,m}%
,G_{o,m}^{\gtrless}\right]  +\left[  H_{l,o},G_{l,o}^{\gtrless}\right]
+\left[  H_{l,m},G_{l,m}^{\gtrless}\right]  \label{2spin1-1}%
\end{align}

\begin{align}
&  \left[  \Sigma^{\lessgtr},\mathrm{Re}G^{r}\right]  _{o,o}=\nonumber\\
&  \left[  \Sigma_{o,o}^{\gtrless},\mathrm{Re}G_{o,o}^{r}\right]  +\left[
\Sigma_{o,m}^{\lessgtr},\mathrm{Re}G_{o,m}^{r}\right]  +\left[  \Sigma
_{l,o}^{\gtrless},\mathrm{Re}G_{l,o}^{r}\right]  +\left[  \Sigma
_{l,m}^{\gtrless},\mathrm{Re}G_{l,m}^{r}\right] \nonumber\\
&  \label{2spin1-2}%
\end{align}%
\begin{align}
&  \left\{  \Gamma,G^{\lessgtr}\right\}  _{o,o}=\nonumber\\
=  &  \left\{  \Gamma_{o,o},G_{o,o}^{\gtrless}\right\}  +\left\{  \Gamma
_{o,m},G_{o,m}^{\gtrless}\right\}  +\left\{  \Gamma_{l,o},G_{l,o}^{\gtrless
}\right\}  +\left\{  \Gamma_{l,m},G_{l,m}^{\gtrless}\right\}  \label{2spin1-3}%
\end{align}%
\begin{align}
&  \left\{  \Sigma^{\lessgtr},A\right\}  _{o,o}=\nonumber\\
&  =\left\{  \Sigma_{o,o}^{\gtrless},A_{o,o}\right\}  +\left\{  \Sigma
_{o,m}^{\lessgtr},A_{o,m}\right\}  +\left\{  \Sigma_{l,o}^{\gtrless}%
,A_{l,o}\right\}  +\left\{  \Sigma_{l,m}^{\gtrless},A_{l,m}\right\}
\label{2spin1-4}%
\end{align}

\begin{align}
&  2^{N_{s}}i\hbar\left(  \frac{\partial}{\partial t_{1}}+\frac{\partial
}{\partial t_{2}}\right)  G_{l,o}^{<}=\nonumber\\
&  =\left[  \mathcal{\tilde{H}},G^{\lessgtr}\right]  _{l,o}+\left[
\Sigma^{\lessgtr},\mathrm{Re}G^{r}\right]  _{l,o}-\frac{i}{2}\left\{
\Gamma,G^{\lessgtr}\right\}  _{l,o}+\frac{i}{2}\left\{  \Sigma^{\lessgtr
},A\right\}  _{l,o} \label{2spin2}%
\end{align}
where,%
\begin{align}
&  \left[  \mathcal{\tilde{H}},G^{\lessgtr}\right]  _{l,o}=\nonumber\\
&  =%
\begin{array}
[c]{c}%
\left[  H_{o,o},G_{l,o}^{\gtrless}\right]  +\left[  H_{o,m},G_{l,m}^{\gtrless
}\right]  +i\epsilon_{ll_{1}l_{2}}\left\{  H_{l_{1},o},G_{l_{2},o}^{\gtrless
}\right\} \\
+\left[  H_{l,o},G_{o,o}^{\lessgtr}\right]  +\left[  H_{l,m},G_{o,m}%
^{\gtrless}\right]  +i\epsilon_{ll_{1}l_{2}}\left\{  H_{l_{1},m},G_{l_{2}%
,m}^{\gtrless}\right\}
\end{array}
\label{2spin2-1}%
\end{align}%
\begin{align}
&  \left[  \Sigma^{\lessgtr},\mathrm{Re}G^{r}\right]  _{l,o}=\nonumber\\
=  &
\begin{array}
[c]{c}%
+\left[  \Sigma_{o,o}^{\gtrless},\mathrm{Re}G_{l,o}^{r}\right]  +\left[
\Sigma_{o,m}^{\gtrless},\mathrm{Re}G_{l,m}^{r}\right]  +i\epsilon_{ll_{1}%
l_{2}}\left\{  \Sigma_{l_{1},o}^{\gtrless},\mathrm{Re}G_{l_{2},o}^{r}\right\}
\\
+\left[  \Sigma_{l,o}^{\gtrless},\mathrm{Re}G_{o,o}^{r}\right]  +\left[
\Sigma_{l,m}^{\gtrless},\mathrm{Re}G_{o,m}^{r}\right]  +i\epsilon_{ll_{1}%
l_{2}}\left\{  \Sigma_{l_{1},m}^{\gtrless},\mathrm{Re}G_{l_{2},m}^{r}\right\}
\end{array}
\label{2spin2-2}%
\end{align}%
\begin{align}
&  \left\{  \Gamma,G^{\lessgtr}\right\}  _{l,o}=\nonumber\\
=  &
\begin{array}
[c]{c}%
\left\{  \Gamma_{o,o},G_{l,o}^{\gtrless}\right\}  +\left\{  \Gamma
_{o,m},G_{l,m}^{\gtrless}\right\}  +i\epsilon_{ll_{1}l_{2}}\left[
\Gamma_{l_{1},o},G_{l_{2},o}^{\gtrless}\right] \\
+\left\{  \Gamma_{l,o},G_{o,o}^{\lessgtr}\right\}  +\left\{  \Gamma
_{l,m},G_{o,m}^{\gtrless}\right\}  +i\epsilon_{ll_{1}l_{2}}\left[
\Gamma_{l_{1},m},G_{l_{2},m}^{\gtrless}\right]
\end{array}
\label{2spin2-3}%
\end{align}%
\begin{align}
&  \left\{  \Sigma^{\lessgtr},A\right\}  _{l,o}=\nonumber\\
=  &
\begin{array}
[c]{c}%
\left\{  \Sigma_{o,o}^{\gtrless},A_{l,o}\right\}  +\left\{  \Sigma
_{o,m}^{\gtrless},A_{l,m}\right\}  +i\epsilon_{ll_{1}l_{2}}\left[
\Sigma_{l_{1},o}^{\gtrless},A_{l_{2},o}\right] \\
+\left\{  \Sigma_{l,o}^{\gtrless},A_{o,o}\right\}  +\left\{  \Sigma
_{l,m}^{\gtrless},A_{o,m}\right\}  +i\epsilon_{ll_{1}l_{2}}\left[
\Sigma_{l_{1},m}^{\gtrless},A_{l_{2},m}\right]
\end{array}
\label{2spin2-4}%
\end{align}

\begin{align}
&  2^{N_{s}}i\hbar\left(  \frac{\partial}{\partial t_{1}}+\frac{\partial
}{\partial t_{2}}\right)  G_{o,m}^{\gtrless}=\nonumber\\
&  =\left[  \mathcal{\tilde{H}},G^{\lessgtr}\right]  _{o,m}+\left[
\Sigma^{\lessgtr},\mathrm{Re}G^{r}\right]  _{o,m}-\frac{i}{2}\left\{
\Gamma,G^{\lessgtr}\right\}  _{o,m}+\frac{i}{2}\left\{  \Sigma^{\lessgtr
},A\right\}  _{o,m} \label{2spin3}%
\end{align}
where,%
\begin{align}
&  \left[  \mathcal{\tilde{H}},G^{\lessgtr}\right]  _{o,m}=\nonumber\\
&  =%
\begin{array}
[c]{c}%
\left[  H_{o,o},G_{o,m}^{\gtrless}\right]  +\left[  H_{o,m},G_{o,o}^{\gtrless
}\right]  +i\epsilon_{mm_{1}m_{2}}\left\{  H_{o,m_{1}},G_{o,m_{2}}^{\gtrless
}\right\} \\
+\left[  H_{l,o},G_{l,m}^{\gtrless}\right]  +\left[  H_{l,m},G_{l,o}%
^{\gtrless}\right]  +i\epsilon_{mm_{1}m_{2}}\left\{  H_{l,m_{1}},G_{l,m_{2}%
}^{\gtrless}\right\}
\end{array}
\label{2spin3-1}%
\end{align}%
\begin{align}
&  \left[  \Sigma^{\lessgtr},\mathrm{Re}G^{r}\right]  _{o,m}=\nonumber\\
&
\begin{array}
[c]{c}%
+\left[  \Sigma_{o,o}^{\lessgtr},\mathrm{Re}G_{o,m}^{r}\right]  +\left[
\Sigma_{o,m}^{\lessgtr},\mathrm{Re}G_{o,o}^{r}\right]  +i\epsilon_{mm_{1}%
m_{2}}\left\{  \Sigma_{o,m_{1}}^{\lessgtr},\mathrm{Re}G_{o,m_{2}}^{r}\right\}
\\
+\left[  \Sigma_{l,o}^{\lessgtr},\mathrm{Re}G_{l,m}^{r}\right]  +\left[
\Sigma_{l,m}^{\lessgtr},\mathrm{Re}G_{l,o}^{r}\right]  +i\epsilon_{mm_{1}%
m_{2}}\left\{  \Sigma_{l,m_{1}}^{\lessgtr},\mathrm{Re}G_{l,m_{2}}^{r}\right\}
\end{array}
\nonumber\\
&  \label{2spin3-2}%
\end{align}%
\begin{align}
&  \left\{  \Gamma,G^{\lessgtr}\right\}  _{o,m}=\nonumber\\
&  =\left\{
\begin{array}
[c]{c}%
\left\{  \Gamma_{o,o},G_{o,o,m}^{\gtrless}\right\}  +\left\{  \Gamma
_{o,m},G_{o,o,o}^{\gtrless}\right\}  +i\epsilon_{mm_{1}m_{2}}\left[
\Gamma_{o,m_{1}},G_{o,m_{2}}^{\gtrless}\right] \\
+\left\{  \Gamma_{l,o},G_{,l,m}^{\gtrless}\right\}  +\left\{  \Gamma
_{l,m},G_{l,o}^{\gtrless}\right\}  +i\epsilon_{mm_{1}m_{2}}\left[
\Gamma_{l,m_{1}},G_{l,m_{2}}^{\gtrless}\right]
\end{array}
\right\}  \label{2spin3-3}%
\end{align}%
\begin{align}
&  \left\{  \Sigma^{\lessgtr},A\right\}  _{o,m}=\nonumber\\
=  &  \left\{
\begin{array}
[c]{c}%
\left\{  \Sigma_{o,o}^{\lessgtr},A_{o,o,m}\right\}  +\left\{  \Sigma
_{o,m}^{\lessgtr},A_{o,o}\right\}  +i\epsilon_{mm_{1}m_{2}}\left[
\Sigma_{o,m_{1}}^{\lessgtr},A_{o,m_{2}}\right] \\
+\left\{  \Sigma_{l,o}^{\lessgtr},A_{l,m}\right\}  +\left\{  \Sigma
_{l,m}^{\lessgtr},A_{l,o}\right\}  +i\epsilon_{mm_{1}m_{2}}\left[
\Sigma_{l,m_{1}}^{\lessgtr},A_{l,m_{2}}\right]
\end{array}
\right\}  \label{2spin3-4}%
\end{align}

\begin{align}
&  2^{N_{s}}i\hbar\left(  \frac{\partial}{\partial t_{1}}+\frac{\partial
}{\partial t_{2}}\right)  G_{l,m}^{<}=\nonumber\\
&  =\left[  \mathcal{\tilde{H}},G^{\lessgtr}\right]  _{l,m}+\left[
\Sigma^{\lessgtr},\mathrm{Re}G^{r}\right]  _{l,m}-\frac{i}{2}\left\{
\Gamma,G^{\lessgtr}\right\}  _{l,m}+\frac{i}{2}\left\{  \Sigma^{\lessgtr
},A\right\}  _{l,m} \label{2spin4}%
\end{align}
where,%
\begin{align}
&  \left[  \mathcal{\tilde{H}},G^{\lessgtr}\right]  _{l,m}=\nonumber\\
&  =%
\begin{array}
[c]{c}%
\begin{array}
[c]{c}%
\left[  H_{o,o},G_{l,m}^{\lessgtr}\right]  +\left[  H_{o,m},G_{l,o}^{\lessgtr
}\right]  +\left[  H_{l,o},G_{o,m}^{\lessgtr}\right]  +\left[  H_{l,m}%
,G_{o,o}^{\lessgtr}\right] \\
+i\epsilon_{mm_{1}m_{2}}\left\{  H_{o,m_{1}},G_{l,m_{2}}^{\lessgtr}\right\}
+i\epsilon_{mm_{1}m_{2}}\left\{  H_{l,m_{1}},G_{o,m_{2}}^{\lessgtr}\right\} \\
+i\epsilon_{ll_{1}l_{2}}\left\{  H_{l_{1},o},G_{l_{2},m}^{\lessgtr}\right\}
+i\epsilon_{ll_{1}l_{2}}\left\{  H_{l_{1},m},G_{l_{2},o}^{\lessgtr}\right\}
\end{array}
\\
+\frac{1}{2}i\epsilon_{ll_{1}l_{2}}i\epsilon_{mm_{1}m_{2}}\left[
H_{l_{1},m_{1}},G_{l_{2},m_{2}}^{\lessgtr}\right]  +\frac{1}{2}i\epsilon
_{ll_{1}l_{2}}i\epsilon_{mm_{1}m_{2}}\left[  H_{l_{1},m_{1}},G_{l_{2},m_{2}%
}^{\lessgtr}\right]
\end{array}
\label{2spin4-1}%
\end{align}%
\begin{align}
&  \left[  \Sigma^{\lessgtr},\mathrm{Re}G^{r}\right]  _{l,m}=\nonumber\\
&  =%
\begin{array}
[c]{c}%
\begin{array}
[c]{c}%
\left[  \Sigma_{o,o}^{\lessgtr},\mathrm{Re}G_{l,m}^{r}\right]  +\left[
\Sigma_{o,m}^{\lessgtr},\mathrm{Re}G_{l,o}^{r}\right]  +\left[  \Sigma
_{l,o}^{\lessgtr},\mathrm{Re}G_{o,m}^{r}\right]  +\left[  \Sigma
_{l,m}^{\lessgtr},\mathrm{Re}G_{o,o}^{r}\right] \\
+i\epsilon_{mm_{1}m_{2}}\left\{  \Sigma_{o,m_{1}}^{\lessgtr},\mathrm{Re}%
G_{l,m_{2}}^{r}\right\}  +i\epsilon_{mm_{1}m_{2}}\left\{  \Sigma_{l,m_{1}%
}^{\lessgtr},\mathrm{Re}G_{o,m_{2}}^{r}\right\} \\
+i\epsilon_{ll_{1}l_{2}}\left\{  \Sigma_{l_{1},o}^{\lessgtr},\mathrm{Re}%
G_{l_{2},m}^{r}\right\}  +i\epsilon_{ll_{1}l_{2}}\left\{  \Sigma_{l_{1}%
,m}^{\lessgtr},\mathrm{Re}G_{l_{2},o}^{r}\right\}
\end{array}
\\
+\frac{1}{2}i\epsilon_{ll_{1}l_{2}}i\epsilon_{mm_{1}m_{2}}\left[
\Sigma_{l_{1},m_{1}}^{\lessgtr},\mathrm{Re}G_{l_{2},m_{2}}^{r}\right]
+\frac{1}{2}i\epsilon_{ll_{1}l_{2}}i\epsilon_{mm_{1}m_{2}}\left[
\Sigma_{l_{1},m_{1}}^{\lessgtr},\mathrm{Re}G_{l_{2},m_{2}}^{r}\right]
\end{array}
\label{2spin4-2}%
\end{align}%
\begin{align}
&  \left\{  \Gamma,G^{\lessgtr}\right\}  _{l,m}=\nonumber\\
=  &
\begin{array}
[c]{c}%
\begin{array}
[c]{c}%
\left\{  \Gamma_{o,o},G_{l,m}^{\lessgtr}\right\}  +\left\{  \Gamma
_{o,m},G_{l,o}^{\lessgtr}\right\}  \left\{  \Gamma_{l,o},G_{o,m}^{\lessgtr
}\right\}  +\left\{  \Gamma_{l,m},G_{o,o}^{\lessgtr}\right\} \\
+i\epsilon_{mm_{1}m_{2}}\left[  \Gamma_{o,m_{1}},G_{l,m_{2}}^{\lessgtr
}\right]  +i\epsilon_{mm_{1}m_{2}}\left[  \Gamma_{l,m_{1}},G_{o,m_{2}%
}^{\lessgtr}\right] \\
+i\epsilon_{ll_{1}l_{2}}\left[  \Gamma_{l_{1},o},G_{l_{2},m}^{\lessgtr
}\right]  +i\epsilon_{ll_{1}l_{2}}\left[  \Gamma_{l_{1},m},G_{l_{2}%
,o}^{\lessgtr}\right]
\end{array}
\\
+\frac{1}{2}i\epsilon_{ll_{1}l_{2}}i\epsilon_{mm_{1}m_{2}}\left\{
\Gamma_{l_{1},m_{1}},G_{l_{2},m_{2}}^{\lessgtr}\right\}  +\frac{1}{2}%
i\epsilon_{ll_{1}l_{2}}i\epsilon_{mm_{1}m_{2}}\left\{  \Gamma_{l_{1},m_{1}%
},G_{l_{2},m_{2}}^{\lessgtr}\right\}
\end{array}
\label{2spin4-3}%
\end{align}%
\begin{align}
&  \left\{  \Sigma^{\lessgtr},A\right\}  _{l,m}=\nonumber\\
=  &
\begin{array}
[c]{c}%
\begin{array}
[c]{c}%
\left\{  \Sigma_{o,o}^{\lessgtr},A_{l,m}\right\}  +\left\{  \Sigma
_{o,m}^{\lessgtr},A_{l,o}\right\}  +\left\{  \Sigma_{l,o}^{\lessgtr}%
,A_{o,m}\right\}  +\left\{  \Sigma_{l,m}^{\lessgtr},A_{o,o}\right\} \\
+i\epsilon_{mm_{1}m_{2}}\left\{  \Sigma_{o,m_{1}}^{\lessgtr},A_{l,m_{2}%
}\right\}  +i\epsilon_{mm_{1}m_{2}}\left\{  \Sigma_{l,m_{1}}^{\lessgtr
},A_{o,m_{2}}\right\} \\
+i\epsilon_{ll_{1}l_{2}}\left\{  \Sigma_{l_{1},o}^{\lessgtr},A_{l_{2}%
,m}\right\}  +i\epsilon_{ll_{1}l_{2}}\left\{  \Sigma_{l_{1},m}^{\lessgtr
},A_{l_{2},o}\right\}
\end{array}
\\
+\frac{1}{2}i\epsilon_{ll_{1}l_{2}}i\epsilon_{mm_{1}m_{2}}\left\{
\Sigma_{l_{1},m_{1}}^{\lessgtr},A_{l_{2},m_{2}}\right\}  +\frac{1}{2}%
i\epsilon_{ll_{1}l_{2}}i\epsilon_{mm_{1}m_{2}}\left\{  \Sigma_{l_{1},m_{1}%
}^{\lessgtr},A_{l_{2},m_{2}}\right\}
\end{array}
\label{2spin4-4}%
\end{align}

The limiting case of Pauli-Dirac spin and pseudospin, $N_{s}=2$ has been
treated in the paper by Buot et al\cite{ref14}. There the Dirac spin
semiconductor Bloch equations (DSSBEs) are treated in the electron-hole
picture, although the corresponding DSSBEs for holes was not treated in that
paper. When we cast the DSSBEs in the electron picture, either by virtue of
the symmetry at low energies of the conduction-valence bands of Dirac points,
or on account of layer pseudospin in bilayer systems, the resulting equations
for Pauli-Dirac spin and pseudospin, using the method of using appropriate
combinations of the DSSBEs employed there,\cite{ref14} exactly reproduce the
results as given in the present limiting case. The explicit form of the
expressions containing the product of two Levi Civita symbols can be obtained
by the method of Ref. \cite{ref14} using the DSSBEs derived there but cast in
the electron picture. The typical result containing two Levi Civita symbols,
for example, the $x$-component of pseudospin in our present notation, is%
\begin{align}
&  2^{N_{s}}i\hbar\left(  \frac{\partial}{\partial t_{1}}+\frac{\partial
}{\partial t_{2}}\right)  \vec{G}_{x}^{<}\nonumber\\
&  \Longrightarrow\left[
\begin{array}
[c]{c}%
\frac{i}{2}\left\{  \left(  \vec{\Gamma}_{y}\times\vec{G}_{z}^{<}-\vec{\Gamma
}_{z}\times\vec{G}_{y}^{<}\right)  +\left(  \vec{G}_{y}^{<}\times\vec{\Gamma
}_{z}-\vec{G}_{z}^{<}\times\vec{\Gamma}_{y}\right)  \right\}  \ \\
-\frac{i}{2}\left\{  \left(  \vec{\Sigma}_{y}^{<}\ \times\vec{A}_{z}%
-\vec{\Sigma}_{z}^{<}\times\vec{A}_{y}\right)  +\left(  \vec{A}_{y}\times
\vec{\Sigma}_{z}^{<}-\vec{A}_{z}\times\vec{\Sigma}_{y}^{<}\right)  \right\}
\end{array}
\right]  \label{two_levi1}%
\end{align}
where the vector denotes the Pauli-Dirac spin. This translates to our result
in term of tensor components, for arbitrary vector components for both Dirac
spin and pseudospin, as
\begin{align}
&  2^{N_{s}}i\hbar\left(  \frac{\partial}{\partial t_{1}}+\frac{\partial
}{\partial t_{2}}\right)  G_{k,m}^{<}\nonumber\\
&  \Longrightarrow+\left[
\begin{array}
[c]{c}%
\frac{i}{2}\epsilon_{kk_{1}k_{2}}\epsilon_{mm_{1}m_{2}}\left\{  \Gamma
_{k_{1},m_{1}},G_{k_{2},m_{2}}^{<}\right\}  \ \\
-\frac{i}{2}\epsilon_{kk_{1}k_{2}}\epsilon_{mm_{1}m_{2}}\left\{  \Sigma
_{k_{1},m_{1}},A_{k_{2},m_{2}}\right\}
\end{array}
\right]  \label{two_levi2}%
\end{align}
Equations (\ref{two_levi1}) -(\ref{two_levi2}) give a clear meaning of the
terms containing two levi-Civita symbols in the transport equations. Indeed
for the scattering terms, it restores the locality character as contained in
the original equations, Eq. (\ref{QTE}) [refer to the discussion in Sec.
\ref{locdiff} of the Appendix].

\subsection{Comparison with the expression in Ref \cite{ref14}}

In the light of the present notations used, the corresponding quantum
transport equation for Pauli-Dirac spin and pseudospin with different
notations employed by Buot et al.\cite{ref14} is shown in Sec. \ref{multiband}
of the Appendix to be identical to the present results.

\section{Limiting Case of a Single Spin, $N_{s}=1$}

To obtain the nonequilibrium quantum transport equations for $N_{s}=1$ from
the above quantum transport equations for $N_{s}=3$, one simply deletes two of
the three indices. For example, for the case of pseudospin degrees of freedom,
we obtain,%

\begin{align}
&  2^{N_{s}}i\hbar\left(  \frac{\partial}{\partial t_{1}}+\frac{\partial
}{\partial t_{2}}\right)  G_{o}^{\gtrless}=\nonumber\\
&  \left[  H_{o},G_{o}^{\gtrless}\right]  +\left[  H_{m},G_{m}^{\gtrless
}\right] \nonumber\\
&  +\left[  \Sigma_{o}^{\gtrless},\mathrm{Re}G_{o}^{r}\right]  +\left[
\Sigma_{m}^{\lessgtr},\mathrm{Re}G_{m}^{r}\right] \nonumber\\
&  -\frac{i}{2}\left\{  \left\{  \Gamma_{o},G_{o}^{\gtrless}\right\}
+\left\{  \Gamma_{m},G_{m}^{\gtrless}\right\}  \right\} \nonumber\\
&  +\frac{i}{2}\left\{  \left\{  \Sigma_{o}^{\gtrless},A_{o}\right\}
+\left\{  \Sigma_{m}^{\lessgtr},A_{m}\right\}  \right\}  \label{1spin1}%
\end{align}
where $N_{s}=1$. The pseudospin magnetization density is given by,%

\begin{align}
&  2^{N_{s}}i\hbar\left(  \frac{\partial}{\partial t_{1}}+\frac{\partial
}{\partial t_{2}}\right)  G_{m}^{\gtrless}=\nonumber\\
&  =\left[  H_{o},G_{m}^{\gtrless}\right]  +\left[  H_{m},G_{o}^{\gtrless
}\right]  +i\epsilon_{mm_{1}m_{2}}\left\{  H_{m_{1}},G_{m_{2}}^{\gtrless
}\right\} \nonumber\\
&  +\left[  \Sigma_{o}^{\lessgtr},\mathrm{Re}G_{m}^{r}\right]  +\left[
\Sigma_{m}^{\lessgtr},\mathrm{Re}G_{o}^{r}\right]  +i\epsilon_{mm_{1}m_{2}%
}\left\{  \Sigma_{m_{1}}^{\lessgtr},\mathrm{Re}G_{m_{2}}^{r}\right\}
\nonumber\\
&  -\frac{i}{2}\left\{  \left\{  \Gamma_{o},G_{m}^{\gtrless}\right\}
+\left\{  \Gamma_{m},G_{o}^{\gtrless}\right\}  +i\epsilon_{mm_{1}m_{2}}\left[
\Gamma_{m_{1}},G_{m_{2}}^{\gtrless}\right]  \right\} \nonumber\\
&  +\frac{i}{2}\left\{  \left\{  \Sigma_{o}^{\lessgtr},A_{m}\right\}
+\left\{  \Sigma_{m}^{\lessgtr},A_{o}\right\}  +i\epsilon_{mm_{1}m_{2}}\left[
\Sigma_{m_{1}}^{\lessgtr},A_{m_{2}}\right]  \right\}  \label{1spin2}%
\end{align}

\subsection{Comparison with the expression given in Ref. \cite{ref12}}

The equation we obtain for Dirac spin is identical to that obtained by Buot et
al,\cite{ref12} which is rearranged as follows,%

\begin{align}
2  &  i\hbar\left(  \frac{\partial}{\partial t_{1}}+\frac{\partial}{\partial
t_{2}}\right)  \vec{S}^{<}=\nonumber\\
&  =\left\{
\begin{array}
[c]{c}%
\left[  \bar{H}+\mathrm{Re}\bar{\Sigma}^{r},\vec{S}^{<}\right] \\
+\left[  \left(  \mathcal{\vec{B}+}\mathrm{Re}\vec{\Xi}^{r}\right)  ,S_{o}%
^{<}\right]  +i\left[  \left(  \mathcal{\vec{B}}+\mathrm{Re}\vec{\Xi}%
^{r}\right)  \times\vec{S}^{<}-\vec{S}^{<}\times\left(  \mathcal{\vec{B}%
}+\mathrm{Re}\vec{\Xi}^{r}\right)  \right]
\end{array}
\right\} \nonumber\\
&  +\left[  \bar{\Sigma}^{<},\mathrm{Re}\vec{S}^{r}\right]  +\left[  \vec{\Xi
}^{<},\mathrm{Re}S_{o}^{r}\right]  +i\left[  \vec{\Xi}^{<}\times
\mathrm{Re}\vec{S}^{r}-\ \mathrm{Re}\vec{S}^{r}\times\vec{\Xi}^{<}\right]
\nonumber\\
&  -\frac{i}{2}\left\{  \left\{  \bar{\Gamma},\vec{S}^{<}\right\}  -\frac
{i}{2}\left\{  \vec{\gamma},S_{o}^{<}\right\}  +\frac{i}{2}\left\{
\vec{\gamma}\ \times\vec{S}^{<}+\ \vec{S}^{<}\times\vec{\gamma}\right\}
\right\} \nonumber\\
&  \ +\frac{i}{2}\left\{  \left\{  \bar{\Sigma}^{<},\mathcal{\vec{A}}\right\}
+\frac{i}{2}\left\{  \vec{\Xi}^{<},\mathcal{A}_{o}\right\}  -\frac{i}%
{2}\left\{  \vec{\Xi}^{<}\times\mathcal{\vec{A}}+\mathcal{\vec{A}}\times
\vec{\Xi}^{<}\right\}  \right\}  \label{singleband}%
\end{align}
where in the present paper we used the following corresponding notations as
follows,%
\begin{align*}
\vec{S}^{<}  &  \Longrightarrow G_{k}^{<}\text{, }S_{o}^{<}\Longrightarrow
G_{o}^{<}\text{,\ }\bar{H}\Longrightarrow H_{o}\text{, }\mathcal{\vec{B}%
}+\mathrm{Re}\vec{\Xi}^{r}\Longrightarrow H_{k}\\
\text{\ }\bar{\Gamma}  &  \Longrightarrow\Gamma_{o}\text{, }\vec{\gamma
}\Longrightarrow\Gamma_{k}\text{, }\mathrm{Re}\bar{\Sigma}^{r}\Longrightarrow
\text{included in }H_{o}\text{, }\\
\bar{\Sigma}^{<}  &  \Longrightarrow\Sigma_{o}^{\lessgtr}\text{, \ }\vec{\Xi
}^{<}\Longrightarrow\Sigma_{k}^{\lessgtr}\\
\mathrm{Re}\vec{S}^{r}  &  \Longrightarrow\mathrm{Re}G_{k}^{r}\text{,
\ }\mathrm{Im}S_{o}^{r}\Longrightarrow-\frac{A_{o}}{2}\text{, }\mathcal{\vec
{A}\Longrightarrow}A_{k}\text{, \ }%
\end{align*}
Note that there were typos in Ref. \cite{ref12}, (Eq. (20), where some terms
have missing extra factor $\frac{1}{2}$ and has been corrected in Eq.
(\ref{singleband}).

\section{\label{discon}Discussion and Concluding Remarks}

We have shown from all of the derived coupled spin quantum transport equations
that there is no formal difference between the resulting equations of
different spins for the same characteristic integer, $N_{s}$, which designates
the number of spin degrees of freedom involved in a system. Thus, for the same
$N_{s}$, the transport equations are mathematically identical and exhibit
identical physical behavior, whether Pauli-Dirac spin or any of the
pseudospins. The maximal entanglement entropy of the system \cite{ref13} is
given by $S=N_{s}\ln2$, by simply associating with the bipartite entanglement
of $N_{s}$-qubit systems. For carrying out numerical simulations \cite{bj,jb}
of real highly-nonlinear and switching spintronic devices, the prescription
for converting all equations to phase space is given in Sec. \ref{lww} of the Appendix.

One of the interesting aspects of entanglement in nonequilibrium superfield
spin quantum-transport physics coupled with lattice Weyl transform techniques
\cite{bj} is its relevance to local and nonlocal or diffusive processes in
phase space or drift-diffusion and quantum tunneling transport in position
space. This diffusive propertyof the commutator is exhibited even in the
leading order in gradient expansion. The commutator represents the inherent
nonlocality in quantum mechanics and describes diffusive and tunneling motion
in phase space of transport kinetics, whereas the anticommutator describes
local or nondiffusive events in $\left(  \vec{p},\vec{q}\right)  $ space
typefied either by cyclotron-orbit current tied to orbit center or by decay
and growth of a phase-space distribution function tied to a point $\left(
\vec{p},\vec{q}\right)  $ in phase space. For example, the nondiffusive
character of the scattering anticommutator has been simply approximated by a
relaxation-time approximation in Refs.\cite{bj,jb}, whereas the commutator is
entirely responsible for bringing in the nonlocality of quantum transport
physics, specifically the quantum tunneling transport in resonant tunneling
nanostructures \cite{bj,jb} which for the first time resolves the
controversial aspects of the current-voltage characteristics of resonant
tunneling diodes.\cite{jb}

For convenience (see Appendix Sec. \ref{lww} for more details), we give here
the following lattice-Weyl (LW) transform to phase space of a commutator
$\left[  A,B\right]  $ and an anticommutator $\left\{  A,B\right\}  $ in terms
of Poisson bracket differential operator, $\Lambda$, as expansion to all
powers of $\hbar$,%
\begin{align}
\left[  A,B\right]  \left(  p,q\right)   &  =\cos\Lambda\ \left[  a\left(
p,q\right)  b\left(  p,q\right)  -b\left(  p,q\right)  a\left(  p,q\right)
\right] \nonumber\\
&  -i\sin\Lambda\left\{  a\left(  p,q\right)  b\left(  p,q\right)  +b\left(
p,q\right)  a\left(  p,q\right)  \right\}  ,\label{nonloc}\\
\left\{  A,B\right\}  \left(  p,q\right)   &  =\cos\Lambda\ \left\{  a\left(
p,q\right)  b\left(  p,q\right)  +b\left(  p,q\right)  a\left(  p,q\right)
\right\} \nonumber\\
&  -i\sin\Lambda\ \left[  a\left(  p,q\right)  b\left(  p,q\right)  -b\left(
p,q\right)  a\left(  p,q\right)  \right]  , \label{loc}%
\end{align}
where $\Lambda=\frac{\hbar}{2}\left(  \frac{\partial^{\left(  a\right)  }%
}{\partial p}\cdot\frac{\partial^{\left(  b\right)  }}{\partial q}%
-\frac{\partial^{\left(  a\right)  }}{\partial q}\cdot\frac{\partial^{\left(
b\right)  }}{\partial p}\right)  $, the Poisson bracket operator. Thus if the
LW transforms $a\left(  p,q\right)  $ and $b\left(  p,q\right)  $are not
matrices, then to lowest order, we have%
\begin{align}
\left[  A,B\right]  \left(  p,q\right)   &  =-i\frac{\hbar}{2}\left(
\frac{\partial a\left(  p,q\right)  }{\partial p}\cdot\frac{\partial b\left(
p,q\right)  }{\partial q}-\frac{\partial a\left(  p,q\right)  }{\partial
q}\cdot\frac{\partial b\left(  p,q\right)  }{\partial p}\right)
,\label{commute}\\
\left\{  A,B\right\}  \left(  p,q\right)   &  =\left\{  a\left(  p,q\right)
b\left(  p,q\right)  +b\left(  p,q\right)  a\left(  p,q\right)  \right\}
=2\ a\left(  p,q\right)  b\left(  p,q\right)  . \label{anticommute}%
\end{align}
clearly showing that even to the lowest order the local and diffusive
properties in phase space of anticommutator and commutator, respectively, are
apparent. Likewise, in its full integral representation [Sec. \ref{lww} in the
Appendix] which is suitable for carrying numerical simulations, Eq.
(\ref{nonloc}) has a nonlocal kernel, whereas Eq. (\ref{loc}) has a local
kernel of integration. It is precisely the nonlocal character of the kernel
given in Eq. (\ref{nonlocalK}) of the Appendix that is responsible for quantum
tunneling in resonant tunneling diodes, \cite{bj,jb} whose unique
current-voltage characteristics cannot be explained by any degree of
semi-classical approximations.\cite{fedyabu}

One observes that by virtue of Eqs. (\ref{torque}) - (\ref{2diffusive}) in the
Appendix, the entangling with a second pseudospin can transform local terms in
the transport equation into diffusive or mobile terms, a sort of
\textit{delocalization}. This type of delocalization is not surprising as it
is exhibited even in the classical case, often used to explain the absence of
diamagnetism in classical free electron gas \cite{refa}. In the classical
case, when localized currents or cyclotron orbits interacts with a
constriction or opposite boundaries (providing torque interaction),
delocalized currents in opposite directions reside in opposite boundaries,
respectively. Thus, the quantum edge states in integer quantum Hall effect
(precursor to topological insulators) are often portrayed semiclassically as
counterpropagating \textquotedblleft skipping orbits\textquotedblright\ that
propagate (delocalize) along the boundary of the system in the opposite sense
to that of the Landau orbits, when a particle in a Landau orbit interacts with
the boundary, and specularly bounces off it (see Fig. \ref{skiporbit})
\cite{refa}.

We are lead to conclude that, quantum mechanically speaking, in the specular
reflection at the boundary, the incoming and outgoing states present the
boundary pseudospin resulting in the boundary-pseudospin dependent interaction
of the Landau orbits. This boundary-pseudospin dependent interaction results
in the delocalization of the electron current along the boundary. Since the
interaction is elastic and does not cost energy the delocalized current
excitations do not have mass or are zero-excitation modes resulting in the
chiral/helical dispersion relations or Dirac point of the excitations as in
the case of integer quantum Hall effect. This is the mechanism leading to QHE-TI.

\begin{figure}
[h!]\centering
\includegraphics[width=1.5in]{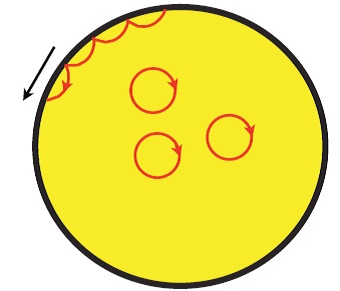}
\caption{Classical skipping orbits in
space, showing the \textit{holographic} nature of topological entanglement,
i.e., \textit{bulk-boundary correspondence}. [Reproduced from Ref. \cite{refa}]. }
\label{skiporbit}
\end{figure}

Therefore, starting from the commutator $\left[  A,B\right]  $ of Eq.
(\ref{QTE}), the dependence in spinor pseudospin degrees of freedom results in
a series of alternating commutator and anti-commutator of the spin tensor
components as the number of entangled spin degrees of freedom, $N_{s}$, grows,
with commutator for even number of $i\epsilon_{ijk}$ and anti-commutator for
odd number of $i\epsilon_{ijk}$. Had we started with $\left\{  A,B\right\}  $,
then there will be a series of alternating anti-commutator and commutator
instead. In other words, entanglement of spins will result in transformations
from local to diffusive motion in phase space, and vice versa. This novel
delocalization mechanism seems to have direct relevance to topological
insulators where torques, due either to spin-orbit coupling (Berry
curvature)\cite{ref14a, ref14b} or magnetic fields resulting in localization
or a gapped spectrum in the bulk, \textit{akin} to integer quantum Hall
effect, are then entangled with pseudospin torque at the boundary consistent
with Dirac points, \textit{akin} to classical specular reflection of cyclotron
orbits at boundaries, resulting in delocalization at the edges (conductive
edge quantum states) of $2$-D systems. We therefore expect that localization
of the edge states can occur if further entanglement with spin-dependent
(spinor) impurities (resulting in spin flips),\cite{ref6} or when layer
pseudospin becomes important in bilayer topological insulators, i.e., when
three spin degrees of freedom become entangled at the edge. There are actually
some evidence on this aspect of edge-state localization in bilayer topological
quantum Hall 2-D systems where edge states are not even mentioned in the
discussion of transport \cite{ref15}. On the other hand, the intralayer gapped
bulk states of the LLO in bilayer systems become mobile due to interaction
with layer pseudospin, in fact in pseudospin ferromagnetic state this becomes
superconducting mediated by the condensation of bilayer bosons of excitons
\cite{ref5}.

To illustrate the point in the above discussion, consider for instance the
first line of Eq. (\ref{2spin2-1}), specifically we are interested in the
term, $i\epsilon_{ll_{1}l_{2}}\left\{  H_{l_{1},o},G_{l_{2},o}^{\gtrless
}\right\}  $, which is the term coming from the Hamiltonian [our spinor
Hamiltonian is generalized to account for the spinor $\mathrm{Re}\Sigma^{r}$
\cite{raja, rajaetal} (which is included in $H$, Eq. (\ref{absorbed_in_H}) as
a spinor mass term) and/or magnetic field and/or Dresselhaus and/or Rashba
spin-orbit coupling] that incorporates the Landau orbit and is thus a local
term by virtue of the anticommutator. However, when the dependence on
pseudospin, $m$, comes into play this local term is transformed into a
nonlocal term in Eq. (\ref{2spin4-1}), specifically through the
term,$i\epsilon_{ll_{1}l_{2}}i\epsilon_{mm_{1}m_{2}}\left[  H_{l_{1},m_{1}%
}G_{l_{2},m_{2}}^{\lessgtr}\right]  $ which is now a nonlocal term, i.e.,
mobile term in terms of transport point of view, as expressed by the
commutator. We believe that the inherent nonlocality property of commutators
involving quantum operators\textit{ }in phase space is fundamentally rooted in
the uncertainty principle of quantum mechanics where the smallest nonlocality
in phase space is measured by the Planck constant $\hbar$. Indeed, the quantum
uncertainty principle naturally follows from the formalism discussed in the
\textbf{Introduction} section.

\subsection{Metallic edge states and bilayer superconductivity}

If the pseudospin $m$ is the boundary pseudospin, then the delocalization will
result in the so-called edge metallic states. Similarly, in bilayer system
with each layer in their LLO localized state, the entangling with the layer
pseudospin will delocalized the LLO states into nonlocal conducting states. In
fact in the layer-pseudospin ferromagnetic state, experiments has shown that
the electron-hole boson condensate at nonzero momentum leads to a new
mechanism for the onset of superconductivity.

\subsection{Localization and delocalization}

Thus, from Eq. (\ref{1spin2}) we expect, in the absence of magnetic field,
localization from Dresselhaus and Rashba spin-orbit coupling as incorporated
in the Hamiltonian, $H$, yielding the \textit{orbit-term} $i\epsilon
_{mm_{1}m_{2}}\left\{  H_{m_{1}},G_{m_{2}}^{\gtrless}\right\}  $, whereas
delocalization contribution is expected from the Dyakonov-Perel and
Elliott-Yafet mechanism as incorporated in the \textit{scattering-out}
$\Gamma$-term, and \textit{scattering-in} $A$-term. These are the commutator
terms in Eq. (\ref{1spin2}) containing one Levi-Civita symbol, which describe
a complex motion similar to the motion and deformation of Landau-orbits,
induced by strong magnetic fields, due to the torque exerted by the electric
field \cite{jefi} resulting in the build-up of Hall voltage.\cite{dp} There
are indeed theoretical and experimental works which support this assertion,
the effect is related to what is often referred to in the literature as weak
localization due to Dresselhaus and Rashba spin-orbit couplings \cite{WL} and
weak antilocalization or supression of scattering rates due to the
Dyakonov-Perel and Elliott-Yafet scattering mechanisms \cite{WAL, WALa} or
entanglement of spin-orbit couplings with another torque (spin) degrees of
freedom. Delocalization or nonlocal scattering terms may also occur where the
scattering centers have a pseudo-spin character, such as the two-level
dependence of the scattering in Wigner-function quantum transport formalism
recently studied by Rossi et al.\cite{rossi, rossi2}

\subsubsection{Crossover from weak localization to weak anti-localization}

In the literature, the term weak localization (WL) and weak antilocalization
(WAL) often specifically refers to the the presence of both spin-orbit
coupling and magnetic field in a $2$-d system. The effects of the magnetic
field is to induce weak localization where carriers execute Wilson loops (i.e.
attributed to $i\epsilon_{kk_{1}k_{2}}\left\{  H_{k_{1},o},G_{k_{2}%
,o}^{\gtrless}\right\}  $ anti-commutator term in transport equation) along
its path, but the effect of spin-orbit coupling is to induce weak
antilocalization (WAL) by trying to modify the Wilson loops into a quantum
diffusive term (i.e., from the $i\epsilon_{kk_{1}k_{2}}i\epsilon_{mm_{1}m_{2}%
}\left[  H_{k_{1},m_{1}}G_{k_{2},m_{2}}^{\lessgtr}\right]  $ commutator term).
These contrasting and competing effects as well as their respective dependence
on magnetic-field strenght for Wislon localization loops, and gated
electric-field strenght\ or thickness of a $2$-d layer system for the
spin-orbit delocalization effects, can result in an observed crossover between
WL and WAL in the magnetoconductivity measurements \cite{CWLWAL} as a function
of gate voltage or magnetic field. The observation of weak antilocalization
implies that a strong spin-orbit coupling is the dominant
effects.\cite{newton}

In exprimental conditions, we can assume that statistically an $\Omega$ number
of particles are executing Wilson loops in their transport, i.e.
$i\epsilon_{kk_{1}k_{2}}\left\{  H_{k_{1},o},G_{k_{2},o}^{\gtrless}\right\}  $
due to the effects of applied magnetic fields, whereas $\left(  N-\Omega
\right)  $ number of particles already felt as well the applied gate-electric
field resulting in a stronger spin-orbit coupling leading to $i\epsilon
_{kk_{1}k_{2}}i\epsilon_{mm_{1}m_{2}}\left[  H_{k_{1},m_{1}}G_{k_{2},m_{2}%
}^{\lessgtr}\right]  $ torque-entanglement modification. At constant gate
voltage, the magnetoconductivity is a decreasing function of the magnetic
field. However by increasing the gate voltage, the spin-orbit coupling start
to become dominant in counteracting the effects of the Wilson
weak-localization loops and therefore a crossover point in the
magnetoconductivity measurent is expcted to occur at $\Omega=0.5N$ after which
the magnetoconductivity will start to increase with the applied gate voltage.
Indeed, the crossoveer point or minimum point of the magnetoconductivity has
been observed in the experiments. The terminology '\textit{weak}' as used in
the experimental observations is due to the fact that in Eqs. (\ref{1spin1}) -
(\ref{1spin2}) or Eqs. (\ref{2spin4-1})-(\ref{2spin4-4}) there are other
important terms that contribute to quantum nonlocality and hence
magnetoconductivity transport measurements, e.g. nonlocality of spin-dependent
scattering terms.\cite{dp}

\begin{remark}
We remark on quantum spin Hall effect topological insulator (QSHE-TI) as it
relates to the above discussion on QHE-TI. In time-reversal symmetric systems
with strong spin-orbit coupling, Kramer's degenerate states play an important
role. Since spin-orbit coupling does not break time reversal symmetry, the
Kramer's degenerate pair remains intact and moves in opposite directions with
opposite spins. Thus, the effective spin-orbit magnetic field rotates the
Kramer's degenerate pair in opposite sense, say clockwise and counterclockwise
rotating pair. When this pair interacts with the boundary pseudospin, in the
same sense as discussed above, this oppositely moving Kramer's Landau-orbit
pair forms metallic edge states of oppositely moving currents having opposite
spins, i.e., helical edge states. This is the mechanism leading to QSHE-TI.
\end{remark}

\begin{remark}
It is worthwhile to add comments on topological Kondo insulator \cite{ki},
which represents the intersection of topology and strongly correlated and
heavy fermion systems. Within the quantum transport perspective, the
hybridization of localized $f$ and mobile $d$-band, which gives rise to a
hybridization gap insulator, is within the domain of quantum transport
localization term expressed by the '\textit{orbit-term'} $i\epsilon
_{kk_{1}k_{2}}\left\{  H_{k_{1}},G_{k_{2}}^{\gtrless}\right\}  $ driven by
strong spin-orbit interaction of the $f$-electrons \cite{ki} as discussed
above. Likewise, the metallic edge states are described by the delocalization
term $i\epsilon_{kk_{1}k_{2}}i\epsilon_{mm_{1}m_{2}}\left[  H_{k_{1},m_{1}%
}G_{k_{2},m_{2}}^{\lessgtr}\right]  $, similar to the mechanism leading to
QSHE-TI with helical edge or surface states. It should be pointed out that
since the orbit centers are the atoms themselves the delocalization at the
surface could results in the restructuring of the atomic surface atoms
resulting in twisting or warping, as exhibited by some Kondo insulator
materials, such as $CeNiSn$ possessing nonsymmorphic symmetries, i.e.,
containing three glide reflections, three screw rotations and an inversion
symmetry \cite{mobius}. There, nonsymmorphic symmetries give rise to a
momentum-dependent twist that enables the surface states to be detached from
the bulk on the glide plane leading to the so-called M\"{o}bius-twisted
surface states \cite{mobius}.
\end{remark}

In summary, it is worth noting that although topological phases are often
discussed in the Hamiltonian context, it has also been shown that associated
topological protection and phenomena can also emerge in open quantum systems
with engineered dissipation \cite{ref16}. In this paper, we have shown that a
rounded picture and intuitive aspects of topological phases and bilayer
superconductivity arise from the entanglement of different pseudospins in
highly nonlinear and nonequilibrium quantum transport equations of spin and
pseudospins. These mechanisms also shed light in the so-called weak
localization (WL) and delocalization (WAL) as observed experimentally
\cite{CWLWAL}. Two entangled torques from two spin degrees of freedom, one in
the bulk extending to the boundary and one at the boundary, are needed to
characterize topological insulators. Moreover in most experiments on WL and
WAL, two entangled torques, one from the applied magnetic field and one from
spin-orbit coupling is also needed to to characterize the results of
experiments on WL versus WAL and their crossover point.

It was mentioned in the \textbf{Introduction}, that a parallel treatment might
be possible for treating majorana edge states in topological superconductors.
Ideally, the zeros of either $\frac{1}{4}\left[  1-\left\langle \Psi
_{i}^{\dagger}\Psi_{j}^{\dagger}\right\rangle +\left\langle \Psi_{i}\Psi
_{j}\right\rangle \right]  $ or $\frac{1}{4}\left[  1+\left\langle \Psi
_{i}^{\dagger}\Psi_{j}^{\dagger}\right\rangle +\left\langle \Psi_{i}\Psi
_{j}\right\rangle \right]  $ would then constitute the probability of finding
only '\textit{one-half} $\left(  \frac{1}{2}\right)  $ fermion' or unpaired
majorana at the edges of $2$-D topological superconductor and at both ends of
Kitaev wire \cite{kitaev}. For the Kitaev wire, the situation bears
similarities with polyacetylene $\left[  C_{2}H_{2}\right]  _{n}$, which is a
prototype of a one-dimensional topological insulator. The entanglement entropy
of a majorana pair is of course determine by this probability, and is given by
$S=\ln2$. In numerical computer simulation, one can look for the big imbalance
between the two quantities to signal the presence of unpaired-majorana, i.e.,
\textit{one-half} fermion density. This would be an interesting topic for
further research, probably as a nonequilibrium quantum superield theory in a
lattice-space framework \cite{ref9a}.

The nonequilibrium quantum transport results of this paper also serve as
springboard to applications in the emerging field of spincaloritronics and
pseudo-spincaloritronics \cite{ref17,ref18}. The formalism employed here can
readily be extended to the inclusion of the spinor form of the
electron-phonon/plasmon and spinon/magnon self-energies, in accounting for the
coupling of the equations obtained here to the spin-dependent nonequilibrium
quantumtransport of phonons and plasmons \cite{ref11}. This will be discussed
in another communication concerning heat flow, spincaloritronics and pseudo-spincaloritronics.

\subsubsection*{Acknowledgment:}

This research is partially supported by the Philippine Council for Industry,
Energy, and Emerging Technology Research and Development of the Department of
Science and Technology (DOST-PCIEERD) of the Philippines through the 'Balik
Scientist Program". Dr. Buot is currently at the Laboratory of Computational
Functional Materials, Nanoscience and Nanotechnology, Department of Physics,
University of San Carlos -Talamban, Cebu City 6000, Philippines. Email: felixa.buot@gmail.com

\bigskip

\appendix

\section{\label{product}Product of $2\times2$ matrices}

\begin{theorem}
Any product of $2\times2$ matrices given by
\[
AB=\left(
\begin{array}
[c]{cc}%
a_{11} & a_{12}\\
a_{21} & a_{22}%
\end{array}
\right)  \left(
\begin{array}
[c]{cc}%
b_{11} & b_{12}\\
b_{21} & b_{22}%
\end{array}
\right)
\]
can be expressed in terms of $2\times2$ identity and Pauli matrices as%
\[
AB=\frac{1}{4}\left\{  S_{o}\left(  ab\right)  I+\vec{V}\left(  ab\right)
\cdot\vec{\sigma}\right\}
\]
where $S_{o}\left(  ab\right)  $ is the \textit{scalar-transforming} operator
on product $AB$ and $\vec{V}\left(  ab\right)  $ is the
\textit{vector-transforming} operator on $AB$. Using the Einstein
\textit{summation} convention for the repeated indices, we have
\begin{align*}
S_{o}\left(  ab\right)   &  =\left[  \bar{a}\bar{b}+a_{k}b_{k}\right] \\
V_{k}\left(  ab\right)   &  =\left[  \bar{a}b_{k}+a_{k}\bar{b}+i\epsilon
_{klm}\left(  a_{l}b_{m}\right)  \right]  .
\end{align*}
Here, the scalar and the vector components of $\vec{a}$ and $\vec{b}$ are
defined by,%
\begin{align*}
\bar{a}  &  =a_{11}+a_{22}\text{, \ \ }\bar{b}=b_{11}+b_{22}\\
a_{x}  &  =a_{12}+a_{21}\text{, \ \ }b_{x}=b_{12}+b_{21}\\
ia_{y}  &  =a_{12}-a_{21}\text{, \ \ }ib_{y}=b_{12}-b_{21}\\
a_{z}  &  =a_{11}-a_{22}\text{, \ \ }b_{z}=b_{11}-b_{22}%
\end{align*}

\end{theorem}

\begin{proof}
Expressing each matrix in terms of the identity and Pauli matrices, we have%
\begin{align*}
AB& =\frac{1}{2}\left(
\begin{array}{cc}
\bar{a}+a_{z} & a_{x}-ia_{y} \\
a_{x}+ia_{y} & \bar{a}-a_{z}%
\end{array}%
\right) \frac{1}{2}\left(
\begin{array}{cc}
\bar{b}+b_{z} & b_{x}-ib_{y} \\
b_{x}+ib_{y} & \bar{b}-b_{z}%
\end{array}%
\right) \\
& =\frac{1}{4}\left(
\begin{array}{cc}
\left( \bar{a}+a_{z}\right) \left( \bar{b}+b_{z}\right) +\left(
a_{x}-ia_{y}\right) \left( b_{x}+ib_{y}\right) & \left( \bar{a}+a_{z}\right)
\left( b_{x}-ib_{y}\right) +\left( a_{x}-ia_{y}\right) \left( \bar{b}%
-b_{z}\right) \\
\left( a_{x}+ia_{y}\right) \left( \bar{b}+b_{z}\right) +\left( \bar{a}%
-a_{z}\right) \left( b_{x}+ib_{y}\right) & \left( a_{x}+ia_{y}\right) \left(
b_{x}-ib_{y}\right) +\left( \bar{a}-a_{z}\right) \left( \bar{b}-b_{z}\right)%
\end{array}%
\right)
\end{align*}%
Upon evaluating the last line, we end up with%
\begin{equation*}
AB=\frac{1}{4}\left\{
\begin{array}{c}
\left[ \bar{a}\bar{b}+\vec{a}\cdot \vec{b}\right] I+\left[ \bar{a}\vec{b}+%
\bar{b}\vec{a}+i\left( \vec{a}\times \vec{b}\right) \right] _{z}\sigma _{z}
\\
+\left[ \left( \bar{a}\vec{b}+\vec{a}\bar{b}\right) +i\left( \vec{a}\times
\vec{b}\right) \right] _{x}\sigma _{x}+\left[ \left( \bar{a}\vec{b}+\vec{a}%
\bar{b}\right) +i\left( \vec{a}\times \vec{b}\right) \right] _{y}\sigma _{y}%
\end{array}%
\right\}
\end{equation*}
\end{proof}

\begin{lemma}
The theorem can be generalized to any binary operation of spinor matrices.
\begin{align*}
S_{o}\left(  a\otimes b\right)   &  =\left[  \bar{a}\otimes\bar{b}+\vec
{a}\otimes\cdot\vec{b}\right] \\
&  =\frac{1}{4}\left[  \bar{a}\otimes\bar{b}+a_{k}\otimes b_{k}\right] \\
\vec{S}\left(  a\otimes b\right)   &  =\frac{1}{4}\left[  \bar{a}\otimes
\vec{b}+\bar{b}\otimes\vec{a}+i\left(  \vec{a}\times\vec{b}\right)  _{\otimes
}\right] \\
S_{k}\left(  ab\right)   &  =\frac{1}{4}\left[  \bar{a}\otimes b_{k}%
+a_{k}\otimes\bar{b}+i\epsilon_{klm}a_{l}\circledast b_{m}\right]  .
\end{align*}
where $\circledast$ will generally be a different binary operation from that
of $\otimes$. This will become clear in the example below.
\end{lemma}

\subsection{Examples}

The Pauli-Dirac spinor form of the commutator of $A$ and $B$ is given by%

\begin{align}
\left[  A,B\right]   &  =\frac{1}{4}\left(  \left[  \bar{a},\vec{b}\right]
+\left[  \vec{a},\bar{b}\right]  +i\left[  \vec{a}\times\vec{b}-\vec{b}%
\times\vec{a}\right]  \right) \nonumber\\
\left[  A,B\right]  _{k}  &  =\frac{1}{4}\left(  \left[  \bar{a},b_{k}\right]
+\left[  a_{k},\bar{b}\right]  +i\epsilon_{kk_{1}k_{2}}\left\{  a_{k_{1}%
},b_{k_{2}}\right\}  \right)  \label{torque}%
\end{align}
where the \textit{bar} indicates total independence with respect to the spin
degrees of freedom, $\left\{  a_{k_{1}},b_{k_{2}}\right\}  $ is the
anticommutation of $a_{k_{1}}$ and $b_{k_{2}}$. Note the change from
commutator to anticommutator for the term containing one Levi-Civita symbol.
To include dependence on another spin degrees of freedom, it is convenient to
arrange the component of the next spin as the succeeding subcripts, such as
dyadic tensors $a_{k,l}$ and $b_{k,l}$ where the index $l$ coresponds to the
next spin variable, such as valley isospin, for example. Instead of using the
bar, quantities independent of two spins will be designated by $a_{o,o}$ and
$b_{o,o}$. In all instances, the theorem on binary product of $2\times2$
matrices can be sucessively applied. Thus, the fully entangled 'torque'
[coming from the last term of Eq. (\ref{torque})] now becomes,%
\begin{equation}
i\epsilon_{ll_{1}l_{2}}i\epsilon_{kk_{1}k_{2}}\left[  a_{k_{1},l_{1}}%
,b_{k_{2},l_{2}}\right]  =-\epsilon_{ll_{1}l_{2}}\epsilon_{kk_{1}k_{2}}\left[
a_{k_{1},l_{1}},b_{k_{2},l_{2}}\right]  \label{2diffusive}%
\end{equation}
Upon adding dependence on one more pseudospin degrees of freedom, such the
pseudospin arising either from the symmetry at low energies of the band
indices at Dirac points in graphene or from layer pseudospin in bilayer
graphene, the entanglement of the torques becomes%
\begin{equation}
-i\epsilon_{mm_{1}m_{2}}\epsilon_{ll_{1}l_{2}}\epsilon_{kk_{1}k_{2}}\left\{
a_{k_{1},l_{1},m_{1}},b_{k_{2},l_{2},m_{2}}\right\}  \label{threepseudo}%
\end{equation}
which is a third rank tensor. Assuming we have considered only the band
pseudospin in Eq. (\ref{threepseudo}), if one further consider the pseudospin
from layer degrees of freedom in bilayer graphene, for example, then the fully
entangled 'torque' becomes%
\[
\epsilon_{nn_{1}n_{2}}\epsilon_{mm_{1}m_{2}}\epsilon_{ll_{1}l_{2}}%
\epsilon_{kk_{1}k_{2}}\left[  a_{k_{1},l_{1},m_{1},n_{1}}b_{k_{2},l_{2}%
,m_{2},n_{2}}\right]
\]
Note that starting from the commutation of $\left[  A,B\right]  $, the
dependence in spin degrees of freedom results in the alternating commutator
and anti-commutator of the spin tensor components in the expression for
entangled torques, with commutator for even number of $i\epsilon_{ijk}$
(\textit{Levi Civita} symbols) and anti-commutator for odd number of
$i\epsilon_{ijk}$. Had we started with $\left\{  A,B\right\}  $, then there
will be alternating anti-commutator and commutator instead as the number of
spin degrees of freedom grows.

These considerations have important significance in formulating the
nonequilibrium spin quantum transport equations. However, for economy of
space, we will only consider the Pauli-Dirac spin, valley spin or isospin, and
pseudospins either due to symmetry of band indices at low energies at the
Dirac points \cite{ref3} or due to bilayer pseudospin. Here we only have eight
coupled transport equations to consider. On the other hand, by including any
other fourth pseudospin degrees of freedom we will double the number of
coupled quantum transport equations to sixteen. Although sixteen coupled spin
quantum transport equations are not treated here, the derivation is
straightforward following the procedures exemplified in this paper. From the
general nonequilibrium coupled spin quantum transport equations various
limiting cases appropriate to the problem at hand can be obained as given in
the text.

\section{\label{locdiff}Local and Nonlocal Terms}

As discussed in the \textbf{Introduction}, the commutator rigorously
incorporates the quantum non-locality in phase space or quantum-diffusive
motion of transport kinetics in phase space [ this is evident even in the
semi-classical approximation leading to Boltzmann transport equation, i.e.,
within the first power of Poisson bracket ] but does exhibit quantum tunneling
across abrupt barriers when all powers of $\hbar$ are taken into account by
using the integral kernel \cite{jb}, whereas the anticommutator describes
local or nondiffusive events in $\left(  \vec{p},\vec{q}\right)  $ space as
typefied by cyclotron-orbit current tied to orbit center or decay and growth
of a phase-space distribution function tied to a point $\left(  \vec{p}%
,\vec{q}\right)  $ in phase space. Moreover, the exact integral form of the
commutator and anticommutator exhibits nonlocal kernel and local kernel of
integration, respectively. See also the discussion in Sec. \ref{lww} in this Appendix.

\section{\label{prodcubematrx}Binary Product of 'Cube Matrices'}

The result of any binary product of matrices characterized by Pauli-Dirac spin
indices, valley or isospin indices, and energy bands or pseudospin indices can
be classified into \emph{eight} tensorial groups, obtained by
\textit{successive} \textit{iteration} of the spinor transformation of
$2\times2$ matrices given in Sec. \ref{product}. We obtained the following
expressions, where as before repeated indices are summed.

\subsection{Total scalars or tensors of rank zero}

Here we have $\left\{  k=0\text{ or }\sum\limits_{k}\text{, \ }l=0\text{ or
}\sum\limits_{l}\text{, \ }m=0\text{ or}\sum\limits_{m}\right\}  $, i.e.,
three $2$-dgrees of freedom, yielding $2^{3}=8$ terms,%

\[
S_{o,o,o}=\left(  \frac{1}{4}\right)  ^{3}\left\{
\begin{array}
[c]{c}%
\left[  \bar{a}_{o,o,o}\bar{b}_{o,o,o}+a_{k,o,o}b_{k,o,o}+\bar{a}_{o,o,m}%
\bar{b}_{o,o,m}\right] \\
+\bar{a}_{o,l,o}\bar{b}_{o,l,o}+a_{k,l,o}b_{k,l,o}+\bar{a}_{o,l,m}\bar
{b}_{o,l,m}\\
+a_{k,o,m}b_{k,o,m}+a_{k,l,m}b_{k,l,m}%
\end{array}
\right\}  \text{ }%
\]

\subsection{Dirac-spin vectors and Torques}

Here we have [$l=0$ or $\sum\limits_{l}$, $m=0$ or $\sum\limits_{m}$] on each
factor yielding $2\times2^{2}=8$ vectors, and two $2$-degrees of freedom for
the torques, $2^{2}=4$ torques%

\begin{equation}
V_{k,o,o}=\left(  \frac{1}{4}\right)  ^{3}\left[
\begin{array}
[c]{c}%
\bar{a}_{o,o,o}b_{k,o,o}+a_{k,o,o}\bar{b}_{o,o,o}+i\epsilon_{kk_{1}k_{2}%
}\left(  a_{k_{1},o,o}b_{k_{2},o,o}\right) \\
+\bar{a}_{o,l,o}b_{k,l,o}+a_{k,l,o}\bar{b}_{o,l,o}+i\epsilon_{kk_{1}k_{2}%
}\left(  a_{k_{1},l,o}b_{k_{2},l,o}\right) \\
+\bar{a}_{o,o,m}b_{k,o,m}+a_{k,o,m}\bar{b}_{o,o,m}+i\epsilon_{kk_{1}k_{2}%
}\left(  a_{k_{1},o,m}b_{k_{2},o,m}\right) \\
+\bar{a}_{o,l,m}b_{k,l,m}+a_{k.l,m}\bar{b}_{o,l,m}+i\epsilon_{kk_{1}k_{2}%
}\left(  a_{k_{1},l,m}b_{k_{2},l,m}\right)
\end{array}
\right]  \text{ }%
\end{equation}

\subsection{Isospin vectors and torques}

Here we have [$k=0$ or $\sum\limits_{k}$, $m=0$ or $\sum\limits_{m}$] on each
factor yielding $2\times2^{2}=8$ vectors, and two $2$-degrees of freedom for
the torques, $2^{2}=4$ torques%

\begin{equation}
V_{o,l,o}=\left(  \frac{1}{4}\right)  ^{3}\left\{
\begin{array}
[c]{c}%
\left[  \bar{a}_{o,o,o}\bar{b}_{o,l,o}+\bar{a}_{o,l,o}\bar{b}_{o,o,o}%
+i\epsilon_{ll_{1}l_{2}}\left(  \bar{a}_{o,l_{1},o}\bar{b}_{o,l_{2},o}\right)
\right] \\
+\bar{a}_{o,o,m}\bar{b}_{o,l,m}+\bar{a}_{o,l,m}\bar{b}_{o,o,m}+i\epsilon
_{ll_{1}l_{2}}\left(  \bar{a}_{o,l_{1},m}\bar{b}_{o,l_{2},m}\right) \\
+\left[  a_{k,o,o}b_{k,l,o}+a_{k,l,o}b_{k,o,o}+i\epsilon_{ll_{1}l_{2}}\left(
a_{k,l_{1},o}b_{k,l_{2},o}\right)  \right] \\
+a_{k,o,m}b_{k,l,m}+a_{k.l,m}b_{k,o,m}+i\epsilon_{ll_{1}l_{2}}\left(
a_{k,l_{1},m}b_{k,l_{2},m}\right)
\end{array}
\right\}  \text{ }%
\end{equation}

\subsection{Pseudospin vectors and torques}

Similarly, here we have [$k=0$ or $\sum\limits_{k}$, $l=0$ or $\sum
\limits_{l}$] on each factor yielding $2\times2^{2}=8$ vectors, and two
$2$-degrees of freedom for the torques, $2^{2}=4$ torques%

\[
V_{o,o,m}=\left(  \frac{1}{4}\right)  ^{3}\left\{
\begin{array}
[c]{c}%
\left[  \bar{a}_{o,o,o}\bar{b}_{o,o,m}+\bar{a}_{o,o,m}\bar{b}_{o,o,o}%
+i\epsilon_{mm_{1}m_{2}}\left(  \bar{a}_{o,o,m_{1}}\bar{b}_{o,o,m_{2}}\right)
\right] \\
+\left[  a_{k,o,o}b_{k,o,m}+a_{k,o,m}b_{k,o,o}+i\epsilon_{mm_{1}m_{2}}\left(
a_{k,o,m_{1}}b_{k,o,m_{2}}\right)  \right] \\
+\left[  \bar{a}_{o,l,o}\bar{b}_{o,l,m}+\bar{a}_{o,l,m}\bar{b}_{o,l,o}%
+i\epsilon_{mm_{1}m_{2}}\left(  \bar{a}_{o,l,m_{1}}\bar{b}_{o,l,m_{2}}\right)
\right] \\
+\left[  a_{k,l,o}b_{k,l,m}+a_{k,l,m}b_{k,l,o}+i\epsilon_{mm_{1}m_{2}}\left(
a_{k,l,m_{1}}b_{k,l,m_{2}}\right)  \right]
\end{array}
\right\}  \text{ }%
\]
\bigskip

\subsection{Dirac spin-isospin dyadics, torques and entangled torques}

Here we have [$m=0$ or $%
%TCIMACRO{\dsum \limits_{m}}%
%BeginExpansion
{\displaystyle\sum\limits_{m}}
%EndExpansion
$, $k,l$ can be distributed in two ways, and $l$ and $k$ can be distributed in
two ways] yielding $2^{3}=8$ dyadics, $2^{3}=8$ torques, and $2^{1}=2$
entangled torques from $m=0$ or $%
%TCIMACRO{\dsum \limits_{m}}%
%BeginExpansion
{\displaystyle\sum\limits_{m}}
%EndExpansion
$ only. We have,%

\[
T_{k,l}=\left(  \frac{1}{4}\right)  ^{3}\left\{
\begin{array}
[c]{c}%
\left[
\begin{array}
[c]{c}%
\bar{a}_{o,o,o}b_{k,l,o}+a_{k,l,o}\bar{b}_{o,o,o}+\bar{a}_{o,l,o}%
b_{k,o,o}+a_{k,o,o}\bar{b}_{o,l,o}\\
+i\epsilon_{kk_{1}k_{2}}\left(  a_{k_{1},o,o}b_{k_{2},l,o}\right)
+i\epsilon_{kk_{1}k_{2}}\left(  a_{k_{1},l,o}b_{k_{2},o,o}\right) \\
+\left[  i\epsilon_{ll_{1}l_{2}}\left(  \bar{a}_{o,l_{1},o}b_{k,l_{2}%
,o}\right)  \right]  +i\epsilon_{ll_{1}l_{2}}\left(  a_{k,l_{1},o}\bar
{b}_{o,l_{2},o}\right) \\
+\frac{1}{2}\left[  i\epsilon_{kk_{1}k_{2}}i\epsilon_{ll_{1}l_{2}}\left(
a_{k_{1},l_{1},o}b_{k_{2},l_{2},o}\right)  \right]  +\frac{1}{2}%
i\epsilon_{kk_{1}k_{2}}i\epsilon_{ll_{1}l_{2}}\left(  a_{k_{1},l_{1}%
,o}b_{k_{2},l_{2},o}\right)
\end{array}
\right] \\
+\left[
\begin{array}
[c]{c}%
+\bar{a}_{o,o,m}b_{k,l,m}+a_{k.l,m}\bar{b}_{o,o,m}+\bar{a}_{o,l,m}%
b_{k,o,m}+a_{k,o,m}\bar{b}_{o,l,m}\\
+i\epsilon_{kk_{1}k_{2}}\left(  a_{k_{1},l,m}b_{k_{2},o,m}\right)
+i\epsilon_{kk_{1}k_{2}}\left(  a_{k_{1},o,m}b_{k_{2},l,m}\right) \\
+i\epsilon_{ll_{1}l_{2}}\left(  a_{k,l_{1},m}\bar{b}_{o,l_{2},m}\right)
+\left[  i\epsilon_{ll_{1}l_{2}}\left(  \bar{a}_{o,l_{1},m}b_{k,l_{2}%
,m}\right)  \right] \\
+\frac{1}{2}\left[  i\epsilon_{kk_{1}k_{2}}i\epsilon_{ll_{1}l_{2}}\left(
a_{k_{1},l_{1},m}b_{k_{2},l_{2},m}\right)  \right]  +\frac{1}{2}\left[
i\epsilon_{kk_{1}k_{2}}i\epsilon_{ll_{1}l_{2}}\left(  a_{k_{1},l_{1}%
,m}b_{k_{2},l_{2},m}\right)  \right]
\end{array}
\right]
\end{array}
\right\}  \text{ }%
\]

\subsection{Dirac spin-pseudospin dyadics, torques and entangled torques}

Here we have [$l=0$, or $%
%TCIMACRO{\dsum \limits_{l}}%
%BeginExpansion
{\displaystyle\sum\limits_{l}}
%EndExpansion
$, $k,m$ can be distributed in two ways, and $m$ and $k$ can be distributed in
two ways] yielding $2^{3}=8$ dyadics, $2^{3}=8$ torques, and $2^{1}=2$
entangled torques from $l=0$, or $%
%TCIMACRO{\dsum \limits_{l}}%
%BeginExpansion
{\displaystyle\sum\limits_{l}}
%EndExpansion
$ only. We have,%

\begin{align*}
T_{k,m}  &  =\left(  \frac{1}{4}\right)  ^{3}\left[
\begin{array}
[c]{c}%
\bar{a}_{o,o,o}b_{k,o,m}+a_{k,o,m}\bar{b}_{o,o,o}+\bar{a}_{o,o,m}%
b_{k,o,o}+a_{k,o,o}\bar{b}_{o,o,m}\\
+i\epsilon_{kk_{1}k_{2}}\left(  a_{k_{1},o,o}b_{k_{2},o,m}\right)
+i\epsilon_{kk_{1}k_{2}}\left(  a_{k_{1},o,m}b_{k_{2},o,o}\right) \\
+i\epsilon_{mm_{1}m_{2}}\left(  \bar{a}_{o,o,m_{1}}b_{k,o,m_{2}}\right)
+i\epsilon_{mm_{1}m_{2}}\left(  a_{k,o,m_{1}}\bar{b}_{o,o,m_{2}}\right) \\
+\frac{1}{2}i\epsilon_{kk_{1}k_{2}}i\epsilon_{mm_{1}m_{2}}\left(
a_{k_{1},o,m_{1}}b_{k_{2},o,m_{2}}\right) \\
+\frac{1}{2}i\epsilon_{kk_{1}k_{2}}i\epsilon_{mm_{1}m_{2}}\left(
a_{k_{1},o,m_{1}}b_{k_{2},o,m_{2}}\right)
\end{array}
\right]  \text{ }\\
&  +\left(  \frac{1}{4}\right)  ^{3}\left[
\begin{array}
[c]{c}%
\bar{a}_{o,l,o}b_{k,l,m}+a_{k.l,m}\bar{b}_{o,l,o}+a_{k,l,o}\bar{b}%
_{o,l,m}+\bar{a}_{o,l,m}b_{k,l,o}\\
+i\epsilon_{kk_{1}k_{2}}\left(  a_{k_{1},l,o}b_{k_{2},l,m}\right)
+i\epsilon_{kk_{1}k_{2}}\left(  a_{k_{1},l,m}b_{k_{2},l,o}\right) \\
+i\epsilon_{mm_{1}m_{2}}\left(  \bar{a}_{o,l,m_{1}}b_{k,l,m2}\right)
+i\epsilon_{mm_{1}m_{2}}\left(  a_{k.l,m_{1}}\bar{b}_{o,l,m_{2}}\right) \\
+\frac{1}{2}i\epsilon_{kk_{1}k_{2}}i\epsilon_{mm_{1}m_{2}}\left(
a_{k_{1},l,m_{1}}b_{k_{2},l,m_{2}}\right)  +\frac{1}{2}i\epsilon_{kk_{1}k_{2}%
}i\epsilon_{mm_{1}m_{2}}\left(  a_{k_{1},o,m_{1}}b_{k_{2},o,m_{2}}\right)
\end{array}
\right]
\end{align*}

\subsection{Isospin-pseudospin dyadics, torques and entangled torques}

Here we have [$k=0$ or $%
%TCIMACRO{\dsum \limits_{k}}%
%BeginExpansion
{\displaystyle\sum\limits_{k}}
%EndExpansion
$, $l,m$ can be distributed in two ways, and $m$ and $l$ can be distributed in
two ways] yielding $2^{3}=8$ dyadics, $2^{3}=8$ torques, and $2^{1}=2$
entangled torques from $k=0$ or $%
%TCIMACRO{\dsum \limits_{k}}%
%BeginExpansion
{\displaystyle\sum\limits_{k}}
%EndExpansion
$ only. We have,%

\begin{align*}
T_{l,m}  &  =\left(  \frac{1}{4}\right)  ^{3}\left[
\begin{array}
[c]{c}%
\bar{a}_{o,o,o}\bar{b}_{o,l,m}+\bar{a}_{o,l,m}\bar{b}_{o,o,o}+\bar{a}%
_{o,l,o}\bar{b}_{o,o,m}+\bar{a}_{o,o,m}\bar{b}_{o,l,o}\\
+i\epsilon_{mm_{1}m_{2}}\left(  \bar{a}_{o,o,m_{1}}\bar{b}_{o,l,m_{2}}\right)
+i\epsilon_{mm_{1}m_{2}}\left(  \bar{a}_{o,l,m_{1}}\bar{b}_{o,o,m_{2}}\right)
\\
+i\epsilon_{ll_{1}l_{2}}\left(  \bar{a}_{o,l_{1},o}\bar{b}_{o,l_{2},m}\right)
+i\epsilon_{ll_{1}l_{2}}\left(  \bar{a}_{o,l_{1},m}\bar{b}_{o,l_{2},o}\right)
\\
+\frac{1}{2}i\epsilon_{ll_{1}l_{2}}i\epsilon_{mm_{1}m_{2}}\left(  \bar
{a}_{o,l_{1},m_{1}}\bar{b}_{o,l_{2},m_{2}}\right)  +\frac{1}{2}i\epsilon
_{ll_{1}l_{2}}i\epsilon_{mm_{1}m_{2}}\left(  \bar{a}_{o,l_{1},m_{1}}\bar
{b}_{o,l_{2},m_{2}}\right)
\end{array}
\right] \\
&  +\left(  \frac{1}{4}\right)  ^{3}\left[
\begin{array}
[c]{c}%
a_{k,l,o}b_{k,o,m}+a_{k,o,m}b_{k,l,o}+a_{k.l,m}b_{k,o,o}+a_{k,o,o}b_{k,l,m}\\
+i\epsilon_{mm_{1}m_{2}}\left(  a_{k.l,m_{1}}b_{k,o,m_{2}}\right)
+i\epsilon_{mm_{1}m_{2}}\left(  a_{k,o,m_{1}}b_{k,l,m_{2}}\right) \\
+i\epsilon_{ll_{1}l_{2}}\left(  a_{k,l_{1},m}b_{k,l_{2},o}\right)
+i\epsilon_{ll_{1}l_{2}}\left(  a_{k,l_{1},o}b_{k,l_{2},m}\right) \\
+\frac{1}{2}i\epsilon_{ll_{1}l_{2}}i\epsilon_{mm_{1}m_{2}}\left(
a_{k,l_{1},m_{1}}b_{k,l_{2},m_{2}}\right)  +\frac{1}{2}i\epsilon_{ll_{1}l_{2}%
}i\epsilon_{mm_{1}m_{2}}\left(  a_{k,l_{1},m_{1}}b_{k,l_{2},m_{2}}\right)
\end{array}
\right]  \text{ }%
\end{align*}

\subsection{\label{tensor}Dirac spin-isospin-pseudospin tensors, torques, and
entangled 2 and 3 torques}

Here we have [$k,l,$ and $m$ can be considered as three qubits] yielding
$2^{3}=8$ tensors of third rank, the number of torques for each spin is
$2^{2}=4$ yielding for three spins $3\times4=12$ torques, the number of
entngled torques for each pair of torques is $2$ and since there are three
pair of torques that can entangle, the results is $3\times2=6$ entangled
torques. Finally we can only have $1$ entangled three torques We have the result,%

\begin{align*}
&  T_{k,l,m}\\
&  =\left(  \frac{1}{4}\right)  ^{3}\left[
\begin{array}
[c]{c}%
a_{k,l,m}b_{0,0,0}+a_{k,l,0}b_{0,0,m}+a_{k,0,m}b_{0,l,0}+a_{0,l,m}b_{k,0,0}\\
+a_{k,0,0}b_{0,l,m}+a_{0,0,m}b_{k,l,0}+a_{0,l,0}b_{k,0,m}+a_{o,o,o}b_{k.l.m}%
\end{array}
\right] \\
&  +\left(  \frac{1}{4}\right)  ^{3}\left[
\begin{array}
[c]{c}%
\epsilon_{kk_{1}k_{2}}a_{k_{1},0,0}b_{k_{2},l,m}+\epsilon_{kk_{1}k_{2}%
}a_{k_{1},l,0}b_{k_{2},0,m}+\epsilon_{kk_{1}k_{2}}a_{k_{1},0,m}b_{k_{2}%
,l,0}+\epsilon_{kk_{1}k_{2}}a_{k_{1},l,m}b_{k_{2},0,0}\\
+\epsilon_{ll_{1}l_{2}}a_{0,l_{1},0}b_{k,l_{2},m}+\epsilon_{ll_{1}l_{2}%
}a_{k,l_{1},0}b_{0,l_{2},m}+\epsilon_{ll_{1}l_{2}}a_{0,l_{1},m}b_{k,l_{2}%
,0}+\epsilon_{ll_{1}l_{2}}a_{k,l_{1},m}b_{0,l_{2},0}\\
+\epsilon_{m;m_{1},m_{2}}a_{0,0.m_{1}}b_{k,l,m_{2}}+\epsilon_{m;m_{1},m_{2}%
}a_{0,l.m_{1}}b_{k,0,m_{2}}+\epsilon_{m;m_{1},m_{2}}a_{k,0.m_{1}}b_{0,l,m_{2}%
}+\epsilon_{m;m_{1},m_{2}}a_{k,l.m_{1}}b_{0,0,m_{2}}%
\end{array}
\right] \\
&  +\left(  \frac{1}{4}\right)  ^{3}\left[
\begin{array}
[c]{c}%
i\epsilon_{kk_{1}k_{2}}i\epsilon_{ll_{1}l_{2}}\left(  a_{k_{1},l_{1}%
,o}b_{k_{2},l_{2},m}+a_{k_{1},l_{1},m}b_{k_{2},l_{2},o}\right) \\
+i\epsilon_{kk_{1}k_{2}}i\epsilon_{mm_{1}m_{2}}\left(  a_{k_{1},o,m_{1}%
}b_{k_{2},l,m_{2}}+a_{k_{1},l,m_{1}}b_{k_{2},o,m_{2}}\right) \\
+i\epsilon_{ll_{1}l_{2}}i\epsilon_{mm_{1}m_{2}}\left(  a_{o,l_{1},m_{1}%
}b_{k,l_{2},m_{2}}+a_{k,l_{1},m_{1}}b_{o,l_{2},m_{2}}\right) \\
+i\epsilon_{kk_{1}k_{2}}i\epsilon_{ll_{1}l_{2}}i\epsilon_{mm_{1}m_{2}}\left(
a_{k_{1},l_{1},m_{1}}b_{k_{2},l_{2},m_{2}}\right)
\end{array}
\right]
\end{align*}

\section{\label{qtens3}Quantum Transport Equations for $N_{s}=3$}

The equation for $G_{o,o,o}^{\gtrless}$ has been given in the text. Here we
give the rest of the equations. We have for the Pauli-Dirac spin,%

\begin{align}
&  2^{N_{s}}i\hbar\left(  \frac{\partial}{\partial t_{1}}+\frac{\partial
}{\partial t_{2}}\right)  G_{k,o,o}^{\gtrless}=\nonumber\\
&  =\left[  \mathcal{\tilde{H}},G^{\lessgtr}\right]  _{k,o,o}+\left[
\Sigma^{\lessgtr},\mathrm{Re}G^{r}\right]  _{k,o,o}-\frac{i}{2}\left\{
\Gamma,G^{\lessgtr}\right\}  _{k,o,o}+\frac{i}{2}\left\{  \Sigma^{\lessgtr
},A\right\}  _{k,o,o} \label{QT2}%
\end{align}
where,%
\begin{align}
&  \left[  \mathcal{\tilde{H}},G^{\lessgtr}\right]  _{k,o,o}=\nonumber\\
&  =%
\begin{array}
[c]{c}%
\left[  H_{o,o,o},G_{k,o,o}^{\gtrless}\right]  +\left[  H_{o,o,m}%
,G_{k,o,m}^{\gtrless}\right]  +i\epsilon_{kk_{1}k_{2}}\left\{  H_{k_{1}%
,o,o},G_{k_{2},o,o}^{\gtrless}\right\} \\
+\left[  H_{o,l,o},G_{k,l,o}^{\gtrless}\right]  +\left[  H_{o,l,m}%
,G_{k,l,m}^{\gtrless}\right]  +i\epsilon_{kk_{1}k_{2}}\left\{  H_{k_{1}%
,l,o},G_{k_{2},l,o}^{\gtrless}\right\} \\
+\left[  H_{k,o,o},G_{o,o,o}^{\gtrless}\right]  +\left[  H_{k,o,m}%
,G_{o,o,m}^{\gtrless}\right]  +i\epsilon_{kk_{1}k_{2}}\left\{  H_{k_{1}%
,o,m},G_{k_{2},o,m}^{\gtrless}\right\} \\
+\left[  H_{k,l,o},G_{o,l,o}^{\gtrless}\right]  +\left[  H_{k.l,m}%
,G_{o,l,m}^{\gtrless}\right]  +i\epsilon_{kk_{1}k_{2}}\left\{  H_{k_{1}%
,l,m},G_{k_{2},l,m}^{\gtrless}\right\}
\end{array}
\label{QT2-1}%
\end{align}%
\begin{align}
&  \left[  \Sigma^{\lessgtr},\mathrm{Re}G^{r}\right]  _{k,o,o}=\nonumber\\
&  =%
\begin{array}
[c]{c}%
\left[  \Sigma_{o,o,o}^{\lessgtr},\mathrm{Re}G_{k,o,o}^{r}\right]  +\left[
\Sigma_{o,o,m}^{\lessgtr},\mathrm{Re}G_{k,o,m}^{r}\right]  +i\epsilon
_{kk_{1}k_{2}}\left\{  \Sigma_{k_{1},o,o}^{\lessgtr},\mathrm{Re}G_{k_{2}%
,o,o}^{r}\right\} \\
+\left[  \Sigma_{o,l,o}^{\lessgtr},\mathrm{Re}G_{k,l,o}^{r}\right]  +\left[
\Sigma_{o,l,m}^{\lessgtr},\mathrm{Re}G_{k,l,m}^{r}\right]  +i\epsilon
_{kk_{1}k_{2}}\left\{  \Sigma_{k_{1},l,o}^{\lessgtr},\mathrm{Re}G_{k_{2}%
,l,o}^{r}\right\} \\
+\left[  \Sigma_{k,o,o}^{\lessgtr},\mathrm{Re}G_{o,o,o}^{r}\right]  +\left[
\Sigma_{k,o,m}^{\lessgtr},\mathrm{Re}G_{o,o,m}^{r}\right]  +i\epsilon
_{kk_{1}k_{2}}\left\{  \Sigma_{k_{1},o,m}^{\lessgtr},\mathrm{Re}G_{k_{2}%
,o,m}^{r}\right\} \\
+\left[  \Sigma_{k,l,o}^{\lessgtr},\mathrm{Re}G_{o,l,o}^{r}\right]  +\left[
\Sigma_{k.l,m}^{\lessgtr},\mathrm{Re}G_{o,l,m}^{r}\right]  +i\epsilon
_{kk_{1}k_{2}}\left\{  \Sigma_{k_{1},l,m}^{\lessgtr},\mathrm{Re}G_{k_{2}%
,l,m}^{r}\right\}
\end{array}
\nonumber\\
&  \label{QT2-2}%
\end{align}%
\begin{align}
&  \left\{  \Gamma,G^{\lessgtr}\right\}  _{k,o,o}=\nonumber\\
&  =%
\begin{array}
[c]{c}%
\left\{  \Gamma_{o,o,o},G_{k,o,o}^{\gtrless}\right\}  +\left\{  \Gamma
_{o,o,m},G_{k,o,m}^{\gtrless}\right\}  +i\epsilon_{kk_{1}k_{2}}\left[
\Gamma_{k_{1},o,o},G_{k_{2},o,o}^{\gtrless}\right] \\
+\left\{  \Gamma_{o,l,o},G_{k,l,o}^{\gtrless}\right\}  +\left\{
\Gamma_{o,l,m},G_{k,l,m}^{\gtrless}\right\}  +i\epsilon_{kk_{1}k_{2}}\left[
\Gamma_{k_{1},l,o},G_{k_{2},l,o}^{\gtrless}\right] \\
+\left\{  \Gamma_{k,o,o},G_{o,o,o}^{\gtrless}\right\}  +\left\{
\Gamma_{k,o,m},G_{o,o,m}^{\gtrless}\right\}  +i\epsilon_{kk_{1}k_{2}}\left[
\Gamma_{k_{1},o,m},G_{k_{2},o,m}^{\gtrless}\right] \\
+\left\{  \Gamma_{k,l,o},G_{o,l,o}^{\gtrless}\right\}  +\left\{
\Gamma_{k.l,m},G_{o,l,m}^{\gtrless}\right\}  +i\epsilon_{kk_{1}k_{2}}\left[
\Gamma_{k_{1},l,m},G_{k_{2},l,m}^{\gtrless}\right]
\end{array}
\label{QT2-3}%
\end{align}%
\begin{align}
&  \left\{  \Sigma^{\lessgtr},A\right\}  _{k,o,o}=\nonumber\\
&  =%
\begin{array}
[c]{c}%
\left\{  \Sigma_{o,o,o}^{\lessgtr},A_{k,o,o}\right\}  +\left\{  \Sigma
_{o,o,m}^{\lessgtr},A_{k,o,m}\right\}  +i\epsilon_{kk_{1}k_{2}}\left[
\Sigma_{k_{1},o,o}^{\lessgtr},A_{k_{2},o,o}\right] \\
+\left\{  \Sigma_{o,l,o}^{\lessgtr},A_{k,l,o}\right\}  +\left\{
\Sigma_{o,l,m}^{\lessgtr},A_{k,l,m}\right\}  +i\epsilon_{kk_{1}k_{2}}\left[
\Sigma_{k_{1},l,o}^{\lessgtr},A_{k_{2},l,o}\right] \\
+\left\{  \Sigma_{k,o,o}^{\lessgtr},A_{o,o,o}\right\}  +\left\{
\Sigma_{k,o,m}^{\lessgtr},A_{o,o,m}\right\}  +i\epsilon_{kk_{1}k_{2}}\left[
\Sigma_{k_{1},o,m}^{\lessgtr},A_{k_{2},o,m}\right] \\
+\left\{  \Sigma_{k,l,o}^{\lessgtr},A_{o,l,o}\right\}  +\left\{
\Sigma_{k.l,m}^{\lessgtr},A_{o,l,m}\right\}  +i\epsilon_{kk_{1}k_{2}}\left[
\Sigma_{k_{1},l,m}^{\lessgtr},A_{k_{2},l,m}\right]
\end{array}
\label{QT2-4}%
\end{align}
We have for the valley spin,%

\begin{align}
&  2^{N_{s}}i\hbar\left(  \frac{\partial}{\partial t_{1}}+\frac{\partial
}{\partial t_{2}}\right)  G_{o,l,o}^{\gtrless}=\nonumber\\
&  \left[  \mathcal{\tilde{H}},G^{\lessgtr}\right]  _{o,l,o}+\left[
\Sigma^{\lessgtr},\mathrm{Re}G^{r}\right]  _{o,l,o}-\frac{i}{2}\left\{
\Gamma,G^{\lessgtr}\right\}  _{o,l,o}+\frac{i}{2}\left\{  \Sigma^{\lessgtr
},A\right\}  _{o,l,o} \label{QT3}%
\end{align}
where%
\begin{align}
&  \left[  \mathcal{\tilde{H}},G^{\lessgtr}\right]  _{o,l,o}=\nonumber\\
&  =%
\begin{array}
[c]{c}%
\left[  H_{o,o,o},G_{o,l,o}^{\gtrless}\right]  +\left[  H_{o,o,m}%
,G_{o,l,m}^{\gtrless}\right]  +i\epsilon_{ll_{1}l_{2}}\left\{  H_{o,l_{1}%
,o},G_{o,l_{2},o}^{\gtrless}\right\} \\
+\left[  H_{o,l,o},G_{o,o,o}^{\lessgtr}\right]  +\left[  H_{o,l,m}%
,G_{o,o,m}^{\gtrless}\right]  +i\epsilon_{ll_{1}l_{2}}\left\{  H_{o,l_{1}%
,m},G_{o,l_{2},m}^{\gtrless}\right\} \\
+\left[  H_{k,o,o},G_{k,l,o}^{\gtrless}\right]  +\left[  H_{k,o,m}%
,G_{k,l,m}^{\gtrless}\right]  +i\epsilon_{ll_{1}l_{2}}\left\{  H_{k,l_{1}%
,o},G_{k,l_{2},o}^{\gtrless}\right\} \\
+\left[  H_{k,l,o},G_{k,o,o}^{\gtrless}\right]  +\left[  H_{k.l,m}%
,G_{k,o,m}^{\gtrless}\right]  +i\epsilon_{ll_{1}l_{2}}\left\{  H_{k,l_{1}%
,m},G_{k,l_{2},m}^{\gtrless}\right\}
\end{array}
\label{QT3-1}%
\end{align}%
\begin{align}
&  \left[  \Sigma^{\lessgtr},\mathrm{Re}G^{r}\right]  _{o,l,o}=\nonumber\\
&  =%
\begin{array}
[c]{c}%
\left[  \Sigma_{o,o,o}^{\gtrless},\mathrm{Re}G_{o,l,o}^{r}\right]  +\left[
\Sigma_{o,o,m}^{\gtrless},\mathrm{Re}G_{o,l,m}^{r}\right]  +i\epsilon
_{ll_{1}l_{2}}\left\{  \Sigma_{o,l_{1},o}^{\gtrless},\mathrm{Re}G_{o,l_{2}%
,o}^{r}\right\} \\
+\left[  \Sigma_{o,l,o}^{\gtrless},\mathrm{Re}G_{o,o,o}^{r}\right]  +\left[
\Sigma_{o,l,m}^{\gtrless},\mathrm{Re}G_{o,o,m}^{r}\right]  +i\epsilon
_{ll_{1}l_{2}}\left\{  \Sigma_{o,l_{1},m}^{\gtrless},\mathrm{Re}G_{o,l_{2}%
,m}^{r}\right\} \\
+\left[  \Sigma_{k,o,o}^{\gtrless},\mathrm{Re}G_{k,l,o}^{r}\right]  +\left[
\Sigma_{k,o,m}^{\gtrless},\mathrm{Re}G_{k,l,m}^{r}\right]  +i\epsilon
_{ll_{1}l_{2}}\left\{  \Sigma_{k,l_{1},o}^{\gtrless},\mathrm{Re}G_{k,l_{2}%
,o}^{r}\right\} \\
+\left[  \Sigma_{k,l,o}^{\gtrless},\mathrm{Re}G_{k,o,o}^{r}\right]  +\left[
\Sigma_{k.l,m}^{\gtrless},\mathrm{Re}G_{k,o,m}^{r}\right]  +i\epsilon
_{ll_{1}l_{2}}\left\{  \Sigma_{k,l_{1},m}^{\gtrless},\mathrm{Re}G_{k,l_{2}%
,m}^{r}\right\}
\end{array}
\nonumber\\
&  \label{QT3-2}%
\end{align}%
\begin{align}
&  \left\{  \Gamma,G^{\lessgtr}\right\}  _{o,l,o}=\nonumber\\
&  =%
\begin{array}
[c]{c}%
\left\{  \Gamma_{o,o,o},G_{o,l,o}^{\gtrless}\right\}  +\left\{  \Gamma
_{o,o,m},G_{o,l,m}^{\gtrless}\right\}  +i\epsilon_{ll_{1}l_{2}}\left[
\Gamma_{o,l_{1},o},G_{o,l_{2},o}^{\gtrless}\right] \\
+\left\{  \Gamma_{o,l,o},G_{o,o,o}^{\lessgtr}\right\}  +\left\{
\Gamma_{o,l,m},G_{o,o,m}^{\gtrless}\right\}  +i\epsilon_{ll_{1}l_{2}}\left[
\Gamma_{o,l_{1},m},G_{o,l_{2},m}^{\gtrless}\right] \\
+\left\{  \Gamma_{k,o,o},G_{k,l,o}^{\gtrless}\right\}  +\left\{
\Gamma_{k,o,m},G_{k,l,m}^{\gtrless}\right\}  +i\epsilon_{ll_{1}l_{2}}\left[
\Gamma_{k,l_{1},o},G_{k,l_{2},o}^{\gtrless}\right] \\
+\left\{  \Gamma_{k,l,o},G_{k,o,o}^{\gtrless}\right\}  +\left\{
\Gamma_{k.l,m},G_{k,o,m}^{\gtrless}\right\}  +i\epsilon_{ll_{1}l_{2}}\left[
\Gamma_{k,l_{1},m},G_{k,l_{2},m}^{\gtrless}\right]
\end{array}
\label{QT3-3}%
\end{align}%
\begin{align}
&  \left\{  \Sigma^{\lessgtr},A\right\}  _{o,l,o}=\nonumber\\
&  =%
\begin{array}
[c]{c}%
\left\{  \Sigma_{o,o,o}^{\gtrless},A_{o,l,o}\right\}  +\left\{  \Sigma
_{o,o,m}^{\gtrless},A_{o,l,m}\right\}  +i\epsilon_{ll_{1}l_{2}}\left[
\Sigma_{o,l_{1},o}^{\gtrless},A_{o,l_{2},o}\right] \\
+\left\{  \Sigma_{o,l,o}^{\gtrless},A_{o,o,o}\right\}  +\left\{
\Sigma_{o,l,m}^{\gtrless},A_{o,o,m}\right\}  +i\epsilon_{ll_{1}l_{2}}\left[
\Sigma_{o,l_{1},m}^{\gtrless},A_{o,l_{2},m}\right] \\
+\left\{  \Sigma_{k,o,o}^{\gtrless},A_{k,l,o}\right\}  +\left\{
\Sigma_{k,o,m}^{\gtrless},A_{k,l,m}\right\}  +i\epsilon_{ll_{1}l_{2}}\left[
\Sigma_{k,l_{1},o}^{\gtrless},A_{k,l_{2},o}\right] \\
+\left\{  \Sigma_{k,l,o}^{\gtrless},A_{k,o,o}\right\}  +\left\{
\Sigma_{k.l,m}^{\gtrless},A_{k,o,m}\right\}  +i\epsilon_{ll_{1}l_{2}}\left[
\Sigma_{k,l_{1},m}^{\gtrless},A_{k,l_{2},m}\right]
\end{array}
\label{QT3-4}%
\end{align}
We have for the pseudospin,%

\begin{align}
&  2^{N_{s}}i\hbar\left(  \frac{\partial}{\partial t_{1}}+\frac{\partial
}{\partial t_{2}}\right)  G_{o,o,m}^{\gtrless}=\nonumber\\
&  =\left[  \mathcal{\tilde{H}},G^{\lessgtr}\right]  _{o,o,m}+\left[
\Sigma^{\lessgtr},\mathrm{Re}G^{r}\right]  _{o,o,m}-\frac{i}{2}\left\{
\Gamma,G^{\lessgtr}\right\}  _{o,o,m}+\frac{i}{2}\left\{  \Sigma^{\lessgtr
},A\right\}  _{o,o,m} \label{QT4}%
\end{align}
where,%
\begin{align}
&  \left[  \mathcal{\tilde{H}},G^{\lessgtr}\right]  _{o,o,m}=\nonumber\\
&  =%
\begin{array}
[c]{c}%
\left[  H_{o,o,o},G_{o,o,m}^{\gtrless}\right]  +\left[  H_{o,o,m}%
,G_{o,o,o}^{\gtrless}\right]  +i\epsilon_{mm_{1}m_{2}}\left\{  H_{o,o,m_{1}%
},G_{o,o,m_{2}}^{\gtrless}\right\} \\
+\left[  H_{o,l,o},G_{o,l,m}^{\gtrless}\right]  +\left[  H_{o,l,m}%
,G_{o,l,o}^{\gtrless}\right]  +i\epsilon_{mm_{1}m_{2}}\left\{  H_{o,l,m_{1}%
},G_{o,l,m_{2}}^{\gtrless}\right\} \\
+\left[  H_{k,o,o},G_{k,o,m}^{\gtrless}\right]  +\left[  H_{k,o,m}%
,G_{k,o,o}^{\gtrless}\right]  +i\epsilon_{mm_{1}m_{2}}\left\{  H_{k,o,m_{1}%
},G_{k,o,m_{2}}^{\gtrless}\right\} \\
+\left[  H_{k,l,o},G_{k,l,m}^{\gtrless}\right]  +\left[  H_{k,l,m}%
,G_{k,l,o}^{\gtrless}\right]  +i\epsilon_{mm_{1}m_{2}}\left\{  H_{k,l,m_{1}%
},G_{k,l,m_{2}}^{\gtrless}\right\}
\end{array}
\label{QT4-1}%
\end{align}%
\begin{align}
&  \left[  \Sigma^{\lessgtr},\mathrm{Re}G^{r}\right]  _{o,o,m}=\nonumber\\
&  =%
\begin{array}
[c]{c}%
\left[  \Sigma_{o,o,o}^{\lessgtr},\mathrm{Re}G_{o,o,m}^{r}\right]  +\left[
\Sigma_{o,o,m}^{\lessgtr},\mathrm{Re}G_{o,o,o}^{r}\right]  +i\epsilon
_{mm_{1}m_{2}}\left\{  \Sigma_{o,o,m_{1}}^{\lessgtr},\mathrm{Re}G_{o,o,m_{2}%
}^{r}\right\} \\
+\left[  \Sigma_{o,l,o}^{\lessgtr},\mathrm{Re}G_{o,l,m}^{r}\right]  +\left[
\Sigma_{o,l,m}^{\lessgtr},\mathrm{Re}G_{o,l,o}^{r}\right]  +i\epsilon
_{mm_{1}m_{2}}\left\{  \Sigma_{o,l,m_{1}}^{\lessgtr},\mathrm{Re}G_{o,l,m_{2}%
}^{r}\right\} \\
+\left[  \Sigma_{k,o,o}^{\lessgtr},\mathrm{Re}G_{k,o,m}^{r}\right]  +\left[
\Sigma_{k,o,m}^{\lessgtr},\mathrm{Re}G_{k,o,o}^{r}\right]  +i\epsilon
_{mm_{1}m_{2}}\left\{  \Sigma_{k,o,m_{1}}^{\lessgtr},\mathrm{Re}G_{k,o,m_{2}%
}^{r}\right\} \\
+\left[  \Sigma_{k,l,o}^{\lessgtr},\mathrm{Re}G_{k,l,m}^{r}\right]  +\left[
\Sigma_{k,l,m}^{\lessgtr},\mathrm{Re}G_{k,l,o}^{r}\right]  +i\epsilon
_{mm_{1}m_{2}}\left\{  \Sigma_{k,l,m_{1}}^{\lessgtr},\mathrm{Re}G_{k,l,m_{2}%
}^{r}\right\}
\end{array}
\nonumber\\
&  \label{QT4-2}%
\end{align}%
\begin{align}
&  \left\{  \Gamma,G^{\lessgtr}\right\}  _{o,o,m}=\nonumber\\
&  =%
\begin{array}
[c]{c}%
\left\{  \Gamma_{o,o,o},G_{o,o,m}^{\gtrless}\right\}  +\left\{  \Gamma
_{o,o,m},G_{o,o,o}^{\gtrless}\right\}  +i\epsilon_{mm_{1}m_{2}}\left[
\Gamma_{o,o,m_{1}},G_{o,o,m_{2}}^{\gtrless}\right] \\
+\left\{  \Gamma_{o,l,o},G_{o,l,m}^{\gtrless}\right\}  +\left\{
\Gamma_{o,l,m},G_{o,l,o}^{\gtrless}\right\}  +i\epsilon_{mm_{1}m_{2}}\left[
\Gamma_{o,l,m_{1}},G_{o,l,m_{2}}^{\gtrless}\right] \\
+\left\{  \Gamma_{k,o,o},G_{k,o,m}^{\gtrless}\right\}  +\left\{
\Gamma_{k,o,m},G_{k,o,o}^{\gtrless}\right\}  +i\epsilon_{mm_{1}m_{2}}\left[
\Gamma_{k,o,m_{1}},G_{k,o,m_{2}}^{\gtrless}\right] \\
+\left\{  \Gamma_{k,l,o},G_{k,l,m}^{\gtrless}\right\}  +\left\{
\Gamma_{k,l,m},G_{k,l,o}^{\gtrless}\right\}  +i\epsilon_{mm_{1}m_{2}}\left[
\Gamma_{k,l,m_{1}},G_{k,l,m_{2}}^{\gtrless}\right]
\end{array}
\label{QT4-3}%
\end{align}%
\begin{align}
&  \left\{  \Sigma^{\lessgtr},A\right\}  _{o,o,m}=\nonumber\\
&  =%
\begin{array}
[c]{c}%
\left\{  \Sigma_{o,o,o}^{\lessgtr},A_{o,o,m}\right\}  +\left\{  \Sigma
_{o,o,m}^{\lessgtr},A_{o,o,o}\right\}  +i\epsilon_{mm_{1}m_{2}}\left[
\Sigma_{o,o,m_{1}}^{\lessgtr},A_{o,o,m_{2}}\right] \\
+\left\{  \Sigma_{o,l,o}^{\lessgtr},A_{o,l,m}\right\}  +\left\{
\Sigma_{o,l,m}^{\lessgtr},A_{o,l,o}\right\}  +i\epsilon_{mm_{1}m_{2}}\left[
\Sigma_{o,l,m_{1}}^{\lessgtr},A_{o,l,m_{2}}\right] \\
+\left\{  \Sigma_{k,o,o}^{\lessgtr},A_{k,o,m}\right\}  +\left\{
\Sigma_{k,o,m}^{\lessgtr},A_{k,o,o}\right\}  +i\epsilon_{mm_{1}m_{2}}\left[
\Sigma_{k,o,m_{1}}^{\lessgtr},A_{k,o,m_{2}}\right] \\
+\left\{  \Sigma_{k,l,o}^{\lessgtr},A_{k,l,m}\right\}  +\left\{
\Sigma_{k,l,m}^{\lessgtr},A_{k,l,o}\right\}  +i\epsilon_{mm_{1}m_{2}}\left[
\Sigma_{k,l,m_{1}}^{\lessgtr},A_{k,l,m_{2}}\right]
\end{array}
\label{QT4-4}%
\end{align}

We have for the entangled Dirac spin and valley spin,%

\begin{align}
&  2^{N_{s}}i\hbar\left(  \frac{\partial}{\partial t_{1}}+\frac{\partial
}{\partial t_{2}}\right)  G_{k,l,o}^{\gtrless}=\nonumber\\
&  =\left[  \mathcal{\tilde{H}},G^{\lessgtr}\right]  _{k,l,o}+\left[
\Sigma^{\lessgtr},\mathrm{Re}G^{r}\right]  _{k,l,o}-\frac{i}{2}\left\{
\Gamma,G^{\lessgtr}\right\}  _{k,l,o}+\frac{i}{2}\left\{  \Sigma^{\lessgtr
},A\right\}  _{k,l,o} \label{QT5}%
\end{align}
where,%
\begin{align}
&  \left[  \mathcal{\tilde{H}},G^{\lessgtr}\right]  _{k,l,o}=\nonumber\\
&  =%
\begin{array}
[c]{c}%
\begin{array}
[c]{c}%
\left[  H_{o,o,o},G_{k,l,o}^{\lessgtr}\right]  +\left[  H_{o,o,m}%
,G_{k,l,m}^{\lessgtr}\right]  +i\epsilon_{kk_{1}k_{2}}\left\{  H_{k_{1}%
,o,o},G_{k_{2},l,o}^{\lessgtr}\right\} \\
+i\epsilon_{ll_{1}l_{2}}\left\{  H_{o,l_{1},o},G_{k,l_{2},o}^{\lessgtr
}\right\}  +\frac{1}{2}i\epsilon_{kk_{1}k_{2}}i\epsilon_{ll_{1}l_{2}}\left[
H_{k_{1},l_{1},o},G_{k_{2},l_{2},o}^{\lessgtr}\right]
\end{array}
\\
+%
\begin{array}
[c]{c}%
\left[  H_{o,l,o},G_{k,o,o}^{\lessgtr}\right]  +\left[  H_{o,l,m}%
,G_{k,o,m}^{\lessgtr}\right]  +i\epsilon_{kk_{1}k_{2}}\left\{  H_{k_{1}%
,l,o},G_{k_{2},o,o}^{\lessgtr}\right\} \\
+i\epsilon_{ll_{1}l_{2}}\left\{  H_{k,l_{1},o}G_{o,l_{2},o}^{\lessgtr
}\right\}  +\frac{1}{2}i\epsilon_{kk_{1}k_{2}}i\epsilon_{ll_{1}l_{2}}\left[
H_{k_{1},l_{1},o},G_{k_{2},l_{2},o}^{\lessgtr}\right]
\end{array}
\\
+%
\begin{array}
[c]{c}%
\left[  H_{k,l,o},G_{o,o,o}^{\lessgtr}\right]  +\left[  H_{k.l,m}%
,G_{o,o,m}^{\lessgtr}\right]  +i\epsilon_{kk_{1}k_{2}}\left\{  H_{k_{1}%
,l,m},G_{k_{2},o,m}^{\lessgtr}\right\} \\
+i\epsilon_{ll_{1}l_{2}}\left\{  H_{k,l_{1},m},G_{o,l_{2},m}^{\lessgtr
}\right\}  +\frac{1}{2}i\epsilon_{kk_{1}k_{2}}i\epsilon_{ll_{1}l_{2}}\left[
H_{k_{1},l_{1},m},G_{k_{2},l_{2},m}^{\lessgtr}\right]
\end{array}
\\
+%
\begin{array}
[c]{c}%
\left[  H_{k,o,o},G_{o,l,o}^{\lessgtr}\right]  +\left[  H_{k,o,m}%
,G_{o,l,m}^{\lessgtr}\right]  +i\epsilon_{kk_{1}k_{2}}\left\{  H_{k_{1}%
,o,m},G_{k_{2},l,m}^{\lessgtr}\right\} \\
+i\epsilon_{ll_{1}l_{2}}\left\{  H_{o,l_{1},m},G_{k,l_{2},m}^{\lessgtr
}\right\}  +\frac{1}{2}i\epsilon_{kk_{1}k_{2}}i\epsilon_{ll_{1}l_{2}}\left[
H_{k_{1},l_{1},m},G_{k_{2},l_{2},m}^{\lessgtr}\right]
\end{array}
\end{array}
\label{QT5-1}%
\end{align}%
\begin{align}
&  \left[  \Sigma^{\lessgtr},\mathrm{Re}G^{r}\right]  _{k,l,o}=\nonumber\\
&  =%
\begin{array}
[c]{c}%
\begin{array}
[c]{c}%
\left[  \Sigma_{o,o,o}^{\lessgtr},\mathrm{Re}G_{k,l,o}^{r}\right]  +\left[
\Sigma_{o,o,m}^{\lessgtr},\mathrm{Re}G_{k,l,m}^{r}\right]  +i\epsilon
_{kk_{1}k_{2}}\left\{  \Sigma_{k_{1},o,o}^{\lessgtr},\mathrm{Re}G_{k_{2}%
,l,o}^{r}\right\} \\
+i\epsilon_{ll_{1}l_{2}}\left\{  \Sigma_{o,l_{1},o}^{\lessgtr},\mathrm{Re}%
G_{k,l_{2},o}^{r}\right\}  +\frac{1}{2}i\epsilon_{kk_{1}k_{2}}i\epsilon
_{ll_{1}l_{2}}\left[  \Sigma_{k_{1},l_{1},o}^{\lessgtr},\mathrm{Re}%
G_{k_{2},l_{2},o}^{r}\right]
\end{array}
\\%
\begin{array}
[c]{c}%
+\left[  \Sigma_{o,l,o}^{\lessgtr},\mathrm{Re}G_{k,o,o}^{r}\right]  +\left[
\Sigma_{o,l,m}^{\lessgtr},\mathrm{Re}G_{k,o,m}^{r}\right]  +i\epsilon
_{kk_{1}k_{2}}\left\{  \Sigma_{k_{1},l,o}^{\lessgtr},\mathrm{Re}G_{k_{2}%
,o,o}^{r}\right\} \\
+i\epsilon_{ll_{1}l_{2}}\left\{  \Sigma_{k,l_{1},o}^{\lessgtr}\mathrm{Re}%
G_{o,l_{2},o}^{r}\right\}  +\frac{1}{2}i\epsilon_{kk_{1}k_{2}}i\epsilon
_{ll_{1}l_{2}}\left[  \Sigma_{k_{1},l_{1},o}^{\lessgtr},\mathrm{Re}%
G_{k_{2},l_{2},o}^{r}\right]
\end{array}
\\%
\begin{array}
[c]{c}%
+\left[  \Sigma_{k,l,o}^{\lessgtr},\mathrm{Re}G_{o,o,o}^{r}\right]  +\left[
\Sigma_{k.l,m}^{\lessgtr},\mathrm{Re}G_{o,o,m}^{r}\right]  +i\epsilon
_{kk_{1}k_{2}}\left\{  \Sigma_{k_{1},l,m}^{\lessgtr},\mathrm{Re}G_{k_{2}%
,o,m}^{r}\right\} \\
+i\epsilon_{ll_{1}l_{2}}\left\{  \Sigma_{k,l_{1},m}^{\lessgtr},\mathrm{Re}%
G_{o,l_{2},m}^{r}\right\}  +\frac{1}{2}i\epsilon_{kk_{1}k_{2}}i\epsilon
_{ll_{1}l_{2}}\left[  \Sigma_{k_{1},l_{1},m}^{\lessgtr},\mathrm{Re}%
G_{k_{2},l_{2},m}^{r}\right]
\end{array}
\\%
\begin{array}
[c]{c}%
+\left[  \Sigma_{k,o,o}^{\lessgtr},\mathrm{Re}G_{o,l,o}^{r}\right]  +\left[
\Sigma_{k,o,m}^{\lessgtr},\mathrm{Re}G_{o,l,m}^{r}\right]  +i\epsilon
_{kk_{1}k_{2}}\left\{  \Sigma_{k_{1},o,m}^{\lessgtr},\mathrm{Re}G_{k_{2}%
,l,m}^{r}\right\} \\
+i\epsilon_{ll_{1}l_{2}}\left\{  \Sigma_{o,l_{1},m}^{\lessgtr},\mathrm{Re}%
G_{k,l_{2},m}^{r}\right\}  +\frac{1}{2}i\epsilon_{kk_{1}k_{2}}i\epsilon
_{ll_{1}l_{2}}\left[  \Sigma_{k_{1},l_{1},m}^{\lessgtr},\mathrm{Re}%
G_{k_{2},l_{2},m}^{r}\right]
\end{array}
\end{array}
\nonumber\\
&  \label{QT5-2}%
\end{align}%
\begin{align}
&  \left\{  \Gamma,G^{\lessgtr}\right\}  _{k,l,o}=\nonumber\\
&  =%
\begin{array}
[c]{c}%
\begin{array}
[c]{c}%
\left\{  \Gamma_{o,o,o},G_{k,l,o}^{\lessgtr}\right\}  +\left\{  \Gamma
_{o,o,m},G_{k,l,m}^{\lessgtr}\right\}  +i\epsilon_{kk_{1}k_{2}}\left[
\Gamma_{k_{1},o,o},G_{k_{2},l,o}^{\lessgtr}\right] \\
+i\epsilon_{ll_{1}l_{2}}\left[  \Gamma_{o,l_{1},o},G_{k,l_{2},o}^{\lessgtr
}\right]  +\frac{1}{2}i\epsilon_{kk_{1}k_{2}}i\epsilon_{ll_{1}l_{2}}\left\{
\Gamma_{k_{1},l_{1},o},G_{k_{2},l_{2},o}^{\lessgtr}\right\}
\end{array}
\\%
\begin{array}
[c]{c}%
+\left\{  \Gamma_{o,l,o},G_{k,o,o}^{\lessgtr}\right\}  +\left\{
\Gamma_{o,l,m},G_{k,o,m}^{\lessgtr}\right\}  +i\epsilon_{kk_{1}k_{2}}\left[
\Gamma_{k_{1},l,o},G_{k_{2},o,o}^{\lessgtr}\right] \\
+i\epsilon_{ll_{1}l_{2}}\left[  \Gamma_{k,l_{1},o}G_{o,l_{2},o}^{\lessgtr
}\right]  +\frac{1}{2}i\epsilon_{kk_{1}k_{2}}i\epsilon_{ll_{1}l_{2}}\left\{
\Gamma_{k_{1},l_{1},o},G_{k_{2},l_{2},o}^{\lessgtr}\right\}
\end{array}
\\%
\begin{array}
[c]{c}%
+\left\{  \Gamma_{k,l,o},G_{o,o,o}^{\lessgtr}\right\}  +\left\{
\Gamma_{k.l,m},G_{o,o,m}^{\lessgtr}\right\}  +i\epsilon_{kk_{1}k_{2}}\left[
\Gamma_{k_{1},l,m},G_{k_{2},o,m}^{\lessgtr}\right] \\
+i\epsilon_{ll_{1}l_{2}}\left[  \Gamma_{k,l_{1},m},G_{o,l_{2},m}^{\lessgtr
}\right]  +\frac{1}{2}i\epsilon_{kk_{1}k_{2}}i\epsilon_{ll_{1}l_{2}}\left\{
\Gamma_{k_{1},l_{1},m},G_{k_{2},l_{2},m}^{\lessgtr}\right\}
\end{array}
\\%
\begin{array}
[c]{c}%
+\left\{  \Gamma_{k,o,o},G_{o,l,o}^{\lessgtr}\right\}  +\left\{
\Gamma_{k,o,m},G_{o,l,m}^{\lessgtr}\right\}  +i\epsilon_{kk_{1}k_{2}}\left[
\Gamma_{k_{1},o,m},G_{k_{2},l,m}^{\lessgtr}\right] \\
+i\epsilon_{ll_{1}l_{2}}\left[  \Gamma_{o,l_{1},m},G_{k,l_{2},m}^{\lessgtr
}\right]  +\frac{1}{2}i\epsilon_{kk_{1}k_{2}}i\epsilon_{ll_{1}l_{2}}\left\{
\Gamma_{k_{1},l_{1},m},G_{k_{2},l_{2},m}^{\lessgtr}\right\}
\end{array}
\end{array}
\label{QT5-3}%
\end{align}%
\begin{align}
&  \left\{  \Sigma^{\lessgtr},A\right\}  _{k,l,o}=\nonumber\\
&  =%
\begin{array}
[c]{c}%
\begin{array}
[c]{c}%
\left\{  \Sigma_{o,o,o}^{\lessgtr},A_{k,l,o}\right\}  +\left\{  \Sigma
_{o,o,m}^{\lessgtr},A_{k,l,m}\right\}  +i\epsilon_{kk_{1}k_{2}}\left[
\Sigma_{k_{1},o,o}^{\lessgtr},A_{k_{2},l,o}\right] \\
+i\epsilon_{ll_{1}l_{2}}\left[  \Sigma_{o,l_{1},o}^{\lessgtr},A_{k,l_{2}%
,o}\right]  +\frac{1}{2}i\epsilon_{kk_{1}k_{2}}i\epsilon_{ll_{1}l_{2}}\left\{
\Sigma_{k_{1},l_{1},o}^{\lessgtr},A_{k_{2},l_{2},o}\right\}
\end{array}
\\
+%
\begin{array}
[c]{c}%
\left\{  \Sigma_{o,l,o}^{\lessgtr},A_{k,o,o}\right\}  +\left\{  \Sigma
_{o,l,m}^{\lessgtr},A_{k,o,m}\right\}  +i\epsilon_{kk_{1}k_{2}}\left[
\Sigma_{k_{1},l,o}^{\lessgtr},A_{k_{2},o,o}\right] \\
+i\epsilon_{ll_{1}l_{2}}\left[  \Sigma_{k,l_{1},o}^{\lessgtr}A_{o,l_{2}%
,o}\right]  +\frac{1}{2}i\epsilon_{kk_{1}k_{2}}i\epsilon_{ll_{1}l_{2}}\left\{
\Sigma_{k_{1},l_{1},o}^{\lessgtr},A_{k_{2},l_{2},o}\right\}
\end{array}
\\
+%
\begin{array}
[c]{c}%
\left\{  \Sigma_{k,l,o}^{\lessgtr},A_{o,o,o}\right\}  +\left\{  \Sigma
_{k.l,m}^{\lessgtr},A_{o,o,m}\right\}  +i\epsilon_{kk_{1}k_{2}}\left[
\Sigma_{k_{1},l,m}^{\lessgtr},A_{k_{2},o,m}\right] \\
+i\epsilon_{ll_{1}l_{2}}\left[  \Sigma_{k,l_{1},m}^{\lessgtr},A_{o,l_{2}%
,m}\right]  +\frac{1}{2}i\epsilon_{kk_{1}k_{2}}i\epsilon_{ll_{1}l_{2}}\left\{
\Sigma_{k_{1},l_{1},m}^{\lessgtr},A_{k_{2},l_{2},m}\right\}
\end{array}
\\
+%
\begin{array}
[c]{c}%
\left\{  \Sigma_{k,o,o}^{\lessgtr},A_{o,l,o}\right\}  +\left\{  \Sigma
_{k,o,m}^{\lessgtr},A_{o,l,m}\right\}  +i\epsilon_{kk_{1}k_{2}}\left[
\Sigma_{k_{1},o,m}^{\lessgtr},A_{k_{2},l,m}\right] \\
+i\epsilon_{ll_{1}l_{2}}\left[  \Sigma_{o,l_{1},m}^{\lessgtr},A_{k,l_{2}%
,m}\right]  +\frac{1}{2}i\epsilon_{kk_{1}k_{2}}i\epsilon_{ll_{1}l_{2}}\left\{
\Sigma_{k_{1},l_{1},m}^{\lessgtr},A_{k_{2},l_{2},m}\right\}
\end{array}
\end{array}
\label{QT5-4}%
\end{align}

We have for the entangled Dirac spin and pseudospin,%

\begin{align}
&  2^{N_{s}}i\hbar\left(  \frac{\partial}{\partial t_{1}}+\frac{\partial
}{\partial t_{2}}\right)  G_{k,o,m}^{\gtrless}=\nonumber\\
&  =\left[  \mathcal{\tilde{H}},G^{\lessgtr}\right]  _{k,o,m}+\left[
\Sigma^{\lessgtr},\mathrm{Re}G^{r}\right]  _{k,o,m}-\frac{i}{2}\left\{
\Gamma,G^{\lessgtr}\right\}  _{k,o,m}+\frac{i}{2}\left\{  \Sigma^{\lessgtr
},A\right\}  _{k,o,m} \label{QT6}%
\end{align}
where,%
\begin{align}
&  \left[  \mathcal{\tilde{H}},G^{\lessgtr}\right]  _{k,o,m}=\nonumber\\
&
\begin{array}
[c]{c}%
\begin{array}
[c]{c}%
\left[  H_{o,o,o},G_{k,o,m}^{\lessgtr}\right]  +\left[  H_{o,o,m}%
,G_{k,o,o}^{\lessgtr}\right]  +i\epsilon_{kk_{1}k_{2}}\left\{  H_{k_{1}%
,o,o},G_{k_{2},o,m}^{\lessgtr}\right\} \\
+i\epsilon_{mm_{1}m_{2}}\left\{  H_{o,o,m_{1}},G_{k,o,m_{2}}^{\lessgtr
}\right\}  +\frac{1}{2}i\epsilon_{kk_{1}k_{2}}i\epsilon_{mm_{1}m_{2}}\left[
H_{k_{1},o,m_{1}}G_{k_{2},o,m_{2}}^{\lessgtr}\right]
\end{array}
\\%
\begin{array}
[c]{c}%
+\left[  H_{o,l,o},G_{k,l,m}^{\lessgtr}\right]  +\left[  H_{o,l,m}%
,G_{k,l,o}^{\lessgtr}\right]  +i\epsilon_{kk_{1}k_{2}}\left\{  H_{k_{1}%
,l,o},G_{k_{2},l,m}^{\lessgtr}\right\} \\
+i\epsilon_{mm_{1}m_{2}}\left\{  H_{o,l,m_{1}},G_{k,l,m2}^{\lessgtr}\right\}
+\frac{1}{2}i\epsilon_{kk_{1}k_{2}}i\epsilon_{mm_{1}m_{2}}\left[
H_{k_{1},l,m_{1}},G_{k_{2},l,m_{2}}^{\lessgtr}\right]
\end{array}
\\%
\begin{array}
[c]{c}%
+\left[  H_{k,o,o},G_{o,o,m}^{\lessgtr}\right]  +\left[  H_{k,o,m}%
,G_{o,o,o}^{\lessgtr}\right]  +i\epsilon_{kk_{1}k_{2}}\left\{  H_{k_{1}%
,o,m},G_{k_{2},o,o}^{\lessgtr}\right\} \\
+i\epsilon_{mm_{1}m_{2}}\left\{  H_{k,o,m_{1}},G_{o,o,m_{2}}^{\lessgtr
}\right\}  +\frac{1}{2}i\epsilon_{kk_{1}k_{2}}i\epsilon_{mm_{1}m_{2}}\left[
H_{k_{1},o,m_{1}},G_{k_{2},o,m_{2}}^{\lessgtr}\right]
\end{array}
\\%
\begin{array}
[c]{c}%
\left[  H_{k,l,o},G_{o,l,m}^{\lessgtr}\right]  +\left[  H_{k.l,m}%
,G_{o,l,o}^{\lessgtr}\right]  +i\epsilon_{kk_{1}k_{2}}\left\{  H_{k_{1}%
,l,m},G_{k_{2},l,o}^{\lessgtr}\right\} \\
+i\epsilon_{mm_{1}m_{2}}\left\{  H_{k.l,m_{1}},G_{o,l,m_{2}}^{\lessgtr
}\right\}  +\frac{1}{2}i\epsilon_{kk_{1}k_{2}}i\epsilon_{mm_{1}m_{2}}\left[
H_{k_{1},o,m_{1}},G_{k_{2},o,m_{2}}^{\lessgtr}\right]
\end{array}
\end{array}
\nonumber\\
&  \label{QT6-1}%
\end{align}%
\begin{align}
&  \left[  \Sigma^{\lessgtr},\mathrm{Re}G^{r}\right]  _{k,o,m}=\nonumber\\
&
\begin{array}
[c]{c}%
\begin{array}
[c]{c}%
\left[  \Sigma_{o,o,o}^{\lessgtr},\mathrm{Re}G_{k,o,m}^{r}\right]  +\left[
\Sigma_{o,o,m}^{\lessgtr},\mathrm{Re}G_{k,o,o}^{r}\right]  +i\epsilon
_{kk_{1}k_{2}}\left\{  \Sigma_{k_{1},o,o}^{\lessgtr},\mathrm{Re}G_{k_{2}%
,o,m}^{r}\right\} \\
+i\epsilon_{mm_{1}m_{2}}\left\{  \Sigma_{o,o,m_{1}}^{\lessgtr},\mathrm{Re}%
G_{k,o,m_{2}}^{r}\right\}  +\frac{1}{2}i\epsilon_{kk_{1}k_{2}}i\epsilon
_{mm_{1}m_{2}}\left[  \Sigma_{k_{1},o,m_{1}}^{\lessgtr}\mathrm{Re}%
G_{k_{2},o,m_{2}}^{r}\right]
\end{array}
\\%
\begin{array}
[c]{c}%
\left[  \Sigma_{o,l,o}^{\lessgtr},\mathrm{Re}G_{k,l,m}^{r}\right]  +\left[
\Sigma_{o,l,m}^{\lessgtr},\mathrm{Re}G_{k,l,o}^{r}\right]  +i\epsilon
_{kk_{1}k_{2}}\left\{  \Sigma_{k_{1},l,o}^{\lessgtr},\mathrm{Re}G_{k_{2}%
,l,m}^{r}\right\} \\
+i\epsilon_{mm_{1}m_{2}}\left\{  \Sigma_{o,l,m_{1}}^{\lessgtr},\mathrm{Re}%
G_{k,l,m2}^{r}\right\}  +\frac{1}{2}i\epsilon_{kk_{1}k_{2}}i\epsilon
_{mm_{1}m_{2}}\left[  \Sigma_{k_{1},l,m_{1}}^{\lessgtr},\mathrm{Re}%
G_{k_{2},l,m_{2}}^{r}\right]
\end{array}
\\%
\begin{array}
[c]{c}%
+\left[  \Sigma_{k,o,o}^{\lessgtr},\mathrm{Re}G_{o,o,m}^{r}\right]  +\left[
\Sigma_{k,o,m}^{\lessgtr},\mathrm{Re}G_{o,o,o}^{r}\right]  +i\epsilon
_{kk_{1}k_{2}}\left\{  \Sigma_{k_{1},o,m}^{\lessgtr},\mathrm{Re}G_{k_{2}%
,o,o}^{r}\right\} \\
+i\epsilon_{mm_{1}m_{2}}\left\{  \Sigma_{k,o,m_{1}}^{\lessgtr},\mathrm{Re}%
G_{o,o,m_{2}}^{r}\right\}  +\frac{1}{2}i\epsilon_{kk_{1}k_{2}}i\epsilon
_{mm_{1}m_{2}}\left[  \Sigma_{k_{1},o,m_{1}}^{\lessgtr},\mathrm{Re}%
G_{k_{2},o,m_{2}}^{r}\right]
\end{array}
\\%
\begin{array}
[c]{c}%
\left[  \Sigma_{k,l,o}^{\lessgtr},\mathrm{Re}G_{o,l,m}^{r}\right]  +\left[
\Sigma_{k.l,m}^{\lessgtr},\mathrm{Re}G_{o,l,o}^{r}\right]  +i\epsilon
_{kk_{1}k_{2}}\left\{  \Sigma_{k_{1},l,m}^{\lessgtr},\mathrm{Re}G_{k_{2}%
,l,o}^{r}\right\} \\
+i\epsilon_{mm_{1}m_{2}}\left\{  \Sigma_{k.l,m_{1}}^{\lessgtr},\mathrm{Re}%
G_{o,l,m_{2}}^{r}\right\}  +\frac{1}{2}i\epsilon_{kk_{1}k_{2}}i\epsilon
_{mm_{1}m_{2}}\left[  \Sigma_{k_{1},o,m_{1}}^{\lessgtr},\mathrm{Re}%
G_{k_{2},o,m_{2}}^{r}\right]
\end{array}
\end{array}
\nonumber\\
&  \label{QT6-2}%
\end{align}%
\begin{align}
&  \left\{  \Gamma,G^{\lessgtr}\right\}  _{k,o,m}=\nonumber\\
&  =%
\begin{array}
[c]{c}%
\begin{array}
[c]{c}%
\left\{  \Gamma_{o,o,o},G_{k,o,m}^{\lessgtr}\right\}  +\left\{  \Gamma
_{o,o,m},G_{k,o,o}^{\lessgtr}\right\}  +i\epsilon_{kk_{1}k_{2}}\left[
\Gamma_{k_{1},o,o},G_{k_{2},o,m}^{\lessgtr}\right] \\
+i\epsilon_{mm_{1}m_{2}}\left[  \Gamma_{o,o,m_{1}},G_{k,o,m_{2}}^{\lessgtr
}\right]  +\frac{1}{2}i\epsilon_{kk_{1}k_{2}}i\epsilon_{mm_{1}m_{2}}\left\{
\Gamma_{k_{1},o,m_{1}}G_{k_{2},o,m_{2}}^{\lessgtr}\right\}
\end{array}
\\%
\begin{array}
[c]{c}%
+\left\{  \Gamma_{o,l,o},G_{k,l,m}^{\lessgtr}\right\}  +\left\{
\Gamma_{o,l,m},G_{k,l,o}^{\lessgtr}\right\}  +i\epsilon_{kk_{1}k_{2}}\left[
\Gamma_{k_{1},l,o},G_{k_{2},l,m}^{\lessgtr}\right] \\
+i\epsilon_{mm_{1}m_{2}}\left[  \Gamma_{o,l,m_{1}},G_{k,l,m2}^{\lessgtr
}\right]  +\frac{1}{2}i\epsilon_{kk_{1}k_{2}}i\epsilon_{mm_{1}m_{2}}\left\{
\Gamma_{k_{1},l,m_{1}},G_{k_{2},l,m_{2}}^{\lessgtr}\right\}
\end{array}
\\%
\begin{array}
[c]{c}%
+\left\{  \Gamma_{k,o,o},G_{o,o,m}^{\lessgtr}\right\}  +\left\{
\Gamma_{k,o,m},G_{o,o,o}^{\lessgtr}\right\}  +i\epsilon_{kk_{1}k_{2}}\left[
\Gamma_{k_{1},o,m},G_{k_{2},o,o}^{\lessgtr}\right] \\
+i\epsilon_{mm_{1}m_{2}}\left[  \Gamma_{k,o,m_{1}},G_{o,o,m_{2}}^{\lessgtr
}\right]  +\frac{1}{2}i\epsilon_{kk_{1}k_{2}}i\epsilon_{mm_{1}m_{2}}\left\{
\Gamma_{k_{1},o,m_{1}},G_{k_{2},o,m_{2}}^{\lessgtr}\right\}
\end{array}
\\%
\begin{array}
[c]{c}%
+\left\{  \Gamma_{k,l,o},G_{o,l,m}^{\lessgtr}\right\}  +\left\{
\Gamma_{k.l,m},G_{o,l,o}^{\lessgtr}\right\}  +i\epsilon_{kk_{1}k_{2}}\left[
\Gamma_{k_{1},l,m},G_{k_{2},l,o}^{\lessgtr}\right] \\
+i\epsilon_{mm_{1}m_{2}}\left[  \Gamma_{k.l,m_{1}},G_{o,l,m_{2}}^{\lessgtr
}\right]  +\frac{1}{2}i\epsilon_{kk_{1}k_{2}}i\epsilon_{mm_{1}m_{2}}\left\{
\Gamma_{k_{1},o,m_{1}},G_{k_{2},o,m_{2}}^{\lessgtr}\right\}
\end{array}
\end{array}
\nonumber\\
&  \label{QT6-3}%
\end{align}%
\begin{align}
&  \left\{  \Sigma^{\lessgtr},A\right\}  _{k,o,m}=\nonumber\\
&  =%
\begin{array}
[c]{c}%
\begin{array}
[c]{c}%
\left\{  \Sigma_{o,o,o}^{\lessgtr},A_{k,o,m}\right\}  +\left\{  \Sigma
_{o,o,m}^{\lessgtr},A_{k,o,o}\right\}  +i\epsilon_{kk_{1}k_{2}}\left[
\Sigma_{k_{1},o,o}^{\lessgtr},A_{k_{2},o,m}\right] \\
+i\epsilon_{mm_{1}m_{2}}\left[  \Sigma_{o,o,m_{1}}^{\lessgtr},A_{k,o,m_{2}%
}\right]  +\frac{1}{2}i\epsilon_{kk_{1}k_{2}}i\epsilon_{mm_{1}m_{2}}\left\{
\Sigma_{k_{1},o,m_{1}}^{\lessgtr}A_{k_{2},o,m_{2}}\right\}
\end{array}
\\%
\begin{array}
[c]{c}%
+\left\{  \Sigma_{o,l,o}^{\lessgtr},A_{k,l,m}\right\}  +\left\{
\Sigma_{o,l,m}^{\lessgtr},A_{k,l,o}\right\}  +i\epsilon_{kk_{1}k_{2}}\left[
\Sigma_{k_{1},l,o}^{\lessgtr},A_{k_{2},l,m}\right] \\
+i\epsilon_{mm_{1}m_{2}}\left[  \Sigma_{o,l,m_{1}}^{\lessgtr},A_{k,l,m2}%
\right]  +\frac{1}{2}i\epsilon_{kk_{1}k_{2}}i\epsilon_{mm_{1}m_{2}}\left\{
\Sigma_{k_{1},l,m_{1}}^{\lessgtr},A_{k_{2},l,m_{2}}\right\}
\end{array}
\\%
\begin{array}
[c]{c}%
+\left\{  \Sigma_{k,o,o}^{\lessgtr},A_{o,o,m}\right\}  +\left\{
\Sigma_{k,o,m}^{\lessgtr},A_{o,o,o}\right\}  +i\epsilon_{kk_{1}k_{2}}\left[
\Sigma_{k_{1},o,m}^{\lessgtr},A_{k_{2},o,o}\right] \\
+i\epsilon_{mm_{1}m_{2}}\left[  \Sigma_{k,o,m_{1}}^{\lessgtr},A_{o,o,m_{2}%
}\right]  +\frac{1}{2}i\epsilon_{kk_{1}k_{2}}i\epsilon_{mm_{1}m_{2}}\left\{
\Sigma_{k_{1},o,m_{1}}^{\lessgtr},A_{k_{2},o,m_{2}}\right\}
\end{array}
\\%
\begin{array}
[c]{c}%
+\left\{  \Sigma_{k,l,o}^{\lessgtr},A_{o,l,m}\right\}  +\left\{
\Sigma_{k.l,m}^{\lessgtr},A_{o,l,o}\right\}  +i\epsilon_{kk_{1}k_{2}}\left[
\Sigma_{k_{1},l,m}^{\lessgtr},A_{k_{2},l,o}\right] \\
+i\epsilon_{mm_{1}m_{2}}\left[  \Sigma_{k.l,m_{1}}^{\lessgtr},A_{o,l,m_{2}%
}\right]  +\frac{1}{2}i\epsilon_{kk_{1}k_{2}}i\epsilon_{mm_{1}m_{2}}\left\{
\Sigma_{k_{1},o,m_{1}}^{\lessgtr},A_{k_{2},o,m_{2}}\right\}
\end{array}
\end{array}
\nonumber\\
&  \label{QT6-4}%
\end{align}

We have for the entangled valley spin and pseudospin,%

\begin{align}
&  2^{N_{s}}i\hbar\left(  \frac{\partial}{\partial t_{1}}+\frac{\partial
}{\partial t_{2}}\right)  G_{o,l,m}^{\gtrless}=\nonumber\\
&  =\left[  \mathcal{\tilde{H}},G^{\lessgtr}\right]  _{o,l,m}+\left[
\Sigma^{\lessgtr},\mathrm{Re}G^{r}\right]  _{o,l,m}-\frac{i}{2}\left\{
\Gamma,G^{\lessgtr}\right\}  _{o,l,m}+\frac{i}{2}\left\{  \Sigma^{\lessgtr
},A\right\}  _{o,l,m} \label{QT7}%
\end{align}
where,%
\begin{align}
&  \left[  \mathcal{\tilde{H}},G^{\lessgtr}\right]  _{o,l,m}=\nonumber\\
&  =%
\begin{array}
[c]{c}%
\begin{array}
[c]{c}%
\left[  H_{o,o,o},G_{o,l,m}^{\lessgtr}\right]  +\left[  H_{o,o,m}%
,G_{o,l,o}^{\lessgtr}\right]  +i\epsilon_{mm_{1}m_{2}}\left\{  H_{o,o,m_{1}%
},G_{o,l,m_{2}}^{\lessgtr}\right\} \\
+i\epsilon_{ll_{1}l_{2}}\left\{  H_{o,l_{1},o},G_{o,l_{2},m}^{\lessgtr
}\right\}  +\frac{1}{2}i\epsilon_{ll_{1}l_{2}}i\epsilon_{mm_{1}m_{2}}\left[
H_{o,l_{1},m_{1}},G_{o,l_{2},m_{2}}^{\lessgtr}\right]
\end{array}
\\%
\begin{array}
[c]{c}%
+\left[  H_{o,l,o},G_{o,o,m}^{\lessgtr}\right]  +\left[  H_{o,l,m}%
,G_{o,o,o}^{\lessgtr}\right]  +i\epsilon_{mm_{1}m_{2}}\left\{  H_{o,l,m_{1}%
},G_{o,o,m_{2}}^{\lessgtr}\right\} \\
+i\epsilon_{ll_{1}l_{2}}\left\{  H_{o,l_{1},m},G_{o,l_{2},o}^{\lessgtr
}\right\}  +\frac{1}{2}i\epsilon_{ll_{1}l_{2}}i\epsilon_{mm_{1}m_{2}}\left[
H_{o,l_{1},m_{1}},G_{o,l_{2},m_{2}}^{\lessgtr}\right]
\end{array}
\\%
\begin{array}
[c]{c}%
+\left[  H_{k,l,o},G_{k,o,m}^{\lessgtr}\right]  +\left[  H_{k.l,m}%
,G_{k,o,o}^{\lessgtr}\right]  +i\epsilon_{mm_{1}m_{2}}\left\{  H_{k.l,m_{1}%
},G_{k,o,m_{2}}^{\lessgtr}\right\} \\
+i\epsilon_{ll_{1}l_{2}}\left\{  H_{k,l_{1},m},G_{k,l_{2},o}^{\lessgtr
}\right\}  +\frac{1}{2}i\epsilon_{ll_{1}l_{2}}i\epsilon_{mm_{1}m_{2}}\left[
H_{k,l_{1},m_{1}},G_{k,l_{2},m_{2}}^{\lessgtr}\right]
\end{array}
\\%
\begin{array}
[c]{c}%
+\left[  H_{k,o,o},G_{k,l,m}^{\lessgtr}\right]  +\left[  H_{k,o,m}%
,G_{k,l,o}^{\lessgtr}\right]  +i\epsilon_{mm_{1}m_{2}}\left\{  H_{k,o,m_{1}%
},G_{k,l,m_{2}}^{\lessgtr}\right\} \\
+i\epsilon_{ll_{1}l_{2}}\left\{  H_{k,l_{1},o},G_{k,l_{2},m}^{\lessgtr
}\right\}  +\frac{1}{2}i\epsilon_{ll_{1}l_{2}}i\epsilon_{mm_{1}m_{2}}\left[
H_{k,l_{1},m_{1}},G_{k,l_{2},m_{2}}^{\lessgtr}\right]
\end{array}
\end{array}
\nonumber\\
&  \label{QT7-1}%
\end{align}%
\begin{align}
&  \left[  \Sigma^{\lessgtr},\mathrm{Re}G^{r}\right]  _{o,l,m}=\nonumber\\
&  =%
\begin{array}
[c]{c}%
\begin{array}
[c]{c}%
\begin{array}
[c]{c}%
\left[  \Sigma_{o,o,o}^{\lessgtr},\mathrm{Re}G_{o,l,m}^{r}\right]  +\left[
\Sigma_{o,o,m}^{\lessgtr},\mathrm{Re}G_{o,l,o}^{r}\right]  +i\epsilon
_{mm_{1}m_{2}}\left\{  \Sigma_{o,o,m_{1}}^{\lessgtr},\mathrm{Re}G_{o,l,m_{2}%
}^{r}\right\} \\
+i\epsilon_{ll_{1}l_{2}}\left\{  \Sigma_{o,l_{1},o}^{\lessgtr},\mathrm{Re}%
G_{o,l_{2},m}^{r}\right\}  +\frac{1}{2}i\epsilon_{ll_{1}l_{2}}i\epsilon
_{mm_{1}m_{2}}\left[  \Sigma_{o,l_{1},m_{1}}^{\lessgtr},\mathrm{Re}%
G_{o,l_{2},m_{2}}^{r}\right]
\end{array}
\\%
\begin{array}
[c]{c}%
+\left[  \Sigma_{o,l,o}^{\lessgtr},\mathrm{Re}G_{o,o,m}^{r}\right]  +\left[
\Sigma_{o,l,m}^{\lessgtr},\mathrm{Re}G_{o,o,o}^{r}\right]  +i\epsilon
_{mm_{1}m_{2}}\left\{  \Sigma_{o,l,m_{1}}^{\lessgtr},\mathrm{Re}G_{o,o,m_{2}%
}^{r}\right\} \\
+i\epsilon_{ll_{1}l_{2}}\left\{  \Sigma_{o,l_{1},m}^{\lessgtr},\mathrm{Re}%
G_{o,l_{2},o}^{r}\right\}  +\frac{1}{2}i\epsilon_{ll_{1}l_{2}}i\epsilon
_{mm_{1}m_{2}}\left[  \Sigma_{o,l_{1},m_{1}}^{\lessgtr},\mathrm{Re}%
G_{o,l_{2},m_{2}}^{r}\right]
\end{array}
\end{array}
\\%
\begin{array}
[c]{c}%
+\left[  \Sigma_{k,l,o}^{\lessgtr},\mathrm{Re}G_{k,o,m}^{r}\right]  +\left[
\Sigma_{k.l,m}^{\lessgtr},\mathrm{Re}G_{k,o,o}^{r}\right]  +i\epsilon
_{mm_{1}m_{2}}\left\{  \Sigma_{k.l,m_{1}}^{\lessgtr},\mathrm{Re}G_{k,o,m_{2}%
}^{r}\right\} \\
+i\epsilon_{ll_{1}l_{2}}\left\{  \Sigma_{k,l_{1},m}^{\lessgtr},\mathrm{Re}%
G_{k,l_{2},o}^{r}\right\}  +\frac{1}{2}i\epsilon_{ll_{1}l_{2}}i\epsilon
_{mm_{1}m_{2}}\left[  \Sigma_{k,l_{1},m_{1}}^{\lessgtr},\mathrm{Re}%
G_{k,l_{2},m_{2}}^{r}\right]
\end{array}
\\%
\begin{array}
[c]{c}%
+\left[  \Sigma_{k,o,o}^{\lessgtr},\mathrm{Re}G_{k,l,m}^{r}\right]  +\left[
\Sigma_{k,o,m}^{\lessgtr},\mathrm{Re}G_{k,l,o}^{r}\right]  +i\epsilon
_{mm_{1}m_{2}}\left\{  \Sigma_{k,o,m_{1}}^{\lessgtr},\mathrm{Re}G_{k,l,m_{2}%
}^{r}\right\} \\
+i\epsilon_{ll_{1}l_{2}}\left\{  \Sigma_{k,l_{1},o}^{\lessgtr},\mathrm{Re}%
G_{k,l_{2},m}^{r}\right\}  +\frac{1}{2}i\epsilon_{ll_{1}l_{2}}i\epsilon
_{mm_{1}m_{2}}\left[  \Sigma_{k,l_{1},m_{1}}^{\lessgtr},\mathrm{Re}%
G_{k,l_{2},m_{2}}^{r}\right]
\end{array}
\end{array}
\nonumber\\
&  \label{QT7-2}%
\end{align}%
\begin{align}
&  \left\{  \Gamma,G^{\lessgtr}\right\}  _{o,l,m}=\nonumber\\
&  =%
\begin{array}
[c]{c}%
\begin{array}
[c]{c}%
\left\{  \Gamma_{o,o,o},G_{o,l,m}^{\lessgtr}\right\}  +\left\{  \Gamma
_{o,o,m},G_{o,l,o}^{\lessgtr}\right\}  +i\epsilon_{mm_{1}m_{2}}\left[
\Gamma_{o,o,m_{1}},G_{o,l,m_{2}}^{\lessgtr}\right] \\
+i\epsilon_{ll_{1}l_{2}}\left[  \Gamma_{o,l_{1},o},G_{o,l_{2},m}^{\lessgtr
}\right]  +\frac{1}{2}i\epsilon_{ll_{1}l_{2}}i\epsilon_{mm_{1}m_{2}}\left\{
\Gamma_{o,l_{1},m_{1}},G_{o,l_{2},m_{2}}^{\lessgtr}\right\}
\end{array}
\\%
\begin{array}
[c]{c}%
+\left\{  \Gamma_{o,l,o},G_{o,o,m}^{\lessgtr}\right\}  +\left\{
\Gamma_{o,l,m},G_{o,o,o}^{\lessgtr}\right\}  +i\epsilon_{mm_{1}m_{2}}\left[
\Gamma_{o,l,m_{1}},G_{o,o,m_{2}}^{\lessgtr}\right] \\
+i\epsilon_{ll_{1}l_{2}}\left[  \Gamma_{o,l_{1},m},G_{o,l_{2},o}^{\lessgtr
}\right]  +\frac{1}{2}i\epsilon_{ll_{1}l_{2}}i\epsilon_{mm_{1}m_{2}}\left\{
\Gamma_{o,l_{1},m_{1}},G_{o,l_{2},m_{2}}^{\lessgtr}\right\}
\end{array}
\\%
\begin{array}
[c]{c}%
+\left\{  \Gamma_{k,l,o},G_{k,o,m}^{\lessgtr}\right\}  +\left\{
\Gamma_{k.l,m},G_{k,o,o}^{\lessgtr}\right\}  +i\epsilon_{mm_{1}m_{2}}\left[
\Gamma_{k.l,m_{1}},G_{k,o,m_{2}}^{\lessgtr}\right] \\
+i\epsilon_{ll_{1}l_{2}}\left[  \Gamma_{k,l_{1},m},G_{k,l_{2},o}^{\lessgtr
}\right]  +\frac{1}{2}i\epsilon_{ll_{1}l_{2}}i\epsilon_{mm_{1}m_{2}}\left\{
\Gamma_{k,l_{1},m_{1}},G_{k,l_{2},m_{2}}^{\lessgtr}\right\}
\end{array}
\\%
\begin{array}
[c]{c}%
+\left\{  \Gamma_{k,o,o},G_{k,l,m}^{\lessgtr}\right\}  +\left\{
\Gamma_{k,o,m},G_{k,l,o}^{\lessgtr}\right\}  +i\epsilon_{mm_{1}m_{2}}\left[
\Gamma_{k,o,m_{1}},G_{k,l,m_{2}}^{\lessgtr}\right] \\
+i\epsilon_{ll_{1}l_{2}}\left[  \Gamma_{k,l_{1},o},G_{k,l_{2},m}^{\lessgtr
}\right]  +\frac{1}{2}i\epsilon_{ll_{1}l_{2}}i\epsilon_{mm_{1}m_{2}}\left\{
\Gamma_{k,l_{1},m_{1}},G_{k,l_{2},m_{2}}^{\lessgtr}\right\}
\end{array}
\end{array}
\nonumber\\
&  \label{QT7-3}%
\end{align}%
\begin{align}
&  \left\{  \Sigma^{\lessgtr},A\right\}  _{o,l,m}=\nonumber\\
&  =%
\begin{array}
[c]{c}%
\begin{array}
[c]{c}%
\left\{  \Sigma_{o,o,o}^{\lessgtr},A_{o,l,m}\right\}  +\left\{  \Sigma
_{o,o,m}^{\lessgtr},A_{o,l,o}\right\}  +i\epsilon_{mm_{1}m_{2}}\left[
\Sigma_{o,o,m_{1}}^{\lessgtr},A_{o,l,m_{2}}\right] \\
+i\epsilon_{ll_{1}l_{2}}\left[  \Sigma_{o,l_{1},o}^{\lessgtr},A_{o,l_{2}%
,m}\right]  +\frac{1}{2}i\epsilon_{ll_{1}l_{2}}i\epsilon_{mm_{1}m_{2}}\left\{
\Sigma_{o,l_{1},m_{1}}^{\lessgtr},A_{o,l_{2},m_{2}}\right\}
\end{array}
\\%
\begin{array}
[c]{c}%
+\left[  \Sigma_{o,l,o}^{\lessgtr},A_{o,o,m}\right]  +\left[  \Sigma
_{o,l,m}^{\lessgtr},A_{o,o,o}\right]  +i\epsilon_{mm_{1}m_{2}}\left\{
\Sigma_{o,l,m_{1}}^{\lessgtr},A_{o,o,m_{2}}\right\} \\
+i\epsilon_{ll_{1}l_{2}}\left\{  \Sigma_{o,l_{1},m}^{\lessgtr},A_{o,l_{2}%
,o}\right\}  +\frac{1}{2}i\epsilon_{ll_{1}l_{2}}i\epsilon_{mm_{1}m_{2}}\left[
\Sigma_{o,l_{1},m_{1}}^{\lessgtr},A_{o,l_{2},m_{2}}\right]
\end{array}
\\%
\begin{array}
[c]{c}%
+\left\{  \Sigma_{k,l,o}^{\lessgtr},A_{k,o,m}\right\}  +\left\{
\Sigma_{k.l,m}^{\lessgtr},A_{k,o,o}\right\}  +i\epsilon_{mm_{1}m_{2}}\left[
\Sigma_{k.l,m_{1}}^{\lessgtr},A_{k,o,m_{2}}\right] \\
+i\epsilon_{ll_{1}l_{2}}\left[  \Sigma_{k,l_{1},m}^{\lessgtr},A_{k,l_{2}%
,o}\right]  +\frac{1}{2}i\epsilon_{ll_{1}l_{2}}i\epsilon_{mm_{1}m_{2}}\left\{
\Sigma_{k,l_{1},m_{1}}^{\lessgtr},A_{k,l_{2},m_{2}}\right\}
\end{array}
\\%
\begin{array}
[c]{c}%
+\left[  \Sigma_{k,o,o}^{\lessgtr},A_{k,l,m}\right]  +\left[  \Sigma
_{k,o,m}^{\lessgtr},A_{k,l,o}\right]  +i\epsilon_{mm_{1}m_{2}}\left\{
\Sigma_{k,o,m_{1}}^{\lessgtr},A_{k,l,m_{2}}\right\} \\
+i\epsilon_{ll_{1}l_{2}}\left\{  \Sigma_{k,l_{1},o}^{\lessgtr},A_{k,l_{2}%
,m}\right\}  +\frac{1}{2}i\epsilon_{ll_{1}l_{2}}i\epsilon_{mm_{1}m_{2}}\left[
\Sigma_{k,l_{1},m_{1}}^{\lessgtr},A_{k,l_{2},m_{2}}\right]
\end{array}
\end{array}
\nonumber\\
&  \label{QT7-4}%
\end{align}

We have for the entangled Dirac spin, valley spin, and pseudospin,%

\begin{align}
&  2^{N_{s}}i\hbar\left(  \frac{\partial}{\partial t_{1}}+\frac{\partial
}{\partial t_{2}}\right)  G_{k,l,m}^{\gtrless}=\nonumber\\
&  =\left[  \mathcal{\tilde{H}},G^{\lessgtr}\right]  _{k,l,m}+\left[
\Sigma^{\lessgtr},\mathrm{Re}G^{r}\right]  _{k,l,m}-\frac{i}{2}\left\{
\Gamma,G^{\lessgtr}\right\}  _{k,l,m}+\frac{i}{2}\left\{  \Sigma^{\lessgtr
},A\right\}  _{k,l,m} \label{QT8}%
\end{align}
where,%
\begin{align}
&  \left[  \mathcal{\tilde{H}},G^{\lessgtr}\right]  _{k,l,m}=\nonumber\\
&  =\left[
\begin{array}
[c]{c}%
\left[  H_{k,l,m}\ ,G_{0,0,0}^{\lessgtr}\right]  +\left[  H_{k,l,0}%
\ ,G_{0,0,m}^{\lessgtr}\right]  +\left[  H_{k,0,m}\ ,G_{0,l,0}^{\lessgtr
}\right]  +\left[  H_{0,l,m}\ ,G_{k,0,0}^{\lessgtr}\right] \\
+\left[  H_{k,0,0}\ ,G_{0,l,m}^{\lessgtr}\right]  +\left[  H_{0,0,m}%
\ ,G_{k,l,0}^{\lessgtr}\right]  +\left[  H_{0,l,0}\ ,G_{k,0,m}^{\lessgtr
}\right]  +\left[  H_{o,o,o}\ ,G_{k.l.m}^{\lessgtr}\right]
\end{array}
\right] \nonumber\\
&  +\left[
\begin{array}
[c]{c}%
\epsilon_{kk_{1}k_{2}}\left\{  H_{k_{1},0,0}\ ,G_{k_{2},l,m}^{\lessgtr
}\right\}  +\epsilon_{kk_{1}k_{2}}\left\{  H_{k_{1},l,0}\ ,G_{k_{2}%
,0,m}^{\lessgtr}\right\} \\
+\epsilon_{kk_{1}k_{2}}\left\{  H_{k_{1},0,m}\ ,G_{k_{2},l,0}^{\lessgtr
}\right\}  +\epsilon_{kk_{1}k_{2}}\left\{  H_{k_{1},l,m}\ ,G_{k_{2}%
,0,0}^{\lessgtr}\right\} \\
+\epsilon_{ll_{1}l_{2}}\left\{  H_{0,l_{1},0}\ ,G_{k,l_{2},m}^{\lessgtr
}\right\}  +\epsilon_{ll_{1}l_{2}}\left\{  H_{k,l_{1},0}\ ,G_{0,l_{2}%
,m}^{\lessgtr}\right\} \\
+\epsilon_{ll_{1}l_{2}}\left\{  H_{0,l_{1},m}\ ,G_{k,l_{2},0}^{\lessgtr
}\right\}  +\epsilon_{ll_{1}l_{2}}\left\{  H_{k,l_{1},m}\ ,G_{0,l_{2}%
,0}^{\lessgtr}\right\} \\
+\epsilon_{m;m_{1},m_{2}}\left\{  H_{0,0.m_{1}}\ ,G_{k,l,m_{2}}^{\lessgtr
}\right\}  +\epsilon_{m;m_{1},m_{2}}\left\{  H_{0,l.m_{1}}\ ,G_{k,0,m_{2}%
}^{\lessgtr}\right\} \\
+\epsilon_{m;m_{1},m_{2}}\left\{  H_{k,0.m_{1}}\ ,G_{0,l,m_{2}}^{\lessgtr
}\right\}  +\epsilon_{m;m_{1},m_{2}}\left\{  H_{k,l.m_{1}}\ ,G_{0,0,m_{2}%
}^{\lessgtr}\right\}
\end{array}
\right] \nonumber\\
&  +\left[
\begin{array}
[c]{c}%
i\epsilon_{kk_{1}k_{2}}i\epsilon_{ll_{1}l_{2}}\left(  \left[  H_{k_{1}%
,l_{1},o},G_{k_{2},l_{2},m}^{\lessgtr}\right]  +\left[  H_{k_{1},l_{1}%
,m},G_{k_{2},l_{2},o}^{\lessgtr}\right]  \right) \\
+i\epsilon_{kk_{1}k_{2}}i\epsilon_{mm_{1}m_{2}}\left(  \left[  H_{k_{1}%
,o,m_{1}},G_{k_{2},l,m_{2}}^{\lessgtr}\right]  +\left[  H_{k_{1},l,m_{1}%
},G_{k_{2},o,m_{2}}^{\lessgtr}\right]  \right) \\
+i\epsilon_{ll_{1}l_{2}}i\epsilon_{mm_{1}m_{2}}\left(  \left[  H_{o,l_{1}%
,m_{1}},G_{k,l_{2},m_{2}}^{\lessgtr}\right]  +\left[  H_{k,l_{1},m_{1}%
},G_{o,l_{2},m_{2}}^{\lessgtr}\right]  \right) \\
+i\epsilon_{kk_{1}k_{2}}i\epsilon_{ll_{1}l_{2}}i\epsilon_{mm_{1}m_{2}}\left\{
H_{k_{1},l_{1},m_{1}},G_{k_{2},l_{2},m_{2}}^{\lessgtr}\right\}
\end{array}
\right]  \label{QT8-1}%
\end{align}%
\begin{align}
&  \left[  \Sigma^{\lessgtr},\mathrm{Re}G^{r}\right]  _{k,l,m}=\nonumber\\
&  =\left[
\begin{array}
[c]{c}%
\left[  \Sigma_{k,l,m}^{\lessgtr}\ ,\mathrm{Re}G_{0,0,0}^{r}\right]  +\left[
\Sigma_{k,l,0}^{\lessgtr}\ ,\mathrm{Re}G_{0,0,m}^{r}\right]  +\left[
\Sigma_{k,0,m}^{\lessgtr}\ ,\mathrm{Re}G_{0,l,0}^{r}\right]  +\left[
\Sigma_{0,l,m}^{\lessgtr}\ ,\mathrm{Re}G_{k,0,0}^{r}\right] \\
+\left[  \Sigma_{k,0,0}^{\lessgtr}\ ,\mathrm{Re}G_{0,l,m}^{r}\right]  +\left[
\Sigma_{0,0,m}^{\lessgtr}\ ,\mathrm{Re}G_{k,l,0}^{r}\right]  +\left[
\Sigma_{0,l,0}^{\lessgtr}\ ,\mathrm{Re}G_{k,0,m}^{r}\right]  +\left[
\Sigma_{o,o,o}^{\lessgtr}\ ,\mathrm{Re}G_{k.l.m}^{r}\right]
\end{array}
\right] \nonumber\\
&  +\left[
\begin{array}
[c]{c}%
\epsilon_{kk_{1}k_{2}}\left\{  \Sigma_{k_{1},0,0}^{\lessgtr}\ ,\mathrm{Re}%
G_{k_{2},l,m}^{r}\right\}  +\epsilon_{kk_{1}k_{2}}\left\{  \Sigma_{k_{1}%
,l,0}^{\lessgtr}\ ,\mathrm{Re}G_{k_{2},0,m}^{r}\right\} \\
+\epsilon_{kk_{1}k_{2}}\left\{  \Sigma_{k_{1},0,m}^{\lessgtr}\ ,\mathrm{Re}%
G_{k_{2},l,0}^{r}\right\}  +\epsilon_{kk_{1}k_{2}}\left\{  \Sigma_{k_{1}%
,l,m}^{\lessgtr}\ ,\mathrm{Re}G_{k_{2},0,0}^{r}\right\} \\
+\epsilon_{ll_{1}l_{2}}\left\{  \Sigma_{0,l_{1},0}^{\lessgtr}\ ,\mathrm{Re}%
G_{k,l_{2},m}^{r}\right\}  +\epsilon_{ll_{1}l_{2}}\left\{  \Sigma_{k,l_{1}%
,0}^{\lessgtr}\ ,\mathrm{Re}G_{0,l_{2},m}^{r}\right\} \\
+\epsilon_{ll_{1}l_{2}}\left\{  \Sigma_{0,l_{1},m}^{\lessgtr}\ ,\mathrm{Re}%
G_{k,l_{2},0}^{r}\right\}  +\epsilon_{ll_{1}l_{2}}\left\{  \Sigma_{k,l_{1}%
,m}^{\lessgtr}\ ,\mathrm{Re}G_{0,l_{2},0}^{r}\right\} \\
+\epsilon_{m;m_{1},m_{2}}\left\{  \Sigma_{0,0.m_{1}}^{\lessgtr}\ ,\mathrm{Re}%
G_{k,l,m_{2}}^{r}\right\}  +\epsilon_{m;m_{1},m_{2}}\left\{  \Sigma
_{0,l.m_{1}}^{\lessgtr}\ ,\mathrm{Re}G_{k,0,m_{2}}^{r}\right\} \\
+\epsilon_{m;m_{1},m_{2}}\left\{  \Sigma_{k,0.m_{1}}^{\lessgtr}\ ,\mathrm{Re}%
G_{0,l,m_{2}}^{r}\right\}  +\epsilon_{m;m_{1},m_{2}}\left\{  \Sigma
_{k,l.m_{1}}^{\lessgtr}\ ,\mathrm{Re}G_{0,0,m_{2}}^{r}\right\}
\end{array}
\right] \nonumber\\
&  +\left[
\begin{array}
[c]{c}%
i\epsilon_{kk_{1}k_{2}}i\epsilon_{ll_{1}l_{2}}\left(  \left[  \Sigma
_{k_{1},l_{1},o}^{\lessgtr}\ ,\mathrm{Re}G_{k_{2},l_{2},m}^{r}\right]
+\left[  \Sigma_{k_{1},l_{1},m}^{\lessgtr}\ ,\mathrm{Re}G_{k_{2},l_{2},o}%
^{r}\right]  \right) \\
+i\epsilon_{kk_{1}k_{2}}i\epsilon_{mm_{1}m_{2}}\left(  \left[  \Sigma
_{k_{1},o,m_{1}}^{\lessgtr}\ ,\mathrm{Re}G_{k_{2},l,m_{2}}^{r}\right]
+\left[  \Sigma_{k_{1},l,m_{1}}^{\lessgtr}\ ,\mathrm{Re}G_{k_{2},o,m_{2}}%
^{r}\right]  \right) \\
+i\epsilon_{ll_{1}l_{2}}i\epsilon_{mm_{1}m_{2}}\left(  \left[  \Sigma
_{o,l_{1},m_{1}}^{\lessgtr}\ ,\mathrm{Re}G_{k,l_{2},m_{2}}^{r}\right]
+\left[  \Sigma_{k,l_{1},m_{1}}^{\lessgtr}\ ,\mathrm{Re}G_{o,l_{2},m_{2}}%
^{r}\right]  \right) \\
+i\epsilon_{kk_{1}k_{2}}i\epsilon_{ll_{1}l_{2}}i\epsilon_{mm_{1}m_{2}}\left\{
\Sigma_{k_{1},l_{1},m_{1}}^{\lessgtr}\ ,\mathrm{Re}G_{k_{2},l_{2},m_{2}}%
^{r}\right\}
\end{array}
\right]  \label{QT8-2}%
\end{align}%
\begin{align}
&  \left\{  \Gamma,G^{\lessgtr}\right\}  _{k,l,m}=\nonumber\\
&  =\left[
\begin{array}
[c]{c}%
\left\{  \Gamma_{k,l,m}\ ,G_{0,0,0}^{\lessgtr}\right\}  +\left\{
\Gamma_{k,l,0}\ ,G_{0,0,m}^{\lessgtr}\right\}  +\left\{  \Gamma_{k,0,m}%
\ ,G_{0,l,0}^{\lessgtr}\right\}  +\left\{  \Gamma_{0,l,m}\ ,G_{k,0,0}%
^{\lessgtr}\right\} \\
+\left\{  \Gamma_{k,0,0}\ ,G_{0,l,m}^{\lessgtr}\right\}  +\left\{
\Gamma_{0,0,m}\ ,G_{k,l,0}^{\lessgtr}\right\}  +\left\{  \Gamma_{0,l,0}%
\ ,G_{k,0,m}^{\lessgtr}\right\}  +\left\{  \Gamma_{o,o,o}\ ,G_{k.l.m}%
^{\lessgtr}\right\}
\end{array}
\right] \nonumber\\
&  +\left[
\begin{array}
[c]{c}%
\epsilon_{kk_{1}k_{2}}\left[  \Gamma_{k_{1},0,0}\ ,G_{k_{2},l,m}^{\lessgtr
}\right]  +\epsilon_{kk_{1}k_{2}}\left[  \Gamma_{k_{1},l,0}\ ,G_{k_{2}%
,0,m}^{\lessgtr}\right] \\
+\epsilon_{kk_{1}k_{2}}\left[  \Gamma_{k_{1},0,m}\ ,G_{k_{2},l,0}^{\lessgtr
}\right]  +\epsilon_{kk_{1}k_{2}}\left[  \Gamma_{k_{1},l,m}\ ,G_{k_{2}%
,0,0}^{\lessgtr}\right] \\
+\epsilon_{ll_{1}l_{2}}\left[  \Gamma_{0,l_{1},0}\ ,G_{k,l_{2},m}^{\lessgtr
}\right]  +\epsilon_{ll_{1}l_{2}}\left[  \Gamma_{k,l_{1},0}\ ,G_{0,l_{2}%
,m}^{\lessgtr}\right] \\
+\epsilon_{ll_{1}l_{2}}\left[  \Gamma_{0,l_{1},m}\ ,G_{k,l_{2},0}^{\lessgtr
}\right]  +\epsilon_{ll_{1}l_{2}}\left[  \Gamma_{k,l_{1},m}\ ,G_{0,l_{2}%
,0}^{\lessgtr}\right] \\
+\epsilon_{m;m_{1},m_{2}}\left[  \Gamma_{0,0.m_{1}}\ ,G_{k,l,m_{2}}^{\lessgtr
}\right]  +\epsilon_{m;m_{1},m_{2}}\left[  \Gamma_{0,l.m_{1}}\ ,G_{k,0,m_{2}%
}^{\lessgtr}\right] \\
+\epsilon_{m;m_{1},m_{2}}\left[  \Gamma_{k,0.m_{1}}\ ,G_{0,l,m_{2}}^{\lessgtr
}\right]  +\epsilon_{m;m_{1},m_{2}}\left[  \Gamma_{k,l.m_{1}}\ ,G_{0,0,m_{2}%
}^{\lessgtr}\right]
\end{array}
\right] \nonumber\\
&  +\left[
\begin{array}
[c]{c}%
i\epsilon_{kk_{1}k_{2}}i\epsilon_{ll_{1}l_{2}}\left(  \left\{  \Gamma
_{k_{1},l_{1},o}\ ,G_{k_{2},l_{2},m}^{\lessgtr}\right\}  +\left\{
\Gamma_{k_{1},l_{1},m}\ ,G_{k_{2},l_{2},o}^{\lessgtr}\right\}  \right) \\
+i\epsilon_{kk_{1}k_{2}}i\epsilon_{mm_{1}m_{2}}\left(  \left\{  \Gamma
_{k_{1},o,m_{1}}\ ,G_{k_{2},l,m_{2}}^{\lessgtr}\right\}  +\left\{
\Gamma_{k_{1},l,m_{1}}\ ,G_{k_{2},o,m_{2}}^{\lessgtr}\right\}  \right) \\
+i\epsilon_{ll_{1}l_{2}}i\epsilon_{mm_{1}m_{2}}\left(  \left\{  \Gamma
_{o,l_{1},m_{1}}\ ,G_{k,l_{2},m_{2}}^{\lessgtr}\right\}  +\left\{
\Gamma_{k,l_{1},m_{1}}\ ,G_{o,l_{2},m_{2}}^{\lessgtr}\right\}  \right) \\
+i\epsilon_{kk_{1}k_{2}}i\epsilon_{ll_{1}l_{2}}i\epsilon_{mm_{1}m_{2}}\left[
\Gamma_{k_{1},l_{1},m_{1}}\ ,G_{k_{2},l_{2},m_{2}}^{\lessgtr}\right]
\end{array}
\right]  \label{QT8-3}%
\end{align}%
\begin{align}
&  \left\{  \Sigma^{\lessgtr},A\right\}  _{k,l,m}=\nonumber\\
&  =\left[
\begin{array}
[c]{c}%
\left\{  \Sigma_{k,l,m}^{\lessgtr}\ ,A_{0,0,0}\right\}  +\left\{
\Sigma_{k,l,0}^{\lessgtr}\ ,A_{0,0,m}\right\}  +\left\{  \Sigma_{k,0,m}%
^{\lessgtr}\ ,A_{0,l,0}\right\}  +\left\{  \Sigma_{0,l,m}^{\lessgtr
}\ ,A_{k,0,0}\right\} \\
+\left\{  \Sigma_{k,0,0}^{\lessgtr}\ ,A_{0,l,m}\right\}  +\left\{
\Sigma_{0,0,m}^{\lessgtr}\ ,A_{k,l,0}\right\}  +\left\{  \Sigma_{0,l,0}%
^{\lessgtr}\ ,A_{k,0,m}\right\}  +\left\{  \Sigma_{o,o,o}^{\lessgtr
}\ ,A_{k.l.m}\right\}
\end{array}
\right] \nonumber\\
&  +\left[
\begin{array}
[c]{c}%
\epsilon_{kk_{1}k_{2}}\left[  \Sigma_{k_{1},0,0}^{\lessgtr}\ ,A_{k_{2}%
,l,m}\right]  +\epsilon_{kk_{1}k_{2}}\left[  \Sigma_{k_{1},l,0}^{\lessgtr
}\ ,A_{k_{2},0,m}\right] \\
+\epsilon_{kk_{1}k_{2}}\left[  \Sigma_{k_{1},0,m}^{\lessgtr}\ ,A_{k_{2}%
,l,0}\right]  +\epsilon_{kk_{1}k_{2}}\left[  \Sigma_{k_{1},l,m}^{\lessgtr
}\ ,A_{k_{2},0,0}\right] \\
+\epsilon_{ll_{1}l_{2}}\left[  \Sigma_{0,l_{1},0}^{\lessgtr}\ ,A_{k,l_{2}%
,m}\right]  +\epsilon_{ll_{1}l_{2}}\left[  \Sigma_{k,l_{1},0}^{\lessgtr
}\ ,A_{0,l_{2},m}\right] \\
+\epsilon_{ll_{1}l_{2}}\left[  \Sigma_{0,l_{1},m}^{\lessgtr}\ ,A_{k,l_{2}%
,0}\right]  +\epsilon_{ll_{1}l_{2}}\left[  \Sigma_{k,l_{1},m}^{\lessgtr
}\ ,A_{0,l_{2},0}\right] \\
+\epsilon_{m;m_{1},m_{2}}\left[  \Sigma_{0,0.m_{1}}^{\lessgtr}\ ,A_{k,l,m_{2}%
}\right]  +\epsilon_{m;m_{1},m_{2}}\left[  \Sigma_{0,l.m_{1}}^{\lessgtr
}\ ,A_{k,0,m_{2}}\right] \\
+\epsilon_{m;m_{1},m_{2}}\left[  \Sigma_{k,0.m_{1}}^{\lessgtr}\ ,A_{0,l,m_{2}%
}\right]  +\epsilon_{m;m_{1},m_{2}}\left[  \Sigma_{k,l.m_{1}}^{\lessgtr
}\ ,A_{0,0,m_{2}}\right]
\end{array}
\right] \nonumber\\
&  +\left[
\begin{array}
[c]{c}%
i\epsilon_{kk_{1}k_{2}}i\epsilon_{ll_{1}l_{2}}\left(  \left\{  \Sigma
_{k_{1},l_{1},o}^{\lessgtr}\ ,A_{k_{2},l_{2},m}\right\}  +\left\{
\Sigma_{k_{1},l_{1},m}^{\lessgtr}\ ,A_{k_{2},l_{2},o}\right\}  \right) \\
+i\epsilon_{kk_{1}k_{2}}i\epsilon_{mm_{1}m_{2}}\left(  \left\{  \Sigma
_{k_{1},l_{1},o}^{\lessgtr}\ ,A_{k_{2},l_{2},m}\right\}  +\left\{
\Sigma_{k_{1},l_{1},m}^{\lessgtr}\ ,A_{k_{2},l_{2},o}\right\}  \right) \\
+i\epsilon_{ll_{1}l_{2}}i\epsilon_{mm_{1}m_{2}}\left(  \left\{  \Sigma
_{k_{1},l_{1},o}^{\lessgtr}\ ,A_{k_{2},l_{2},m}\right\}  +\left\{
\Sigma_{k_{1},l_{1},m}^{\lessgtr}\ ,A_{k_{2},l_{2},o}\right\}  \right) \\
+i\epsilon_{kk_{1}k_{2}}i\epsilon_{ll_{1}l_{2}}i\epsilon_{mm_{1}m_{2}}\left[
\Sigma_{k_{1},l_{1},m_{1}}^{\lessgtr}\ ,A_{k_{2},l_{2},m_{2}}\right]
\end{array}
\right]  \label{QT8-4}%
\end{align}

\section{\label{multiband}Comparison with Multiband Quantum Transport
Equations}

Here we make a comparison for $N_{s}=2$ with multi-band quantum spin transport
equations of Buot et al \cite{ref14}. Note that Ref. \cite{ref14} essentially
calculates Dirac spin semiconductor Bloch equations (DSSBEs) of the spin
magnetization quantum transport equations (SMQTEs) in the electron-hole
picture. We immediately see that the expression of $G_{k,o}^{<} $, i.e.,
pseudo-spin independent Pauli Dirac spin transport equations, is identical to
the expression of the transport equation of a conduction-band spin transport
of $\vec{S}_{cc}$ of Ref. \cite{ref14}, upon using the correspondence of the
symbols used in Ref. \cite{ref14} to the symbols used in this paper, i.e.,
$\Xi_{\gamma\gamma}\Longrightarrow\Sigma_{\gamma\gamma}$, $\delta_{\alpha
\beta}^{\gtrless}\Longrightarrow\Sigma_{\alpha\beta}^{\gtrless}$, and
$S^{\gtrless}\Longrightarrow G^{\gtrless}$.

By transforming the DSSBEs to the electron picture using Tables 2 - 4, and
using the following relations,%
\begin{align*}
iA\left(  1,2\right)   &  =-2i\mathrm{Im}G^{r}=-\left(  G^{>}\left(
1,2\right)  -G^{<}\left(  1,2\right)  \right)  ,\\
i\Gamma\left(  1,2\right)   &  =-2i\mathrm{Im}\Sigma^{r}=-\left(  \Sigma
^{>}\left(  1,2\right)  -\Sigma^{<}\left(  1,2\right)  \right)
\end{align*}
in combining the equations to obtain the equations for the components of the
magnetization density tensors, the resulting equations exactly reproduce the
equations for $N_{s}=2$ of this paper. The virtue of doing the calculations
from the DSSBEs is that it give us an explicit expression for the sort of
terms containing two Levi-Civita tensors. The calculation shows the
equivalence, for $m=x$-component of pseudospin with $k$ an arbitrary component
of Dirac spin, of the expression,
\begin{align}
&  \left(  \vec{\Sigma}_{y}^{<}\ \times\vec{A}_{z}-\vec{\Sigma}_{z}^{<}%
\times\vec{A}_{y}\right)  +\left(  \vec{A}_{y}\times\vec{\Sigma}_{z}^{<}%
-\vec{A}_{z}\times\vec{\Sigma}_{y}^{<}\right) \nonumber\\
&  =\left\{
\begin{array}
[c]{c}%
\epsilon_{kk_{1}k_{2}}\epsilon_{mm_{1}m_{2}}\vec{\Sigma}_{k_{1,}m_{1}}^{<}%
\vec{A}_{k_{2}m_{2}}\\
+\epsilon_{kk_{1}k_{2}}\epsilon_{mm_{1}m_{2}}\vec{A}_{k_{1,}m_{1}}\vec{\Sigma
}_{k_{2}m_{2}}^{<}%
\end{array}
\right\} \nonumber\\
&  =\epsilon_{kk_{1}k_{2}}\epsilon_{mm_{1}m_{2}}\left\{  \vec{\Sigma}%
_{k_{1,}m_{1}}^{<},\vec{A}_{k_{2}m_{2}}\right\}  \label{eps-eps}%
\end{align}
Note that the $\epsilon_{kk_{1}k_{2}}\epsilon_{mm_{1}m_{2}}\left\{
\vec{\Sigma}_{k_{1,}m_{1}}^{<},\vec{A}_{k_{2}m_{2}}\right\}  $ and similar
other terms are completely symmetric in the simultaneous exchange of $k_{1}%
$and $k_{2}$ and $m_{1}$ and $m_{2}$, respectively. The last line of Eq.
(\ref{eps-eps}) holds for arbitrary components, $k$ and $m$. Note that in this
paper, we have generalized the spinor Hamiltonian, Eq. (\ref{absorbed_in_H}),
to account for the spinor $\mathrm{Re}\Sigma^{r}$ \cite{raja,rajaetal} and/or
magnetic field and/or Dresselhaus and/or Rashba spin-orbit coupling.

\section{Conversion Map from Electron-Hole Picture to Electron Picture}

We have the mapping between electron picture and electron-hole picture using
the following table,\cite{ref10}

\begin{center}
$\overset{\text{Table 2. Mapping from Electron to Electron-Hole Picture}}{%
\begin{tabular}
[c]{|c|c|c|c|}\hline
$\mathbf{electron}\ picture$ & $\left\langle e-\ field\right\rangle $ &
$\left\langle e-h\ \ field\right\rangle $ & $\mathbf{e-h}\ \ picture$\\\hline
$-i\hslash G_{vv}^{<}\left(  12\right)  $ & $\left\langle \psi_{v}^{\dagger
}\left(  2\right)  \psi_{v}\left(  1\right)  \right\rangle $ & $\left\langle
\phi_{v}\left(  2\right)  \phi_{v}^{\dagger}\left(  1\right)  \right\rangle $
& $i\hslash G_{vv}^{h,>T}\left(  12\right)  $\\\hline
$-i\hslash G_{vc}^{<}\left(  12\right)  $ & $\left\langle \psi_{c}^{\dagger
}\left(  2\right)  \psi_{v}\left(  1\right)  \right\rangle $ & $\left\langle
\psi_{c}^{\dagger}\left(  2\right)  \phi_{v}^{\dagger}\left(  1\right)
\right\rangle $ & $-i\hslash g_{ee,vc}^{e-h,<}\left(  12\right)  $\\\hline
$-i\hslash G_{cv}^{<}\left(  12\right)  $ & $\left\langle \psi_{v}^{\dagger
}\left(  2\right)  \psi_{c}\left(  1\right)  \right\rangle $ & $\left\langle
\phi_{v}\left(  2\right)  \psi_{c}\left(  1\right)  \right\rangle $ &
$-i\hslash g_{hh,cv}^{e-h,<}\left(  12\right)  $\\\hline
$-i\hslash G_{cc}^{<}\left(  12\right)  $ & $\left\langle \psi_{c}^{\dagger
}\left(  2\right)  \psi_{c}\left(  1\right)  \right\rangle $ & $\left\langle
\psi_{c}^{\dagger}\left(  2\right)  \psi_{c}\left(  1\right)  \right\rangle $
& $-i\hslash G_{cc}^{e-h,<}\left(  12\right)  $\\\hline
$i\hslash G_{vv}^{>T}\left(  12\right)  $ & $\left\langle \psi_{v}\left(
2\right)  \psi_{v}^{\dagger}\left(  1\right)  \right\rangle $ & $\left\langle
\phi_{v}^{\dagger}\left(  2\right)  \phi_{v}\left(  1\right)  \right\rangle $
& $-i\hslash G_{vv}^{h,<}\left(  12\right)  $\\\hline
$i\hslash G_{vv}^{>}\left(  12\right)  $ & $\left\langle \psi_{v}\left(
1\right)  \psi_{v}^{\dagger}\left(  2\right)  \right\rangle $ & $\left\langle
\phi_{v}^{\dagger}\left(  1\right)  \phi_{v}\left(  2\right)  \right\rangle $
& $-i\hslash G_{vv}^{h,<T}\left(  12\right)  $\\\hline
\end{tabular}
\ }$
\end{center}

We also have the following Tables, which can also be similarly applied to the self-energies,

\begin{center}
$\overset{\text{Table 3.}}{%
\begin{tabular}
[c]{|c|c|}\hline
$\mathbf{electron}\ picture$ & $\mathbf{e-h}\ \ picture$\\\hline
$G_{vv}^{r}\left(  12\right)  $ & $-G_{vv}^{e-h,aT}\left(  12\right)
$\\\hline
$G_{vc}^{r}\left(  12\right)  $ & $g_{ee,vc}^{e-h,r}\left(  12\right)
$\\\hline
$G_{cv}^{r}\left(  12\right)  $ & $g_{hh,cv}^{e-h,r}\left(  12\right)
$\\\hline
$G_{cc}^{r}\left(  12\right)  $ & $G_{cc}^{e-h,r}\left(  12\right)  $\\\hline
\end{tabular}
\ }\overset{\text{Table 4.}}{\ \
\begin{tabular}
[c]{|c|c|}\hline
$\mathbf{electron}\ picture$ & $\mathbf{e-h}\ \ picture$\\\hline
$G_{vv}^{a}\left(  12\right)  $ & $-G_{vv}^{e-h,rT}\left(  12\right)
$\\\hline
$G_{vc}^{a}\left(  12\right)  $ & $g_{ee,vc}^{e-h,a}\left(  12\right)
$\\\hline
$G_{cv}^{a}\left(  12\right)  $ & $g_{hh,cv}^{e-h,a}\left(  12\right)
$\\\hline
$G_{cc}^{a}\left(  12\right)  $ & $G_{cc}^{e-h,a}\left(  12\right)  $\\\hline
\end{tabular}
\ }$
\end{center}

\section{\label{lww}Transport Equations in Phase Space}

The transport equations in phase space are obtained by applying the "lattice"
Weyl transformation (although continuum approximation is often employed in
this paper, this is not essential and we use the word "lattice" when referring
to solid-state problems \cite{bj,jb,ref10,trH}) of the propagator equations
for the respective excitations by using the following identities:

\subsection{"Lattice" Weyl transform of differential operators}%

\begin{equation}
i\hbar\left(  \frac{\partial}{\partial t_{1}}+\frac{\partial}{\partial t_{2}%
}\right)  Y\left(  t_{1},t_{2}\right)  \Leftrightarrow i\hbar\frac{\partial
}{\partial t}Y\left(  E,t\right)  \label{eq8.34a}%
\end{equation}

\begin{equation}
\left(  \frac{\partial^{2}}{\partial t_{1}^{2}}-\frac{\partial^{2}}{\partial
t_{2}^{2}}\right)  Y\left(  t_{1},t_{2}\right)  \Leftrightarrow\frac{2i}%
{\hbar}E\frac{\partial}{\partial t}Y\left(  E,t\right)  \label{eq8.34b}%
\end{equation}%
\begin{equation}
\left(  \nabla_{1}^{2}-\nabla_{2}^{2}\right)  Y\left(  r_{1},r_{2}\right)
\Leftrightarrow\frac{2i}{\hbar}p\cdot\nabla_{q}Y\left(  p,q\right)
\label{eq8.34c}%
\end{equation}

\subsection{Lattice Weyl transform of product of arbitrary operators}

\subsubsection{Poisson bracket expansion}

In terms of differential "Poisson bracket" operator,
\begin{equation}
AB\left(  p,q\right)  =\exp\frac{\hbar}{i}\left(  \frac{\partial^{\left(
a\right)  }}{\partial p}\cdot\frac{\partial^{\left(  b\right)  }}{\partial
q}-\frac{\partial^{\left(  a\right)  }}{\partial q}\cdot\frac{\partial
^{\left(  b\right)  }}{\partial p}\right)  \ a\left(  p,q\right)  \ b\left(
p,q\right)  \label{eq8.35}%
\end{equation}
Thus, we are lead to the lattice Weyl transform of a commutator, $\left[
A,B\right]  \left(  p,q\right)  $, and an anticommutator, $\left\{
A,B\right\}  \left(  p,q\right)  $ as%
\begin{align}
\left[  A,B\right]  \left(  p,q\right)   &  =\cos cos\Lambda\left[  a\left(
p,q\right)  \ b\left(  p,q\right)  -b\left(  p,q\right)  a\left(  p,q\right)
\right] \nonumber\\
&  -i\sin\Lambda\left[  a\left(  p,q\right)  \ b\left(  p,q\right)  +b\left(
p,q\right)  a\left(  p,q\right)  \right]  \label{eq8.38}%
\end{align}%
\begin{align}
\left\{  A,B\right\}  \left(  p,q\right)   &  =\cos cos\Lambda\left[  a\left(
p,q\right)  \ b\left(  p,q\right)  +b\left(  p,q\right)  a\left(  p,q\right)
\right] \nonumber\\
&  -i\sin\Lambda\left[  a\left(  p,q\right)  \ b\left(  p,q\right)  -b\left(
p,q\right)  a\left(  p,q\right)  \right]  \label{eq8.39}%
\end{align}
where $\Lambda=\frac{\hbar}{2}\left(  \frac{\partial^{\left(  a\right)  }%
}{\partial p}\cdot\frac{\partial^{\left(  b\right)  }}{\partial q}%
-\frac{\partial^{\left(  a\right)  }}{\partial q}\cdot\frac{\partial^{\left(
b\right)  }}{\partial p}\right)  $. Note that so far no semi-classical
approximation is envolved in the derivation.

\subsubsection{Integral representation}

In terms of integral operators, we have,%
\begin{align}
AB\left(  p,q\right)   &  =\frac{1}{\left(  \hbar^{4}\right)  ^{2}}\int
dp^{\prime}\ dq^{\prime}\ K_{A}^{+}\left(  p,q;p^{\prime}q^{\prime}\right)
\ b\left(  p^{\prime},q^{\prime}\right) \nonumber\\
&  =\frac{1}{\left(  \hbar^{4}\right)  ^{2}}\int dp^{\prime}\ dq^{\prime
}\ a\left(  p^{\prime},q^{\prime}\right)  \ K_{B}^{-}\left(  p,q;p^{\prime
}q^{\prime}\right)  \ \label{eq8.36}%
\end{align}
where integral kernels are defined by%
\begin{equation}
K_{Y}^{\pm}\left(  p,q;p^{\prime}q^{\prime}\right)  =\int dudv\exp\left\{
\frac{i}{\hbar}\left[  \left(  p-p^{\prime}\right)  \cdot v+\left(
q-q^{\prime}\right)  \cdot u\right]  \right\}  y\left(  p\pm\frac{u}{2}%
,q\mp\frac{v}{2}\right)  \text{.} \label{eq8.37}%
\end{equation}
Therefore, in terms of integral operators,%
\begin{equation}
\left[  A,B\right]  \left(  p,q\right)  =\frac{1}{\left(  \hbar^{4}\right)
^{2}}\int dp^{\prime}\ dq^{\prime}\ \left[  K_{A}^{+}\left(  p,q;p^{\prime
}q^{\prime}\right)  \ b\left(  p^{\prime},q^{\prime}\right)  -b\left(
p^{\prime},q^{\prime}\right)  \ K_{A}^{-}\left(  p,q;p^{\prime}q^{\prime
}\right)  \ \right]  \label{eq8.40}%
\end{equation}%
\begin{equation}
\left\{  A,B\right\}  \left(  p,q\right)  =\frac{1}{\left(  \hbar^{4}\right)
^{2}}\int dp^{\prime}\ dq^{\prime}\ \left[  K_{A}^{+}\left(  p,q;p^{\prime
}q^{\prime}\right)  \ b\left(  p^{\prime},q^{\prime}\right)  +b\left(
p^{\prime},q^{\prime}\right)  \ K_{A}^{-}\left(  p,q;p^{\prime}q^{\prime
}\right)  \ \right]  \label{eq8.41}%
\end{equation}

\subsubsection{Local and nonlocal integral kernels}

The above expressions simplify considerably when the lattice Weyl transform
$a\left(  p,q\right)  $ and $b\left(  p,q\right)  $ are scalar functions. We
have for the integral representations,%

\begin{align}
\left[  A,B\right]  \left(  p,q\right)   &  =\frac{1}{\left(  \hbar
^{4}\right)  ^{2}}\int dp^{\prime}\ dq^{\prime}dudv\exp\left\{  \frac{i}%
{\hbar}\left[  \left(  p-p^{\prime}\right)  \cdot v+\left(  q-q^{\prime
}\right)  \cdot u\right]  \right\} \nonumber\\
&  \times\left[
\begin{array}
[c]{c}%
a\left(  p+\frac{u}{2},q-\frac{v}{2}\right) \\
-a\left(  p-\frac{u}{2},q+\frac{v}{2}\right)
\end{array}
\right]  \ b\left(  p^{\prime},q^{\prime}\right)  \label{nonlocalK}%
\end{align}%
\begin{align}
\left\{  A,B\right\}  \left(  p,q\right)   &  =\frac{1}{\left(  \hbar
^{4}\right)  ^{2}}\int dp^{\prime}\ dq^{\prime}dudv\exp\left\{  \frac{i}%
{\hbar}\left[  \left(  p-p^{\prime}\right)  \cdot v+\left(  q-q^{\prime
}\right)  \cdot u\right]  \right\} \nonumber\\
&  \times\left[
\begin{array}
[c]{c}%
a\left(  p+\frac{u}{2},q-\frac{v}{2}\right) \\
+a\left(  p-\frac{u}{2},q+\frac{v}{2}\right)
\end{array}
\right]  \ b\left(  p^{\prime},q^{\prime}\right)  \label{localK}%
\end{align}
Thus, in its integral forms, which are better suited for carrying numerical
simulations, Eq. (\ref{nonlocalK}) exhibits a \textit{nonlocal} kernel for the
commutator, $\left[  A,B\right]  \left(  p,q\right)  $, whereas
Eq.(\ref{localK}) shows an '\textit{averaging'} type of kernel tied at point
$\left(  p,q\right)  $ or \textit{local} kernel for the anticommutator
$\left\{  A,B\right\}  \left(  p,q\right)  $ to a leading order. Indeed, it is
precisely the nonlocal character of the kernel in Eq. (\ref{nonlocalK}) that
is entirely responsible for quantum-tunneling transport in resonant tunneling
diodes \cite{bj,jb} which exhibit an entirely quantum nonlocality characteristics.

\end{document}